\documentclass[english,a4paper,12pt,twoside,openany]{book}
\usepackage[utf8]{inputenc}
\usepackage[english]{babel}
\usepackage{amsmath}
\usepackage{amsfonts}
\usepackage{amssymb}
\usepackage{amsthm}
\usepackage{cancel}
\usepackage{slashed}
\usepackage{graphics}
\usepackage{graphicx}
\usepackage{fancyhdr}
\usepackage{multirow}
\usepackage{graphicx}
\usepackage{ragged2e}
\usepackage{hyperref}
\usepackage{fancyhdr}
\usepackage{bm}
\usepackage[left=2.75cm,top=2cm,bottom=2.5cm,right=2.25cm,includehead]{geometry}

\usepackage[titletoc]{appendix}
\usepackage{youngtab}
\usepackage{dsfont}
\usepackage{tikz}
\usetikzlibrary{shapes,arrows}
\usepackage{bbm}
\usepackage{mcite}
\usepackage{cite}

\usepackage{verbatim}
\usepackage{indentfirst}

\usepackage{framed}
\usepackage[normalem]{ulem}
\usepackage{enumerate}
\usepackage{cancel}
\usepackage{physics}
\usepackage{captdef}
\usepackage{mathtools}

\theoremstyle{definition}
\newtheorem{theorem}{Theorem}
\newtheorem{corollary}{Corollary}
\newtheorem*{definition}{Definition}
\newtheorem{example}{Example}

\newcommand{\diff}{\mathrm{d}}
\newcommand{\lag}{\mathcal{L}}

\fancypagestyle{plain}{%
	\fancyhf{} 
	\fancyfoot{} 
	
	}

\renewcommand{\chaptermark}[1]%
{\markboth{{\thechapter.\ #1}}{}}
\renewcommand{\sectionmark}[1]%
{\markright{{\thesection.\ #1}}}

\begin{document}
	
\title{The information loss paradox}
\author{Francisco Mart\'inez L\'opez}
\date{\today}

\setcounter{page}{0}
\pagenumbering{roman}

\begin{titlepage}

	\begin{center}
		
		\vspace*{0.1cm}
		
		{\textbf{\Huge The information loss \vspace{3mm} \\  paradox}}
		
		\vspace{2cm}
		
		\includegraphics[scale=1]{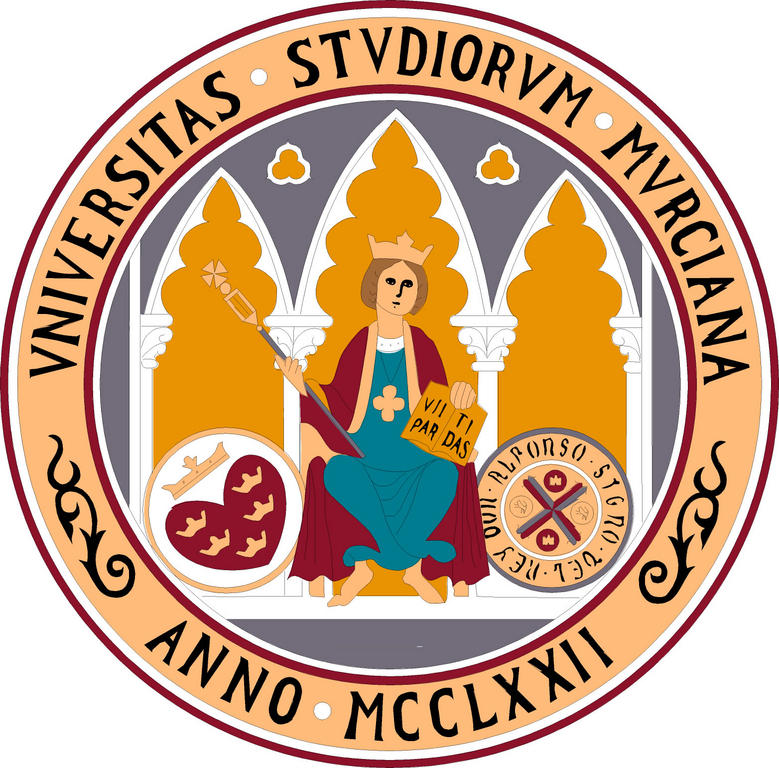}

		\vspace{2cm}
		
		\emph{A dissertation submitted to\\ the University of Murcia for the\\ degree of graduate in Physics}
		
		\vspace{1.5cm}
		
		{\LARGE Francisco Mart\'inez L\'opez}
		
		\vspace{1cm}

		{\Large Supervised by Jose Juan Fern\'andez Melgarejo}		
		
		\
		
		{\Large \hspace{2.5mm} and Emilio Torrente Luj\'an}
		
		\vspace{1cm}
		
		{\large June 2019}
		
		\vspace{1cm}
		
		{\large Department of Physics
			\vspace{2mm}
			\\
			University of Murcia
		}
		
	\end{center}
	
\end{titlepage}

\newpage
\thispagestyle{empty}
\mbox{}
\newpage
\thispagestyle{empty}
\mbox{}

\newpage
\thispagestyle{empty}

\begin{center}
{\Large \bf DECLARACIÓN DE ORIGINALIDAD}	
\end{center}

\

D. FRANCISCO MARTÍNEZ LÓPEZ estudiante del Grado en Física de la Facultad de Química de la Universidad de Murcia, DECLARO:

\

Que el Trabajo de Fin de Grado que presento para su exposición y defensa titulado ``THE INFORMATION LOSS PARADOX'' y cuyos tutores son:

\

D. JOSE JUAN FERNÁNDEZ MELGAREJO,

D. EMILIO TORRENTE LUJÁN,

\

$\hspace{-6mm}$es original y que todas las fuentes utilizadas para su realización han sido debidamente citadas en el mismo.

\begin{flushright}
	Murcia, a 14 de Junio de 2019.
\end{flushright}

Firma:

\begin{flushleft}
		\includegraphics[width=0.4\linewidth]{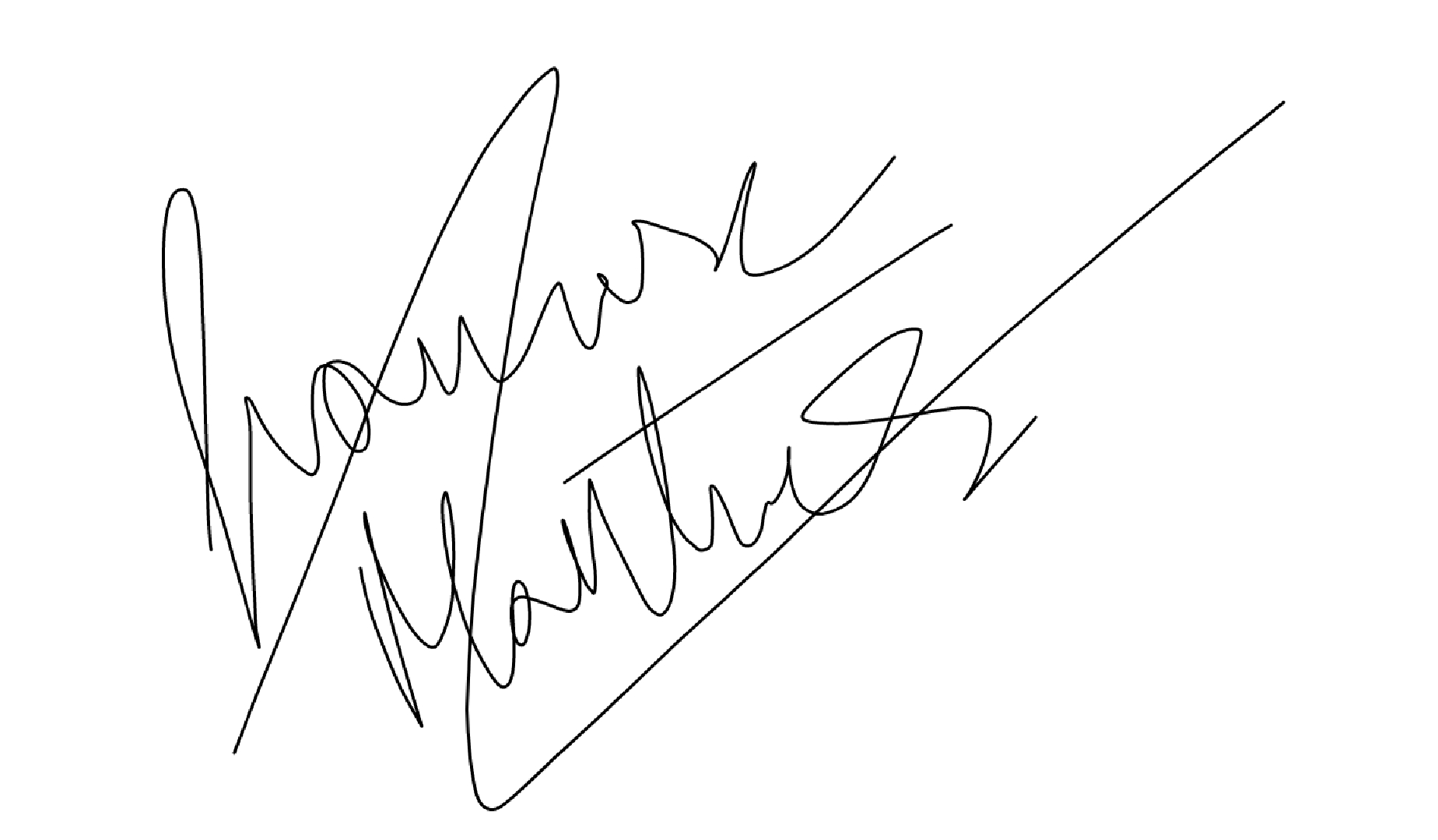}	
\end{flushleft}

\newpage

\begin{center}
	\textbf{\small{Resumen}}
		\justify
		En esta tesis hemos estudiado la paradoja de la pérdida de información en detalle. Como un primer paso, hemos derivado los principales resultados de la teoría cuántica de campos en espaciotiempos curvos. Hemos discutido el caso del campo escalar de Klein-Gordon y concluido con una deducción del llamado efecto Unruh en el espacio de Minkowski. Tras dar un breve listado de definiciones necesarias, como gravedad superficial y factor de ``redshift'', los hemos aplicado junto con los resultados del efecto Unruh para obtener la temperatura de la radiación de Hawking. Después, hemos empleado el formalismo de TCC en espaciotiempos curvos para obtener rigurosamente la distribución de la radiación, considerando el proceso de formación de un agujero negro. A continuación, nos hemos centrado en los estados mecanocuánticos de los cuantos de radiación y la masa en el agujero negro, probando que a primer orden más pequeñas correcciones (condición necesaria para despreciar efectos de gravedad cuántica en la Física usual) la conclusión de Hawking de estados mezcla/remanentes sigue siendo correcta. Finalmente, hemos presentado algunos de los principales resultados de estudios recientes sobre simetrías asintóticas y el grupo de simetrías $\mathrm{BMS}_4$. Hemos concluido presentando algunas ideas que relacionan el efecto de memoria gravitacional con las supertransformaciones y el ``soft hair'' que portan los agujeros negros.
\end{center}

\vspace*{0.5cm}

\begin{center}
	\textbf{\small{Abstract}}
		\justify
		In this thesis, we have studied the information loss paradox in detail. As a first step, we have derived the main results of quantum field theory in a curved background. We have discussed the case of the free scalar Klein-Gordon field and concluded with a derivation of the so-called Unruh effect in Minkowski spacetime. After giving a brief survey of necessary concepts, such as surface gravity and the redshift factor, we have applied them along the results from the Unruh effect to derive the temperature of Hawking radiation. Later, we have used the formalism of QFT in curved spacetime to rigorously obtain the distribution of the radiation, considering a black hole formation process. Thus, we have focused on the quantum mechanical states of the radiation quanta and the mass in the black hole, showing that at first order plus small corrections (condition needed to neglect effects of quantum gravity in normal physics) the Hawking conclusion of mixed states/remnants holds. Finally, we have presented some of the principal results of the recent study of asymptotic symmetries and the $\mathrm{BMS}_4$ symmetry group. We have concluded presenting some ideas relating the gravitational memory effect with the supertransformations and the soft hair carrying the black holes.
\end{center}

\newpage

\tableofcontents

\listoffigures\footnote{All figures appearing in the text have been created by the author, using \textit{Inkscape}. Graphics of functions have been generated using \textit{Maxima} and \textit{gnuplot}.}

\mainmatter
\include{introduction-190601}
\chapter{QFT in curved spacetime}\label{Chapter 2}

General Relativity is a purely classical theory, in its framework all observable quantities have always definite values. But our world is known to be described, on fundamental level, by the principles of quantum mechanics. So on, search for a theory of quantum gravity is one of the hot topics of research in theoretical physics nowadays.

\

The aim of this chapter is to study how free quantum-mechanical matter fields propagate in a fixed curved spacetime background. The underlying reason to study non-interacting fields is that we are interested in the effects of the spacetime itself on the fields. This discussion will lead us to the so-called Unruh effect in flat spacetime. The goal of this is to understand the physical basis of the Hawking radiation (which is the topic of the next chapter).

\

Our discussion here is fundamentally based on \cite{c.qft} and in a more mathematical rigour on \cite{w.qft}. Some ideas were also taken from \cite{t.bh}.

\section{Quantization of the free scalar field}

The minimal-coupling principle gives us a ``simple recipe'' to generalize the laws of physics for curved spacetime. To do so we express our theories, which we know are valid in flat spacetime, in a coordinate-invariant form and then assert that they remain true in curved spacetime. This usually translates into replacing the Minkoswki metric by a generic metric and the partial derivatives by covariant derivatives.

\

The Lagrangian density of a scalar field $\phi$ in curved spacetime is \footnote{In \cite{c.qft} eq. (9.87) a factor $\sqrt{-g}$ is included. We omit it in order to follow the standard fashion.}:
\begin{equation}\label{1.1}
\lag = -\frac{1}{2} g^{\mu \nu} \phi_{;\mu} \phi_{;\nu} - \frac{1}{2} m^2 \phi^2 - \xi R \phi^2,
\end{equation}
where $m$ is the mass, $R$ the Ricci curvature scalar and $\xi$ a constant which parametrized the coupling to the curvature scalar. This expression differs from its flat-spacetime analogue (besides the appearance of the metric $g_{\mu \nu}$ and the covariant derivatives) in the addition of a direct coupling to the Ricci curvature scalar. This coupling is parametrized by a constant $\xi$, which usually takes values $\xi=0$ (minimal coupling) or $\xi=\frac{n-2}{4(n-1)}$ (conformal coupling), where $n$ is the dimension of the spacetime. We can compute the conjugate momentum of the field \footnote{A discrepancy is found between this expression and the one found in \cite{c.qft} eq. (9.90), which gives $\pi=\phi_{;0}$.}:
\begin{equation}\label{1.2}
\pi=\frac{\partial \lag}{\partial \phi_{;0}}=-g^{\mu 0} \phi_{;\mu}.
\end{equation}

The generalization of the Euler-Lagrange equations is straightforward, following the steps of the minimal coupling principle:
\begin{equation}\label{1.3}
\nabla_{\mu} \frac{\partial \lag}{\partial \phi_{;\mu}} - \frac{\partial \lag}{\partial \phi}=0.
\end{equation}

From this, we arrive to the equation of motion of the scalar field:
\begin{equation}\label{1.4}
\Box \phi - m^2 \phi - \xi R \phi=0,
\end{equation}
where the operator of the first term is defined as:
\begin{equation}\label{1.5}
\Box = g^{\mu \nu} \nabla_{\mu} \nabla_{\nu}.
\end{equation}

The solutions to the equation (\ref{1.4}) span a space $\mathbf{S}$ with an inner product defined on a Cauchy surface with induced metric $h_{ab}$ as:
\begin{equation}\label{1.6}
\left(\phi_1,\phi_2\right)=-i \int_{\Sigma} \diff^{n-1} y \ \left(\phi_1 \nabla_{\mu} \phi_2^*-\phi_2^* \nabla_{\mu} \phi_1\right) n^\mu \sqrt{|h|}.
\end{equation}

This product does not depend of the choice of the hypersurface $\Sigma$. Let us consider another Cauchy surface $\Sigma'$ with the inner product defined in the same fashion. Now, if for two arbitrary solutions $\phi_1$ and $\phi_2$, we compute the difference between the inner products defined on the two different hypersurfaces we obtain \footnote{Something must be said about why these two Cauchy surfaces define a closed region.}:
\begin{equation}\label{1.7}
\begin{split}
\left(\phi_1,\phi_2\right)_{\Sigma}-\left(\phi_1,\phi_2\right)_{\Sigma'}&=-i \int_{\Sigma} \diff^{n-1} y \ \left(\phi_1 \nabla_{\mu} \phi_2^*-\phi_2^* \nabla_{\mu} \phi_1\right) n^\mu \sqrt{|h|}\\
&\hphantom{=}+\imath \int_{\Sigma'} \diff^{n-1} y \ \left(\phi_1 \nabla_{\mu} \phi_2^*-\phi_2^* \nabla_{\mu} \phi_1\right) n^\mu \sqrt{|h|}\\
&=\imath \int_{\mathcal{V}} \diff^{n} x \ \nabla^{\mu} \left(\phi_1 \nabla_{\mu} \phi_2^*-\phi_2^* \nabla_{\mu} \phi_1\right)\sqrt{-g}\\
&=\imath \int_{\mathcal{V}} \diff^{n} x \ g^{\mu \nu}\left(\cancel{\nabla_{\nu} \phi_1 \nabla_{\mu} \phi_2^*} + \phi_1 \nabla_{\mu} \nabla_{\nu} \phi_2^*\right.\\
&\hphantom{=}\left.-\cancel{\nabla_{\nu}\phi_2^* \nabla_{\mu} \phi_1}-\phi_2^* \nabla_{\mu} \nabla_{\nu} \phi_1\right)\sqrt{-g}\\
&=\imath \int_{\mathcal{V}} \diff^{n} x \ \left(\phi_1 (m^2 \phi_2^* + \xi R \phi_2^*)-\phi_2^*(m^2 \phi_1 + \xi R \phi_1)\right)\\
&=0.
\end{split}
\end{equation}

In the second equality we have used Stoke's theorem and in the third the Klein-Gordon equation (\ref{1.4}). We can now impose the canonical commutation relations, promoting the fields and its conjugate momentum to linear operators of the Hilbert space of states:
\begin{equation}\label{1.8}
\begin{split}
&\left[\phi(\vec{x},t),\phi(\vec{x'},t)\right]=0,\\
&\left[\pi(\vec{x},t),\pi(\vec{x'},t)\right]=0,\\
&\left[\phi(\vec{x},t),\pi(\vec{x'},t)\right]=\imath \ \delta^{(n-1)}(\vec{x}-\vec{x'}).\\
\end{split}
\end{equation}

Now, if we want to continue working in analogy to flat spacetime we should look for a set of normal modes forming a complete basis of the space $\mathbf{S}$. Since in general there will not be any timelike Killing vector, we can not find solutions which factorize into a time-dependent and a space-dependent factor. That way, we can not classify modes as positive or negative frequency, which is the common procedure in flat spacetime.

\

Anyway, we can always find a set of solutions $\left\{f_i\right\}$ to the equation (\ref{1.4}) that are orthonormal:
\begin{equation}\label{1.9}
(f_i,f_j)=\delta_{ij}.
\end{equation}

The corresponding conjugate modes will obey:
\begin{equation}\label{1.10}
(f_i^*,f_j^*)=-\delta_{ij}.
\end{equation}

Assuming that the index denoting the modes is discrete, these can be used to expand our field as:
\begin{equation}\label{1.11}
\phi=\sum_i \left(\hat{a}_i f_i + \hat{a}_i^{\dagger} f_i^*\right),
\end{equation}
where $\hat{a}_i$ and $\hat{a}^{\dagger}_i$ are suitable operators for the expansion. The commutation relations of the operators $\hat{a}_i$ and $\hat{a}_i^{\dagger}$ are easily obtained if we plug this expansion into the canonical commutation relations that we have introduced previously:
\begin{equation}\label{1.12}
\begin{split}
&\left[\hat{a}_i,\hat{a}_j\right]=0,\\
&\left[\hat{a}_i^{\dagger},\hat{a}_j^{\dagger}\right]=0,\\
&\left[\hat{a}_i,\hat{a}_j^{\dagger}\right]=\delta_{ij}.\\
\end{split}
\end{equation}

As we note, these operators obey the characteristic commutation relations of creation and annihilation operators of the simple harmonic oscillator. The difference is that we have now an infinite number of them. For the harmonic oscillator, we use this operators to build a basis of the Hilbert space, consisting of the set of eigenfunctions of the harmonic oscillator. Now, since we do not have any preferred basis modes, our set of operators will define a vacuum state which depends on our election:
\begin{equation}\label{1.13}
\hat{a}_i \ket{0_f}=0 \hspace{5mm} \mathrm{for \ all} \ i.
\end{equation}

We put the subscript on the vacuum state to keep in mind that it is defined with respect to the modes $f_i$. The entire Fock space can be built from this. Generically, a state with different kind of excitations would be written as:
\begin{equation}\label{1.14}
\ket{n_1,n_2,...,n_i,...,n_j}=\frac{1}{\sqrt{n_1! n_2! ... n_i! ... n_j!}} \left(\hat{a}_1^{\dagger}\right)^{n_1} \left(\hat{a}_2^{\dagger}\right)^{n_2} ... \left(\hat{a}_i^{\dagger}\right)^{n_i} ... \left(\hat{a}_j^{\dagger}\right)^{n_j} \ket{0_f},
\end{equation}
where $n_i$ are the number of excitations of momenta $\vec{k}_i$. Acting on one of those states, the operators change the excitations as expected:
\begin{equation}\label{1.15}
\begin{split}
&\hat{a}_i \ket{n_1,n_2,...,n_i,...,n_j}= \sqrt{n_i} \ket{n_1,n_2,...,n_i-1,...,n_j},\\
&\hat{a}_i^{\dagger} \ket{n_1,n_2,...,n_i,...,n_j}= \sqrt{n_i+1} \ket{n_1,n_2,...,n_i+1,...,n_j}.\\
\end{split}
\end{equation}

For each mode we can also define a number operator $\hat{n}_{fi}$:
\begin{equation}\label{1.16}
\hat{n}_{fi}\equiv\hat{a}_i^{\dagger} \hat{a}_i.
\end{equation}

Those operators will obey the following eigenvalues equation:
\begin{equation}\label{1.17}
\hat{n}_{fi} \ket{n_1,n_2,...,n_i,...,n_j} = n_i \ket{n_1,n_2,...,n_i,...,n_j}.
\end{equation}

\section{Bogoliubov transformations}

Now, consider another set of orthonormal modes $g_i(x^\mu)$ with all the same properties as the original modes $f_i(x^\mu)$. The field operator may be expand in such a new complete basis as:
\begin{equation}\label{1.18}
\phi=\sum_i \left(\hat{b}_i g_i + \hat{b}_i^{\dagger} g_i^*\right).
\end{equation}

By performing this expansion we obtain a new set of annihilation and creation operators, obeying the usual commutation relations:
\begin{equation}\label{1.19}
\begin{split}
&\left[\hat{b}_i,\hat{b}_j\right]=0,\\
&\left[\hat{b}_i^{\dagger},\hat{b}_j^{\dagger}\right]=0,\\
&\left[\hat{b}_i,\hat{b}_j^{\dagger}\right]=\delta_{ij}.\\
\end{split}
\end{equation}

There will be a vacuum state associated with those pairs of operators:
\begin{equation}\label{1.20}
\hat{b}_i \ket{0_g}=0 \hspace{5mm} \mathrm{for \ all} \ i.
\end{equation}

As before, the Fock basis is constructed by repeated application of the creation operator on the vacuum state.

\

In flat spacetime we can choose a natural set of modes demanding they are positive-frequency with respect to the time coordinate. In the transition to curved spacetime we have lost this possibility. If one observer defines particles with respect to the set $f_i$ and another observer uses the set $g_j$, they will generally disagree on how many particles there are.

\

We can expand the different sets of modes in terms of the others:
\begin{equation}\label{1.21}
\begin{split}
&g_i=\sum_j\left(\alpha_{ij} f_j + \beta_{ij} f_j^*\right),\\
&f_i=\sum_j\left(\alpha_{ji}^* g_j - \beta_{ji} g_j^*\right),\\
\end{split}
\end{equation}
where $\alpha_{ij}$ and $\beta_{ij}$ are the corresponding matrix coefficients of the expansion. This transformation between the two sets of modes is known as the Bogoliubov transformation. The Bogoliubov coefficients are expressed as:
\begin{equation}\label{1.22}
\begin{split}
&\alpha_{ij}=(g_i,f_j),\\
&\beta_{ij}=-(g_i,f_j^*).\\
\end{split}
\end{equation}

These relations can be easily derived from equation (\ref{1.21}) using the orthonormality relations of the modes. The coefficients satisfy normalization conditions:
\begin{equation}\label{1.23}
\begin{split}
&\sum_k \left(\alpha_{ik}\alpha_{jk}^* - \beta_{ik}\beta_{jk}^*\right)=\delta_{ij},\\
&\sum_k \left(\alpha_{ik}\beta_{jk} - \beta_{ik}\alpha_{jk}\right)=0.\\
\end{split}
\end{equation}

We can relate the different sets of creation and annihilation operators using the expansion (\ref{1.18}) and the expressions (\ref{1.22}). Making so, we obtain:
\begin{equation}\label{1.24}
\begin{split}
\phi&=\sum_j \left(\hat{b}_j g_j + \hat{b}_j^{\dagger} g_j^*\right)\\
&=\sum_j \left(\hat{b}_j \sum_i\left(\alpha_{ji} f_i + \beta_{ji} f_i^*\right) + \hat{b}_i^{\dagger} \sum_i\left(\alpha_{ji}^* f_i^* + \beta_{ji}^* f_i\right)\right)\\
&=\sum_{i,j} \left[\left(\alpha_{ji} \hat{b}_j + \beta_{ji}^* \hat{b}_j^{\dagger}\right) f_i + \left(\beta_{ji} \hat{b}_j + \alpha_{ji}^* \hat{b}_j^{\dagger}\right) f_i^*\right]\\
&=\sum_{i} \left(\hat{a}_i f_i + \hat{a}_i^{\dagger} f_i^*\right).\\
\end{split}
\end{equation}

This way, we can identify terms and express the transformations of the operators in terms of the Bogoliubov coefficients:
\begin{equation}\label{1.25}
\begin{split}
\hat{a}_i=\sum_j \left(\alpha_{ji} \hat{b}_j + \beta_{ji}^* \hat{b}_j^{\dagger}\right),\\
\hat{b}_i=\sum_j \left(\alpha_{ij}^* \hat{a}_j - \beta_{ij}^* \hat{a}_j^{\dagger}\right).\\
\end{split}
\end{equation}

At this point, we can look out the dependence on our election of modes in the number of particles we observe. If we choose the vacuum state $\ket{0_f}$, in which there are no particles in virtue of (\ref{1.17}), we can calculate how many particles does an observer measure if she use the mode $g_i$. To do so, we calculate the expectation value of the number operator $\hat{n}_{gi}$:
\begin{equation}\label{1.26}
\begin{split}
\expval{\hat{n}_{gi}}{0_f}&=\expval{\hat{b}_i^{\dagger}\hat{b}_i}{0_f}\\
&=\ev**{\sum_{j,k} \left(\alpha_{ij} \hat{a}_j^{\dagger} - \beta_{ij} \hat{a}_j\right) \left(\alpha_{ik}^* \hat{a}_k - \beta_{ik}^* \hat{a}_k^{\dagger}\right)}{0_f}\\
&=\sum_{j,k} \beta_{ij} \beta_{ik}^* \ev{\hat{a}_j\hat{a}_k^{\dagger}}{0_f}\\
&=\sum_{j,k} \beta_{ij} \beta_{ik}^* \ev{\hat{a}_k^{\dagger}\hat{a}_j+\delta_{jk}}{0_f}\\
&=\sum_{j,k} \beta_{ij} \beta_{ik}^* \delta_{jk} \braket{0_f}\\
&=\sum_{j} \beta_{ij} \beta_{ij}^*=\sum_{j} |\beta_{ij}|^2.\\
\end{split}
\end{equation}

In this calculation we have taken into account the commutation relations of the annihilation and creation operators and how those act on the vacuum state. This quantity is in general non-zero. Thus, what it seems to be empty space for an observer, it could be full of particles for another.

\

Let us consider now an experimental set up. Firstly, we need to specify what is the definition of particle that uses a detector travelling in a curved spacetime. A detector will measure proper time along its trajectory, and with respect to it, we define a set of modes of positive ($\omega>0$) or negative ($\omega<0$) frequency:
\begin{equation}\label{1.27}
\frac{D}{\diff \tau} \ f_i = -i \omega f_i,
\end{equation}
where $D/\diff \tau$ is covariant differentiation with respect proper time $\tau$. If the spacetime is static, we have a timelike Killing vector $K^\mu$ \footnote{Requiring the spacetime to be static is a sufficient but non-necessary condition to have this Killing vector field. If it is stationary we find the same.}. Then we can choose coordinates a set of in which time-space cross terms in the metric cancel. In this case we can find separable solutions to (\ref{1.4}) of the form:
\begin{equation}\label{1.28}
f_{\omega}(t,\vec{x})=\mathrm{e}^{-i \omega t} \bar{f}(\vec{x}).
\end{equation}

Those modes may be described as positive frequency in a coordinate-invariant form as:
\begin{equation}\label{1.29}
\pounds_{K} f_{\omega}=-i\omega f_{\omega}.
\end{equation}

Now, the modes $f_i$ may be a natural basis for describing the Fock space of the detector if it follows an orbit of the Killing field, i.e. the four-velocity $u^\mu$ is proportional to $K^\mu$, and so proper time is proportional to $t$.

\subsection{Particle creation in non-stationary spacetimes}

Consider a spacetime which is static in the asymptotic past and future. Let us assume that in-between, we have a disturbance, i.e., for some time interval we have a time-dependent metric. Then, we can express a solution of the Klein-Gordon equation before the perturbation in terms of some normal modes with its annihilation and creation operators. This modes are solutions of the Klein-Gordon equation only before the disturbance. After that takes place, the field may be expressed in terms of a different set of modes with its corresponding Bogoliubov-transformed operators.

\begin{figure}[t]
	\centering
	\includegraphics[width=0.4\linewidth]{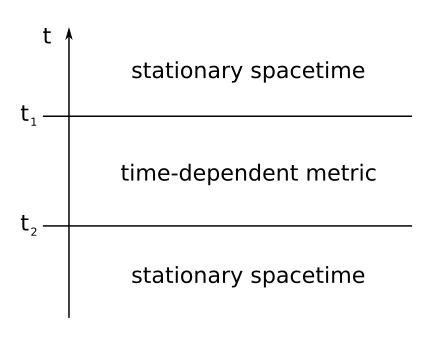}
	\caption{\textit{Diagram of a specetime which is stationary in the asymptotic past and future, but has a time-dependent part in-between.}}
	\label{Figure 1.1}
\end{figure}

The Bogoliubov transformation expresses the operators used in the asymptotic future in terms of the ones of the asymptotic past. As we have seen before, this leads to a possible particle detection when we are using the second set of modes even when at the beginning we are in the vacuum state of the original modes. This way, the disturbance has produced particles which did not exist earlier.

\section{The Unruh effect}

Following the ideas that we have introduced, as a previous step in order to reach the physical understanding of the Hawking radiation, we are going to study a phenomenon that occurs in flat spacetime, the Unruh effect \cite{unruh}. It consist of the discrepancy between the number of particles observed by inertial and accelerated observers in Minkowski spacetime, as a consequence of its different notions of positive-frequency modes.

\

The trajectory of an accelerated observer in Minkowski space will follow the orbit of a timelike Killing vector. As we have seen, we can expand a field in terms of an adequate set of modes in that case (\ref{1.28}). When we compare the vacuum state defined in the Minkowski space to the expectation value of the number operator for the accelerated (Rindler) observer we will obtain a thermal spectrum of particles.

\

To make the derivation as simple as possible, we consider the case of a massless scalar field in two dimensions. An observer in (rectilinear) motion with constant proper acceleration $\alpha$ would follow a trajectory \cite{rindler}:
\begin{equation}\label{1.30}
\begin{split}
&x^0(\tau)=t(\tau)=\frac{1}{\alpha} \ \sinh{(\alpha \tau)},\\
&x^1(\tau)=x(\tau)=\frac{1}{\alpha} \ \cosh{(\alpha \tau)}.\\
\end{split}
\end{equation}

The previous equations define a hyperbolic motion asymptoting to null paths in the past and the future:
\begin{equation}\label{1.31}
x^2(\tau)+t^2(\tau)=\alpha^2.
\end{equation}

Adapted to this kind of motion, we can define a new set of coordinates $(\eta,\xi)$ in Minkowski spacetime:
\begin{equation}\label{1.32}
\begin{split}
&t=\frac{\mathrm{e}^{a \xi}}{a} \ \sinh{(a \eta)},\\
&x=\frac{\mathrm{e}^{a \xi}}{a} \ \cosh{(a \eta)}.\\
\end{split}
\end{equation}

These coordinates cover only the wedge $x>|t|$, known as region I, which is the region accessible to an observer with constant proper acceleration in the direction of positive $x$. If we flip signs in the previous equations the coordinates will cover $x<|t|$, labelled as region IV. Let us note that we can not use the coordinates simultaneously in both regions. However, if we indicate explicitly in which region we are working there will be no problem.

\begin{figure}[t]
	\centering
	\includegraphics[width=0.6\linewidth]{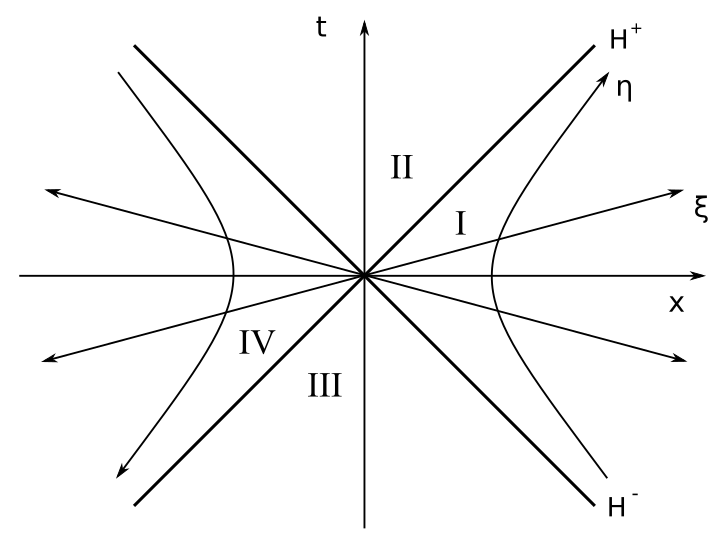}
	\caption{\textit{Diagram of Minkowski spacetime, using Rindler coordinates $(\eta,\xi)$.}}
	\label{Figure 1.2}
\end{figure}

In these coordinate system, the constant proper acceleration trajectory (\ref{1.30}) takes the form (we just need to equate equations (\ref{1.30}) and (\ref{1.32}), which is straightforward to solve):
\begin{equation}\label{1.33}
\begin{split}
&\eta(\tau)=\frac{\alpha}{a} \ \tau,\\
&\xi(\tau)=\frac{1}{a} \ \ln(\frac{a}{\alpha}).\\
\end{split}
\end{equation}

The line element of Minkowski spacetime in Minkowskian coordinates is:
\begin{equation}\label{1.34}
\diff s^2 = - \diff t^2 + \diff x^2.
\end{equation}

In the new coordinates, it is given by:
\begin{equation}\label{1.35}
\diff s^2 = \mathrm{e}^{2 a \xi} \left(- \diff \eta^2 + \diff \xi^2\right).
\end{equation}

The metric is independent of $\eta$, and because of this it is direct that $\partial_\eta$ is a Killing vector. Thus, this vector field may be used to define positive-frequency modes to build a proper Fock basis of the Hilbert space, as discussed in the previous section. The massless Klein-Gordon equation, in this new set of coordinates (known as Rindler coordinates), is:
\begin{equation}\label{1.36}
\Box \phi= \mathrm{e}^{2 a \xi} \left(-\partial_\eta^2 +\partial_\xi^2\right) \phi=0.
\end{equation}

The solution modes in region I must have positive frequency with respect to the future-directed Killing vector $\partial_\eta$. Since this vector is past-directed in region IV, in that portion of the Minkowski spacetime we are going to build positive-frequency modes with respect to $\partial_{-\eta}=-\partial_\eta$, which is future-directed in IV. Properly normalized plane waves solve equation (\ref{1.36}), so the two sets of modes we must introduce take the form:
\begin{equation}\label{1.37}
\begin{split}
&g_k^{(1)}=\left\{\begin{array}{c}\frac{\mathrm{e}^{\imath(k\xi-\omega\eta)}}{\sqrt{4\pi\omega}}\\ 0\end{array}\right. \left.\begin{array}{c} \mathrm{I} \\ \mathrm{IV}\end{array}\right.,\\
&g_k^{(2)}=\left\{\begin{array}{c}0\\ \frac{\mathrm{e}^{\imath(k\xi-\omega\eta)}}{\sqrt{4\pi\omega}}\end{array}\right. \left.\begin{array}{c} \mathrm{I} \\ \mathrm{IV}\end{array}\right.,\\
\end{split}
\end{equation}
where $\omega=|k|$. With respect to the suitable future-directed Killing vector, the modes are positive-frequency:
\begin{equation}\label{1.38}
\begin{split}
&\partial_\eta g_k^{(1)}=-\imath \omega g_k^{(1)},\\
&\partial_{-\eta} g_k^{(2)}=-\imath \omega g_k^{(1)}.\\
\end{split}
\end{equation}

The two sets of modes, together with its complex conjugates, form a basis of the space of solutions of the Klein-Gordon equation through the entire spacetime. Assuming the index of the modes to be continuous, we can expand any solution $\phi$ of (\ref{1.36}) in the form:
\begin{equation}\label{4.39}
\phi=\int \diff k \left(\hat{b}_k^{(1)} g_k^{(1)}+\hat{b}_k^{(1)\dagger}g_k^{(1)*}+\hat{b}_k^{(2)} g_k^{(2)}+\hat{b}_k^{(2)\dagger}g_k^{(2)*}\right).
\end{equation}

Alternatively, we can expand the field in terms of a set of usual Minkowski modes:
\begin{equation}\label{4.40}
\phi=\int \diff k \left(\hat{a}_{k} f_{k}+\hat{a}_{k}^\dagger f_{k}^*\right).
\end{equation}

It is straightforward to calculate the inner products of the modes (\ref{1.37}) and check that gives the same result as using the ordinary modes. In this case, as we have seen previously, the Hilbert space is the same but the Fock spaces that generate our distinct creation and annihilation operators will be different. In order to probe this, we should evaluate the average of the Rindler number operator in the Minkowski vacuum state.

\

Since that would be a difficult task, we look for a more direct option. We may find a set of modes which overlaps easily with the Rindler modes and share at least the same vacuum state as the Minkowski modes. To do so, we can start with Rindler modes and extend them analytically through the entire spacetime in terms of the Minkowski coordinates. We may write for the exponentials in (\ref{1.32}) and its analogous for region IV the following:
\begin{equation}\label{1.41}
\mathrm{e}^{\mp a(\eta \mp \xi)}=\left\{\begin{array}{c}a(\mp t+x)\\a(\pm t -x)\end{array}\right.\begin{array}{c}\mathrm{I}\\\mathrm{IV}\end{array}.
\end{equation}

So, the Rindler modes (those with $k>0$, so $k=\omega$) may be written as:
\begin{equation}\label{1.42}
\begin{split}
&g_k^{(1)}=\frac{\left[a(-t+x)\right]^{\imath \omega/a}}{\sqrt{4 \pi \omega}},\\
&g_k^{(2)}=\frac{\left[a(-t-x)\right]^{-\imath \omega/a}}{\sqrt{4 \pi \omega}}.\\
\end{split}
\end{equation}

The analytical extension is performed by using this expression for any $(t,x)$. We see that at $t=0$ our modes do not coincide. We can reverse the wave number of $g_k^{(2)}$ and then take the complex conjugate:
\begin{equation}\label{1.43}
\begin{split}
g_{-k}^{(2)*}&=\frac{\left[a(t-x)\right]^{\imath \omega/a}}{\sqrt{4 \pi \omega}}\\
&=\frac{\left[a \ \mathrm{e}^{-\imath \pi}(-t+x)\right]^{\imath \omega/a}}{\sqrt{4 \pi \omega}}.\\
\end{split}
\end{equation}

The same procedure may be done for the other set of modes:
\begin{equation}\label{1.44}
\begin{split}
g_{-k}^{(1)*}&=\frac{\left[a(t+x)\right]^{-\imath \omega/a}}{\sqrt{4 \pi \omega}}\\
&=\frac{\left[a \ \mathrm{e}^{\imath \pi}(-t-x)\right]^{-\imath \omega/a}}{\sqrt{4 \pi \omega}}.\\
\end{split}
\end{equation}

We can build some linear combinations which are well defined over the surface $t=0$:
\begin{equation}\label{1.45}
\begin{split}
&h_k^{(1)}=\mathrm{e}^{\pi\omega/2a} \ g_k^{(1)} + \mathrm{e}^{-\pi\omega/2a} \ g_{-k}^{(2)*},\\
&h_k^{(2)}=\mathrm{e}^{\pi\omega/2a} \ g_k^{(2)} + \mathrm{e}^{-\pi\omega/2a} \ g_{-k}^{(1)*}.\\
\end{split}
\end{equation}

We can obtain the normalization by computing the inner product, considering that the Rindler modes are orthonormal.
\begin{equation}\label{1.46}
\begin{split}
\left(h_{k_1}^{(1)},h_{k_2}^{(1)}\right)&=\mathrm{e}^{\pi/2a \left(\omega_1+\omega_2\right)} \left(g_{k_1}^{(1)},g_{k_2}^{(1)}\right)+\mathrm{e}^{-\pi/2a \left(\omega_1+\omega_2\right)} \left(g_{-k_1}^{(2)*},g_{-k_2}^{(2)*}\right)\\
&=\mathrm{e}^{\pi/2a \left(\omega_1+\omega_2\right)} \ \delta \left(k_1-k_2\right)+\mathrm{e}^{-\pi/2a \left(\omega_1+\omega_2\right)} \ \delta \left(-k_1+k_2\right)\\
&=\left(\mathrm{e}^{\pi\omega_1/a}-\mathrm{e}^{-\pi\omega_1/a}\right) \delta \left(k_1-k_2\right)\\
&=2 \sinh(\frac{\pi\omega_1}{a}) \ \delta \left(k_1-k_2\right).\\
\end{split}
\end{equation}

We can redefine the modes in order to orthonormalize them:
\begin{equation}\label{1.47}
\begin{split}
&h_k^{(1)}=\frac{\mathrm{e}^{\pi\omega/2a} \ g_k^{(1)} + \mathrm{e}^{-\pi\omega/2a} \ g_{-k}^{(2)*}}{\sqrt{2 \sinh(\frac{\pi\omega}{a})}},\\
&h_k^{(2)}=\frac{\mathrm{e}^{\pi\omega/2a} \ g_k^{(2)} + \mathrm{e}^{-\pi\omega/2a} \ g_{-k}^{(1)*}}{\sqrt{2 \sinh(\frac{\pi\omega}{a})}}.\\
\end{split}
\end{equation}

A Rindler mode, for example, $g_{k}^{(1)}$, at $t=0$ is only valid in the positive number semi-line $\mathbb{R}^+$. As a consequence, it can not be expressed in terms of positive frequency plane waves. Then Rindler's annihilation operator is a superposition of the Minkowski's creation and annihilation operators. Unlike this, the new modes may be expanded as purely positive-frequency plane waves since they are analytic in the same portion of the complex plane of its arguments as the Minkowski modes. Thus, their annihilation operators may produce the same vacuum state:
\begin{equation}\label{1.48}
\hat{a}_k \ket{0_M}=\hat{c}_k^{(1)} \ket{0_M}=\hat{c}_k^{(2)} \ket{0_M}=0,
\end{equation}
where $\ket{0_M}$ is the vacuum state associated with the Minkowski modes. The excitations may not coincide, but as we are only interested in the vacuum state this will work. As discussed in the previous section, the Bogoliubov transformation between two set of modes provide an expression of the ladder operators associated to one set of modes in terms of the other ones. Since we know the transformation between $h_{k}^{(1,2)}$ and $g_{k}^{(1,2)}$, we have for the operators:
\begin{equation}\label{1.49}
\begin{split}
&\hat{b}_k^{(1)}=\frac{\mathrm{e}^{\pi\omega/2a} \ \hat{c}_k^{(1)} + \mathrm{e}^{-\pi\omega/2a} \ \hat{c}_{-k}^{(2)\dagger}}{\sqrt{2 \sinh(\frac{\pi\omega}{a})}},\\
&\hat{b}_k^{(2)}=\frac{\mathrm{e}^{\pi\omega/2a} \ \hat{c}_k^{(2)} + \mathrm{e}^{-\pi\omega/2a} \ \hat{c}_{-k}^{(1)\dagger}}{\sqrt{2 \sinh(\frac{\pi\omega}{a})}}.\\
\end{split}
\end{equation}

Hence, it is possible to express the Rindler number operators in terms of these $\hat{c}_k^{(1,2)}$. For an observer moving in region I we can calculate the expected number of particles in Minkowski vacuum. As $\hat{c}_k^{(1)\dagger} \ket{0_M}$ is a one-particle state, it will only survive the term with $\hat{c}_{-k}^{(2)}\hat{c}_{-k}^{(2)\dagger}$\footnote{In \cite{c.qft} equation (9.163) the term that survives is the one with $\hat{c}_{-k}^{(1)}\hat{c}_{-k}^{(1)\dagger}$ but if you perform the calculation that one does not appear at all.}, which is a part of the product $\hat{n}_{gk}^{(1)}=\hat{b}_{k}^{(1)\dagger}\hat{b}_{k}^{(1)}$. Explicitly, one can write:
\begin{equation}\label{1.50}
\begin{split}
\expval{\hat{n}_{gk}^{(1)}}{0_M}&=\expval{\hat{b}_{k}^{(1)\dagger}\hat{b}_{k}^{(1)}}{0_M}\\
&=\frac{\expval{\mathrm{e}^{-\pi\omega/a} \ \hat{c}_{-k}^{(2)}\hat{c}_{-k}^{(2)\dagger}}{0_M}}{2 \sinh(\frac{\pi\omega}{a})}\\
&=\frac{\mathrm{e}^{-\pi\omega/a} \ \delta(0)}{2 \sinh(\frac{\pi\omega}{a})}\\
&=\frac{\mathrm{e}^{-\pi\omega/a} \ \delta(0)}{\mathrm{e}^{\pi\omega/a}-\mathrm{e}^{-\pi\omega/a}}\\
&=\frac{\delta(0)}{\mathrm{e}^{2\pi\omega/a}-1}.\\
\end{split}
\end{equation}

We arrive to a result that looks like a Planck spectrum\footnote{The delta is a consequence of using plane waves: if the set of modes were defined as wave packets, we would arrive to a finite solution.}, with a temperature:
\begin{equation}\label{1.51}
T_U=\frac{a}{2\pi}.
\end{equation}

Reintroducing constants using dimensional analysis, one gets:
\begin{equation}\label{1.52}
T_U=\frac{\hbar a}{2\pi k_B c}.
\end{equation}

We see that if we take the limit $\hbar\rightarrow 0$ this temperature vanishes, which points out its quantum-mechanical origin. This is the so-called Unruh effect: an observer moving with uniform acceleration (a Rindler observer) will observe a thermal spectrum of particles even in Minkowski vacuum. This reveals the thermal nature of vacuum in field theory.

\

In this chapter we have obtain a relation between the non-stationary nature of a spacetime and the creation of particles in it, and how to quantify it in terms of the Bogoliubov coefficients. We have applied that to the case of accelerated observers in Minkoski space, obtaining the Unruh effect. In the next chapter we will apply again that, but now to the formation of a black hole. Since this conforms a non-stationary spacetime, we expect again a particle production.
\chapter{Hawking radiation}\label{Chapter 3}

In this chapter, we are going to study in detail the relation between the non-stationary nature of a black hole spacetime and the creation of particles. The discovery of the thermal spectrum of emission of black holes sets new standards in the way we think about them. As they emit particles, they must have a non-zero temperature and therefore they are in an equilibrium state. This considerations establish a connection with the thermodynamics of black holes. Black holes form from gravitational collapse and by this mechanism they would evaporate and disappear in a finite amount of time.

\

We will start by reviewing some concepts that will appear in our discussion, such as Killing horizon, surface gravity and redshift factor. Later, we will explore a derivation of the Hawking effect based on the Unruh effect \footnote{We will mainly follow the discussions in \cite{c.qft} and \cite{j.bh}.}.

\

Finally we will discuss the Hawking effect in a simple scenario of black hole formation, following the ideas of particle creation in non-stationary spacetimes, which were succinctly mentioned previously. Lastly, we explore the eventual evaporation of a black hole as a consequence of particle emission. We can find further insight into these topics in \cite{f.mbhe} and \cite{v.bms} respectively.

\section{Surface gravity, redshift, acceleration}

A \textbf{Killing horizon} is defined as a null hypersurface along which certain Killing vector field is null. This is a concept that is independent from the one of event horizon, but sometimes closely related. By definition, the Killing vector turns out to be normal to the associate Killing horizon.

\

In Minkowski spacetime, for example, there are no event horizons at all. However, consider the Killing vector field which generates a boost in the $x$-direction:
\begin{equation}\label{2.1}
\chi=x\partial_t+t\partial_x.
\end{equation}

The norm of this vector is easily computed:
\begin{equation}\label{2.2}
\chi_\mu \chi^\mu=-x^2+t^2.
\end{equation}

This expression is null at the null surfaces:
\begin{equation}\label{2.3}
x=\pm t.
\end{equation}

Thus, those surfaces are Killing horizons. In general, since every Killing vector is normal to a Killing horizon, it must obey a geodesic equation along this:
\begin{equation}\label{2.4}
\xi^\mu \nabla_\mu \xi^\nu=\xi^{\mu} \xi^{\nu}_{\ ;\mu}=-\kappa \xi^\nu.
\end{equation}

As the curves are not affine-parametrized, in general the right-hand side is nonzero. The parameter $\kappa$ is known as \textbf{surface gravity}, which is constant over the horizon (aside from some exceptions). Invoking the Killing equation $\pounds_\xi g_{\mu\nu}=0$ and the fact that the Killing vector field is normal to the null surface, $\xi_{\left[\mu\right.} \nabla_\nu \xi_{\left.\sigma\right]}=0$\footnote{This is a consequence of Frobenius's theorem, as it is indicated in Appendix B of \cite{w.qft}.}, one can obtain a simpler expression for the surface gravity. If we expand the relation for hypersurface orthogonality and plug the Killing equation, we obtain:
\begin{equation}\label{2.5}
\begin{split}
\xi_{\left[\mu\right.} \xi_{\left.\sigma;\nu\right]}&=\frac{1}{6} \left(\xi_\mu \xi_{\sigma;\nu}-\xi_\nu \xi_{\sigma;\mu}+\xi_\nu \xi_{\mu;\sigma}-\xi_\sigma \xi_{\mu;\nu}+\xi_\sigma \xi_{\nu;\mu}-\xi_\mu \xi_{\nu;\sigma}\right)\\
&=\frac{1}{6} \left(\xi_\mu \xi_{\sigma;\nu}+\xi_\nu \xi_{\mu;\sigma}+\xi_\nu \xi_{\mu;\sigma}+\xi_\sigma \xi_{\nu;\mu}+\xi_\sigma \xi_{\nu;\mu}+\xi_\mu \xi_{\sigma;\nu}\right)\\
&=\frac{1}{3} \left(\xi_\mu \xi_{\sigma;\nu}+\xi_\nu \xi_{\mu;\sigma}+\xi_\sigma \xi_{\nu;\mu}\right)=0.\\
\end{split}
\end{equation}

Now we can contract that expression with $\xi^{\sigma}\xi^{\nu;\mu}$ and use the geodesic equation (\ref{2.4}) to get:
\begin{equation}\label{2.6}
\begin{split}
-\xi_\sigma \xi^{\sigma} \ \xi^{\nu;\mu}\xi_{\nu;\mu}&=\xi_\mu\xi^{\nu;\mu} \ \xi^{\sigma} \xi_{\sigma;\nu}+ \xi_\nu\xi^{\nu;\mu} \ \xi^{\sigma}\xi_{\mu;\sigma}\\
&=-\xi_\mu\xi^{\nu;\mu} \ \xi^{\sigma} \xi_{\nu;\sigma}- \xi_\nu\xi^{\mu;\nu} \ \xi^{\sigma}\xi_{\mu;\sigma}\\
&=-\xi_\mu\xi^{\ ;\mu}_{\nu} \ \xi^{\sigma} \xi_{\ ;\sigma}^{\nu}- \xi_\nu\xi^{\ ;\nu}_\mu \ \xi^{\sigma}\xi_{\ ;\sigma}^{\mu}\\
&=\kappa^2 \xi_\nu \xi^\nu + \kappa^2 \xi_\mu \xi^\mu.\\
\end{split}
\end{equation}

Finally we end up with the desired equation:
\begin{equation}\label{2.7}
\kappa^2=-\frac{1}{2} \xi_{\nu;\mu} \xi^{\nu;\mu}=-\frac{1}{2} \nabla_\mu \xi_\nu \nabla^\mu \xi^\nu.
\end{equation}

In a static, asymptotically flat spacetime we can interpret the surface gravity as the acceleration of a fixed observer near the horizon, as seen by an observer at infinity. In this case we can normalize the time translation Killing vector $K=\partial_t$ as $K_\mu K^\mu\stackrel{r \rightarrow \infty}{=} -1$ in order to obtain a unique value for the surface gravity.

\

For a fixed (or static) observer, its velocity is proportional to the Killing field $K^\mu$:
\begin{equation}\label{2.8}
K^\mu=V \ U^\mu.
\end{equation}

But as the velocity is normalized to $U_\mu U^\mu=-1$, this proportionality function results:
\begin{equation}\label{2.9}
\begin{split}
K_\mu K^\mu&=K_\mu V \ U^\mu\\
&=V^2 \ U_\mu U^\mu\\
&=-V^2.
\end{split}
\end{equation}

Solving for $V$ we get:
\begin{equation}\label{2.10}
V=\sqrt{-K_\mu K^\mu}.
\end{equation}

This quantity $V$ is called the \textbf{redshift factor}. It is a measure of the frequency shift of a radiation emitted and observed by static observers. In terms of its wavelengths, the relation takes the form:
\begin{equation}\label{2.11}
\lambda_2=\frac{V_2}{V_1} \ \lambda_1.
\end{equation}

At infinity $V=1$ and $\lambda_\infty=\lambda/V$. Now, we can relate the surface gravity with the four-acceleration $a^\mu$. It is given by:
\begin{equation}\label{2.12}
a^\mu=\frac{\mathrm{D} U^\mu}{\diff \lambda}=U^\nu \nabla_\nu U^\mu.
\end{equation}

This may be re-expressed using the redshift factor as:
\begin{equation}\label{2.13}
a_\mu=\nabla_\mu \ln(V).
\end{equation}

Its magnitude is straightforward to compute:
\begin{equation}\label{2.14}
a=\sqrt{a_\mu a^\mu}=\frac{\sqrt{\nabla_\mu V \nabla^\mu V}}{V}.
\end{equation}

This quantity diverges at the horizon, since the Killing vector is null along it. Namely, what an observer at infinity would measure is that the acceleration of a fixed observer is redshifted, so we recover the surface gravity:
\begin{equation}\label{2.15}
\begin{split}
\kappa^2&=(a \ V)^2= \nabla_\mu V \nabla^\mu V\\
&=\nabla_\mu \sqrt{-K_\nu K^\nu} \ \nabla^\mu \sqrt{-K_\sigma K^\sigma}\\
&=-\frac{1}{4 K_\nu K^\nu} \nabla_\mu (K_\nu K^\nu) \ \nabla^\mu (K_\sigma K^\sigma)\\
&=-\frac{1}{4 K_\nu K^\nu}  (K_\nu \nabla_\mu  K^\nu + K^\nu \nabla_\mu  K_\nu) \ (K_\sigma \nabla^\mu  K^\sigma + K^\sigma \nabla^\mu  K_\sigma)\\
&=-\frac{1}{4 K_\nu K^\nu} \ 2 K^\nu \nabla_\mu K_\nu \ K_\nu \nabla^\mu  K^\nu\\
&=-\frac{1}{2}\ \nabla_\mu K_\nu \nabla^\mu  K^\nu.
\end{split}
\end{equation}

\section{The Hawking effect}

Consider a non-rotating fixed observer near a Schwarzschild black hole. The Schwarzschild metric is given by: 
\begin{equation}\label{2.16}
\diff s^2=-\left(1-\frac{2M}{r}\right) \diff t^2 + \left(1-\frac{2M}{r}\right)^{-1} \diff r^2 + r^2 \diff \Omega^2.
\end{equation}

The vector $K^\mu=\left[1,0,0,0\right]$ is a Killing vector for this metric. The velocity of the fixed observer must be zero in every directions except the time-like. Applying the normalization of the four-velocity one gets:
\begin{equation}\label{2.17}
U^{\mu}U_{\mu}=U^{\mu}U^{\nu}g_{\mu \nu}=\left(U^0\right)^2 g_{00}=-1.
\end{equation}
\begin{equation}\label{2.18}
U^{\mu}=\left[\left(1-\frac{2M}{r}\right)^{-1/2},0,0,0\right].
\end{equation}

Now, we can compute the redshift factor for this geometry, since the only non-vanishing components of the Killing vector and the velocity are the $0$'s:
\begin{equation}\label{2.19}
V=\left(1-\frac{2M}{r}\right)^{1/2}.
\end{equation}

It is straightforward now to calculate the magnitude of the acceleration using (\ref{2.14}):
\begin{equation}\label{2.20}
\begin{split}
a&=\frac{\sqrt{\nabla_\mu V \nabla^\mu V}}{V}\\
&=\frac{M}{r^2\left(1-\frac{2M}{r}\right)^{1/2}}.\\
\end{split}
\end{equation}

If our observer is pretty close to the event horizon then $r-2M\ll2M$ and the acceleration becomes:
\begin{equation}\label{2.21}
a\gg\frac{1}{2M}.
\end{equation}

The time scale, which is the inverse of the acceleration, is therefore small compared to the Schwarzschild radius. Thus the curvature is negligible and the spacetime is basically Minkowski. Then the quantum fluctuations of the vacuum will be the ones of flat spacetime. In that sense, a freely falling observer would measure for a scalar field the Minkowski vacuum near the horizon.

\

Returning to our previous considerations, the fixed observer near the horizon would experience the Unruh effect. Now another observer at a large distance compared with the Schwarzschild radius will not measure a radiation of the Unruh type. Instead, she will measure the radiation observed near the horizon propagating with an appropriate redshift. The temperature of such a thermal radiation is:
\begin{equation}\label{2.22}
T_2=\frac{V_1}{V_2} \ T_1= \frac{V_1}{V_2} \ \frac{a}{2\pi}.
\end{equation}

As we discussed in the previous section, at infinity, $V_2=1$. Therefore:
\begin{equation}\label{2.23}
T_H=\lim_{r\rightarrow 2M} \frac{a \ V}{2\pi}=\frac{\kappa}{2 \pi}.
\end{equation}

When this observer is far from the black hole, she measures a thermal radiation at a certain temperature proportional to its surface gravity. This is known as Hawking effect, and the radiation itself the Hawking radiation. For the Schwarzschild geometry the surface gravity results:
\begin{equation}\label{2.24}
\kappa=\lim_{r\rightarrow 2M}a \ V=\lim_{r\rightarrow 2M}\frac{M}{r^2}=\frac{1}{4M}.
\end{equation}

And so, the Hawking temperature is:
\begin{equation}\label{2.25}
T_H=\frac{1}{8\pi M}.
\end{equation}

If we restore units by using dimensional analysis, we find out:
\begin{equation}\label{2.26}
T_H=\frac{\hbar c^3}{8\pi G k_B M}.
\end{equation}

It is notable saying that this is the only formula where $\hbar$, $c$, $k_B$ and $G$ appear simultaneously. It relates the quantum effects with the macroscopic/thermodynamic world by means of the gravitational interaction. This derivation obviously fails when the state of the quantum field is not regular near the horizon. The only vacuum state that is regular everywhere and invariant under $K=\partial_t$ is named the Hartle-Hawking state.

\

\section{Hawking radiation as a consequence of gravitational collapse}

We will consider the simplest black hole formation process. The black hole is generated by the collapse of a single shock wave. Such an spacetime has a line element in the advanced Eddington-Finkelstein gauge:
\begin{equation}\label{2.27}
\diff s^2=-\left(1-\frac{2 M(v)}{r}\right) \diff v^2 + 2 \diff v \diff r + r^2 \diff \Omega^2.
\end{equation}

The mass function describes the location of the shock wave. If it is placed at some $v=v_0$, then we will have:
\begin{equation}\label{2.28}
M(v)=M \ H(v-v_0),
\end{equation}
where the function $H(v-v_0)$ is a displaced Heaviside step function. Therefore, this spacetime is a composite of two patches, one Minkowski and one Schwarzschild. This is the simplest non-stationary spacetime we can build in which we expect particle creation in the way we saw in Section \textbf{4.2.1}. For convenience we will refer t the Minkowski region as ``in'' region and the ``out'' region will be the Schwarzschild one.

\begin{figure}[t]
	\centering
	\includegraphics[width=0.4\linewidth]{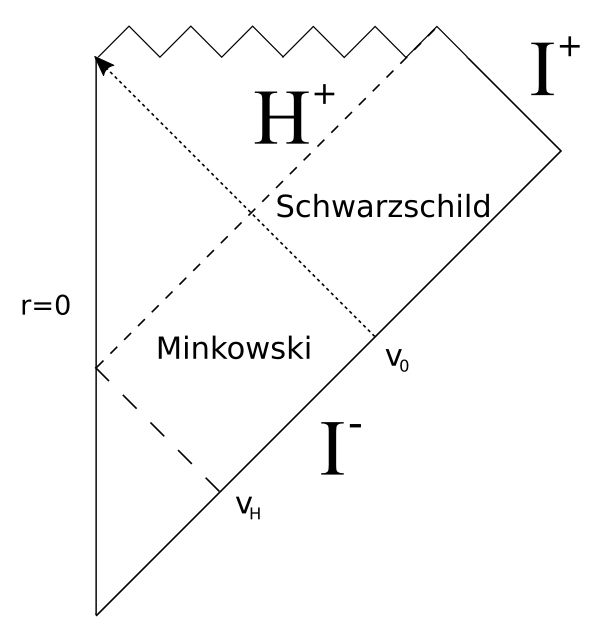}
	\caption[\textit{Formation of a Schwarzschild black hole by a shock wave at $v=v_0$.}]{\textit{Formation of a Schwarzschild black hole by a shock wave at $v=v_0$. $\mathrm{I}^+$ and $\mathrm{I}^+$ are the future and past null infinities respectively.}}
	\label{Figure 2.1}
\end{figure}

Let us consider now the massless Klein-Gordon equation. Led by the spherical symmetry of our spacetime, it is convenient to expand the field in terms of the spherical harmonics:
\begin{equation}\label{2.29}
\phi(t,\vec{x})=\sum_{l,m} \frac{f_l(t,r)}{r} Y_m^l(\theta,\varphi).
\end{equation}

In terms of common directional derivatives, the box operator (\ref{1.5}) takes the form:
\begin{equation}\label{2.30}
\begin{split}
\square \phi&= g^{\mu \nu} \nabla_\mu \nabla_\nu \phi\\
&=\frac{1}{\sqrt{g}} \ \partial_\mu \left(\sqrt{g} \ g^{\mu \nu} \partial_\nu \phi\right).\\
\end{split}
\end{equation}

For Minkowski, we have $\sqrt{g}=r^2 \sin \theta$ and hence the Klein-Gordon equation reduces to:
\begin{equation}\label{2.31}
\begin{split}
\square \phi&= -\partial_t^2 f_l \ \frac{Y_m^l}{r} + \partial_r^2 f_l \ \ \frac{Y_m^l}{r}\\
&\phantom{=} + \left[\frac{1}{r^2 \sin \theta} \ \partial_\theta \left(\sin\theta \ \partial_\theta Y_m^l\right)+\frac{1}{r^2 \sin^2\theta} \ \partial_\varphi^2 Y_m^l\right] \frac{f_l}{r}\\
&=0.\\
\end{split}
\end{equation}

We can perform the analogous calculation with the Schwarzchild metric with the tortoise radial coordinate:
\begin{equation}\label{2.32}
\diff s^2= \left(1-\frac{2M}{r}\right) \left(-\diff t^2 + \diff r_*^2\right) + r^2 \diff \Omega^2,
\end{equation}
where $r_*$ is the tortoise radial coordinate, defined as:
\begin{equation}\label{2.33}
r_*=r+2M \log\left(\frac{r-2M}{2M}\right)
\end{equation}

The Klein-Gordon equation turns out to be:
\begin{equation}\label{2.34}
\begin{split}
\square \phi&=-\left(1-\frac{2M}{r}\right)^{-1} \partial_t^2 f_l \ \frac{Y_m^l}{r}\\
&\hphantom{=}+ \left(1-\frac{2M}{r}\right)^{-1} \frac{1}{r^2} \left[Y_m^l \ r \ \partial_{r_*}^2 f_l  - \frac{2M}{r^2} \left(1-\frac{2M}{r}\right) Y_m^l \ f_l\right]\\
&\hphantom{=} + \left[\frac{1}{r^2 \sin \theta} \ \partial_\theta \left(\sin\theta \ \partial_\theta Y_m^l\right)+\frac{1}{r^2 \sin^2\theta} \ \partial_\varphi^2 Y_m^l\right] \frac{f_l}{r}\\
&=0.\\
\end{split}
\end{equation}

We may use the definition of the spherical harmonics in terms of the Laplacian on the 2-sphere. Because the spherical harmonics satisfy the Laplacian equation on the 2-sphere:
\begin{equation}\label{2.35}
\begin{split}
\Delta Y_m^l(\theta,\phi)&=\frac{1}{r^2 \sin \theta} \ \partial_\theta \left(\sin\theta \ \partial_\theta Y_m^l\right)+\frac{1}{r^2 \sin^2\theta} \ \partial_\varphi^2 Y_m^l\\
&=-\frac{l(l+1)}{r^2} \ Y_m^l(\theta,\phi),
\end{split}
\end{equation}
the above equation simplifies. Plugging this into our previous results and rearranging the expressions we get two two-dimensional wave equations. In the Minkowski case:
\begin{equation}\label{2.36}
\left[-\partial_t^2+\partial_r^2-\frac{l(l+1)}{r^2}\right]f_l(t,r)=0.
\end{equation}

For the Schwarzschild geometry, it takes the form:
\begin{equation}\label{52.37}
\left[-\partial_t^2+\partial_{r_*}^2-\left(1-\frac{2M}{r}\right)\left(\frac{l(l+1)}{r^2}+\frac{2M}{r^3}\right)\right] f_l(t,r)=0.
\end{equation}

We see that at the horizon, $r=2M$, the potential appearing in the equations vanishes. Since we are interested in phenomena happening near the horizon we will neglect the potential everywhere. This way, in both regions the free field equations are satisfied. Even more, we can assume a harmonic time dependence:
\begin{equation}\label{2.38}
f(t,r)=\mathrm{e}^{-i\omega t} f(r).
\end{equation}

The free field equations now take a very simple form. In the case of the Minkowskian ``in'' region ($v<v_0$) we have:
\begin{equation}\label{2.39}
\frac{\diff^2 f(r)}{\diff r^2} + \omega^2 f(r)=0,
\end{equation}
whereas for the Schwarzschild ``out'' region ($v>v_0$):
\begin{equation}\label{2.40}
\frac{\diff^2 f(r)}{\diff r_*^2} + \omega^2 f(r)=0.
\end{equation}

In order to give a solution to this equations in a natural set of coordinates, we are going to introduce a set of null coordinates in each region. In the Minkowski region we will have $u_{in}=t_{in}-r_{in}$ and $v=t_{in}+r_{in}$. Therefore the line element is:
\begin{equation}\label{2.41}
\diff s^2=-\diff v \diff u_{in} + r_{in}^2 \diff \Omega^2.
\end{equation}

In the out, Schwarzschild, region this coordinates result $u_{out}=t_{out}-r_{*out}$ and $v=t_{out}+r_{*out}$, and hence:
\begin{equation}\label{2.42}
\diff s^2=-\left(1-\frac{2M}{r_{out}}\right)\diff v \diff u_{out} + r_{out}^2 \diff \Omega^2.
\end{equation}

The solutions in the ``in'' and ``out'' region are ingoing and outgoing plane waves, respectively. We will take the modes in each region associated with the natural time parameter, at $I^-$ in the Minkowski part, and at $I^+$ in the Schwarzschild one. In the first case the modes take the form:
\begin{equation}\label{2.43}
f_\omega^{in}(v)=\zeta_\omega^{in} \ \mathrm{e}^{-i \omega v},
\end{equation}
where $\omega$ is its frequency and $\zeta^{out}_{\omega}$ is its amplitude. In the out region we have a similar expression:
\begin{equation}\label{2.44}
f_\omega^{out}(u_{out})=\zeta_\omega^{out} \ \mathrm{e}^{-i \omega u_{out}},
\end{equation}
where $\zeta^{out}_{\omega}$ is its amplitude. The normalization constants $\zeta_\omega^{in}$ and $\zeta_\omega^{out}$ may be fixed by imposing the normalization conditions. To do so, we calculate the inner product $\left(f_{\omega}^{in},f_{\omega'}^{in}\right)$. As the procedure is analogous in both cases, we will restrict ourselves to the calculation of the ``in'' case.

\begin{equation}\label{2.45}
\begin{split}
\left(f_{\omega}^{in},f_{\omega'}^{in}\right)&=-i \int_{I^{-}} \diff v \ r^2 \diff \Omega \left(f_{\omega}^{in} \partial_v f_{\omega'}^{in *}-f_{\omega'}^{in *} \partial_v f_{\omega}^{in}\right)\\
&=-4\pi r^2 i \int_{-\infty}^{+\infty} \diff v \left[i \omega' \zeta_\omega^{in} \zeta_{\omega'}^{in *} \ \mathrm{e}^{-i (\omega-\omega') v}+i \omega \zeta_\omega^{in} \zeta_{\omega'}^{in *} \ \mathrm{e}^{-i (\omega-\omega') v}\right]\\
&=4\pi r^2 \zeta_\omega^{in} \zeta_{\omega'}^{in *} (\omega+\omega') \int_{-\infty}^{+\infty} \diff v \ \mathrm{e}^{-i (\omega-\omega') v}\\
&=8\pi^2 r^2 \zeta_\omega^{in} \zeta_{\omega'}^{in *} (\omega+\omega') \ \delta(\omega-\omega').
\end{split}
\end{equation}

By assuming the constant to be real, its value turns out to be:
\begin{equation}\label{2.46}
\zeta_\omega^{in}=\frac{1}{4\pi r\sqrt{\omega}}.
\end{equation}

The same result arises when we calculate the normalization constant for the ``out'' region, just integrating over $I^+$\footnote{The scalar product between solutions of the K-G equation must be defined by an integration over a Cauchy surface. Formally, $I^+$ is not a proper Cauchy surface, we must add the event horizon $H^+$ in order to be consistent. But the modes crossing $H^+$ do not affect the calculation of particle production, so we are going to omit them.} and imposing normalization:
\begin{equation}\label{2.47}
\zeta_\omega^{out}=\frac{1}{4\pi r\sqrt{\omega}}.
\end{equation}

Now, to evaluate the particle production we need an expression for the corresponding Bogoliubov coefficients. As we have seen previously, they determine the mean value of the number operator associated with a set of modes in the vacuum state of the other ones. Choosing $I^-$ as the Cauchy surface for the scalar product, the most important coefficients to compute are written as:
\begin{equation}\label{2.48}
\beta_{\omega \omega'}=-\left(f_{\omega}^{out},f_{\omega'}^{in *}\right)=i \int_{I^{-}} \diff v \ r^2 \diff \Omega \left(f_{\omega}^{out} \partial_v f_{\omega'}^{in}-f_{\omega'}^{in} \partial_v f_{\omega}^{out}\right).
\end{equation}

If we want to do the calculation, we must first know the behaviour of $f_{\omega}^{out}$ at $I^-$. First of all, we must impose a matching condition along $I^-$, near the location of the shock wave, so the metric on both sides are the same. This means:
\begin{equation}\label{2.49}
r(v_0,u_{in})=r(v_0,u_{out}).
\end{equation}

We can use the expression of the Schwarzschild tortoise radial coordinate in terms of the usual radial coordinate (\ref{2.33}). Writing the radial coordinates in terms of $v$, $u_{in}$ and $u_{out}$ and using (\ref{2.47}) we get a relation between the coordinates $u_{in}$ and $u_{out}$:
\begin{equation}\label{2.50}
u_{out}=u_{in}-4M \log\left(\frac{v_0-4M-u_{in}}{4M}\right).
\end{equation}

The free field equation in the Minkowski part of the spacetime implies a regularity condition at $r=0$ and thus the modes must vanish. In this part the ``out'' modes should therefore take the form\footnote{In \cite{f.mbhe} equation (3.71) appears only $H(v_H-v)$ but since you want that term to contribute only when $v<v_H$ I think you should put a $1-H(v_H-v)$ instead.}:
\begin{equation}\label{2.51}
f_{\omega}^{out}=\frac{1}{4\pi\sqrt{\omega}} \left[\frac{\mathrm{e}^{-i \omega u_{out}(u_{in})}}{r} - \frac{\mathrm{e}^{-i \omega u_{out}(v)}}{r} \left(1-H(v_H-v)\right)\right],
\end{equation}
where $v_H=v_0-4M$ is the location of the null ray which forms the event horizon. Since the particle creation will occur at late times, we are going to study this limiting case. We will have $u_{out}\rightarrow\infty$ so we can write:
\begin{equation}\label{2.52}
u_{out}(u_{in})=v_H-4M \log\left(\frac{v_0-4M-u_{in}}{4M}\right).
\end{equation}

Thus we have:
\begin{equation}\label{2.53}
u_{out}(v)=v_H-4M \log\left(\frac{v_0-4M-v}{4M}\right)
\end{equation}
and, at $I^-$, near $v_H$, the modes will behave as:
\begin{equation}\label{2.54}
f_{\omega}^{out}\approx-\frac{1}{4\pi\sqrt{\omega}} \frac{\mathrm{e}^{-i \omega \left[v_H-4M \log\left(\frac{v_0-4M-v}{4M}\right)\right]}}{r} \left(1-H(v_H-v)\right).
\end{equation}

This results to be a superposition of positive and negative frequency modes and therefore we expect particle creation. Before inserting this result in (\ref{2.48}), we can notice something about that integral. Integration by parts allow us to write:
\begin{equation}\label{2.55}
\begin{split}
\beta_{\omega \omega'}&=-\left(f_{\omega}^{out},f_{\omega'}^{in *}\right)=i \int_{I^{-}} \diff v \ r^2 \diff \Omega \left(f_{\omega}^{out} \partial_v f_{\omega'}^{in}-f_{\omega'}^{in} \partial_v f_{\omega}^{out}\right)\\
&=i \int_{I^{-}} \diff v \ r^2 \diff \Omega \left(2 f_{\omega}^{out} \partial_v f_{\omega'}^{in}-\partial_v \left[f_{\omega'}^{in} f_{\omega}^{out}\right]\right)\\
&=2 i \int_{I^{-}} \diff v \ r^2 \diff \Omega \ f_{\omega}^{out} \partial_v f_{\omega'}^{in}.\\
\end{split}
\end{equation}

The resulting boundary term in the second line vanishes because $f_{\omega}^{out}$ is zero at both $v=\pm\infty$. Now we can use the late-time expression for $f_{\omega}^{out}$ and calculate the Bogoliubov coefficients:
\begin{equation}\label{2.56}
\begin{split}
\beta_{\omega \omega'}&=2 i \int_{I^{-}} \diff v \ r^2 \diff \Omega \ f_{\omega}^{out} \partial_v f_{\omega'}^{in}\\
&\simeq-\frac{1}{2\pi} \sqrt{\frac{\omega'}{\omega}} \int_{-\infty}^{v_H} \diff v \ \mathrm{e}^{-i \omega \left[v_H-4M \log\left(\frac{v_H-v}{4M}\right)\right]} \ \mathrm{e}^{-i \omega' v}.\\
\end{split}
\end{equation}

Making a change of variable $x=v_H-v$ and expanding the first exponential we get a simpler integral:
\begin{equation}\label{2.57}
\begin{split}
\beta_{\omega \omega'}&\simeq-\frac{1}{2\pi} \sqrt{\frac{\omega'}{\omega}} \ \mathrm{e}^{-i (\omega+\omega') v_H} (4M)^{-i 4M\omega} \int_{0}^{+\infty} \diff x \ x^{i 4M\omega} \ \mathrm{e}^{i \omega' x}.\\
\end{split}
\end{equation}

This last integral does not converge. However, an extra integration over frequencies (namely, the construction of a wave packet) makes it finite. We can avoid this by adding an infinitesimal real constant $-\epsilon$ to the exponential so that we can use:
\begin{equation}\label{2.58}
\int_0^{+\infty} \diff x \ x^a \ \mathrm{e}^{-b x}=\Gamma(a+1) \ b^{-1-a},
\end{equation}
where $\Gamma(x)$ is the Euler Gamma function. Thus, we finally obtain:
\begin{equation}\label{2.59}
\beta_{\omega \omega'}\simeq-\frac{1}{2\pi} \sqrt{\frac{\omega'}{\omega}} \ \mathrm{e}^{-i (\omega+\omega') v_H} (4M)^{-i 4M\omega} \ \Gamma(1+i 4M\omega) \ \left(-i \omega' + \epsilon\right)^{-1-i 4M\omega}.
\end{equation}

In a similar fashion, we can calculate the value of the other Bogoliubov coefficient:
\begin{equation}\label{2.60}
\alpha_{\omega \omega'}=\left(f_{\omega}^{out},f_{\omega'}^{in}\right)=-i \int_{I^{-}} \diff v \ r^2 \diff \Omega \left(f_{\omega}^{out} \partial_v f_{\omega'}^{in *}-f_{\omega'}^{in *} \partial_v f_{\omega}^{out}\right).
\end{equation}

Following the same steps as before we end up with the result:
\begin{equation}\label{2.61}
\alpha_{\omega \omega'}\simeq-\frac{1}{2\pi} \sqrt{\frac{\omega'}{\omega}} \ \mathrm{e}^{-i (\omega-\omega') v_H} (4M)^{-i 4M\omega} \ \Gamma(1+i 4M\omega) \ \left(+i \omega' + \epsilon\right)^{-1-i 4M\omega}.
\end{equation}

In order to relate both coefficients we compute the quotient:
\begin{equation}\label{2.62}
\begin{split}
\frac{\alpha_{\omega \omega'}}{\beta_{\omega \omega'}}&=\mathrm{e}^{i 2 v_H \omega'} \left(\frac{i \omega' + \epsilon}{-i \omega' + \epsilon}\right)^{-1-i 4M\omega}=\mathrm{e}^{i 2 v_H \omega'} \left(\frac{\omega' -i \epsilon}{-\omega' -i \epsilon}\right)^{-1-i 4M\omega}\\
&=\mathrm{e}^{i 2 v_H \omega'} \mathrm{e}^{\left(-1-i 4M\omega\right) \left[\log\left(\omega'-i \epsilon\right)-\log\left(-\omega'-i \epsilon\right)\right]}.
\end{split}
\end{equation}

Having into account the relation:
\begin{equation}\label{2.63}
\log (-x)=\log x - i \pi,
\end{equation}
and Taylor-expanding the logarithms and neglecting the terms in $\epsilon$ we get:
\begin{equation}\label{2.64}
\begin{split}
\frac{\alpha_{\omega \omega'}}{\beta_{\omega \omega'}}&=\mathrm{e}^{i 2 v_H \omega'} \mathrm{e}^{\left(-1-i 4M\omega\right) \left[\cancel{\log(\omega')}-\cancel{\log(\omega')}+i \pi\right]}\\
&=\mathrm{e}^{i 2 v_H \omega'} \mathrm{e}^{-i \pi} \mathrm{e}^{4 \pi M\omega}\\
&=-\mathrm{e}^{i 2 v_H \omega'} \mathrm{e}^{4 \pi M\omega}.\\
\end{split}
\end{equation}

Thus, we obtain the remarkable relation between the modulus of the Bogoliubov coefficients:
\begin{equation}\label{2.65}
\left|\alpha_{\omega\omega'}\right|=\mathrm{e}^{4 \pi M\omega} \left|\beta_{\omega\omega'}\right|.
\end{equation}

Now we can use the normalization conditions of the Bogoliubov coefficients (\ref{1.23}), for the values $i=j$ to obtain, together with (\ref{2.65}):
\begin{equation}\label{2.66}
\int_{0}^{+\infty} \diff \omega' \left(\left|\alpha_{\omega \omega'}\right|^2 - \left|\beta_{\omega \omega'}\right|^2\right)=\left(\mathrm{e}^{8\pi M\omega}-1\right) \int_{0}^{+\infty} \diff \omega' \ \left|\beta_{\omega \omega'}\right|^2=1.
\end{equation}

The expectation value of the number operator of the ``out'' part of the spacetime in the ``in'' state can be calculated from this last result. Remembering (\ref{1.26}) one can conclude:
\begin{equation}\label{2.67}
\expval{\hat{n}_\omega^{out}}{in}=\int_{0}^{+\infty} \diff \omega' \ \left|\beta_{\omega \omega'}\right|^2=\frac{1}{\mathrm{e}^{8\pi M\omega}-1}.
\end{equation}

This coincides with a Planckian distribution of thermal radiation, following the Bose-Einstein statistic. From there we can expect a certain temperature for the radiating body, the black hole:
\begin{equation}\label{2.68}
T_H=\frac{1}{8\pi M}=\frac{\hbar c^3}{8 \pi G k_B M}.
\end{equation}

We see that we have recovered the previous result (\ref{2.25}).

\section{Black hole evaporation}

As the black hole has a non-zero temperature, we can use the Stefan-Boltzmann law to estimate the radiated power:
\begin{equation}\label{2.69}
P=A\sigma T^4.
\end{equation}

The area of the Schwarzschild black hole is calculated as follows:
\begin{equation}\label{2.70}
\begin{split}
A&=\int_{r=2M}\diff \theta \diff \varphi \sqrt{g_{\theta \theta} \ g_{\varphi\varphi}}\\
&=\int_0^\pi \diff \theta \int_0^{2\pi} \diff \varphi \ 4M^2 \sin \theta\\
&=16\pi M^2.
\end{split}
\end{equation}

Knowing the value of the Stefan constant (in the natural units where $c=\hbar=G=k_B=1$):
\begin{equation}\label{2.71}
\sigma=\frac{\pi^2}{60},
\end{equation}
along with (\ref{2.68}) we can express the radiated power as:
\begin{equation}\label{2.72}
P=\frac{1}{15360 \pi M^2}=\frac{K_{ev}}{M^2}.
\end{equation}

\begin{figure}[t]
	\centering
	\includegraphics[width=0.7\linewidth]{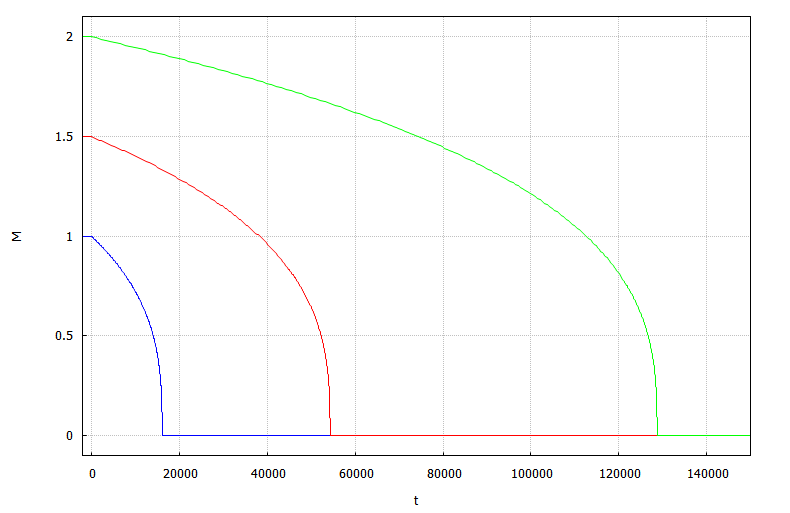}
	\caption{\textit{Evaporation of Schwarzschild black holes of different initial mass.}}
	\label{Figure 2.2}
\end{figure}

From the definition of radiated power it is straightforward to obtain a relation involving the mass:
\begin{equation}\label{2.73}
P=-\frac{\diff E}{\diff t}=-\frac{\diff M}{\diff t}=\frac{K_{ev}}{M^2}.
\end{equation}

To solve this differential equation we must impose two boundary conditions. The first one is related to the initial mass of the black hole: when the black hole is formed, it has a certain mass $M_0$ and from a particular time (we set that to be $t=0$) it starts to evaporate. The second condition is related to the evaporation time: when some finite time has passed, call it $t_f$, the black hole disappears and therefore $M(t\geq t_f)=0$. Thus we obtain:
\begin{equation}\label{2.74}
M(t)=\left\{ \begin{array}{cc}M_0;&t\leq 0\\\left(M_0^3-3 K_{ev} t\right)^{1/3};&0<t<t_f\\0;&t\geq t_f\end{array}\right..
\end{equation}

A straightforward calculation reveals the value of the time at which the black hole disappears:
\begin{equation}\label{2.75}
t_f=\frac{M_0^3}{3K_{ev}}=5120\pi M_0^3.
\end{equation}

We have obtained, as we predicted, the Hawking effect associated to the formation of a black hole. As we have discussed, this implies that the black hole will lose mass and eventually disappear. In the next chapter, we will study in detail the evaporation process. We are going to consider the quantum-mechanical state of this radiation and how it evolves as the black hole decreases.
\chapter{The information loss paradox}\label{Chapter 4}

\section{Heuristic approach to the problem}

Let us study Hawking radiation from the prespective of creation and annihilation of a pair of particle and antiparticle \cite{c.qft}. In particular, let us consider the creation of a pair near the horizon of a black hole. It is possible that one of the partners falls into the hole whereas the other one escapes to infinity. As seen from infinity, the particle that falls has a negative energy because inside the horizon the Killing vector $K=\partial_t$ is spacelike, and the total energy of the pair vanishes. This flux of particles is what we call Hawking radiation.

\begin{figure}[t]
	\centering
	\includegraphics[width=0.4\linewidth]{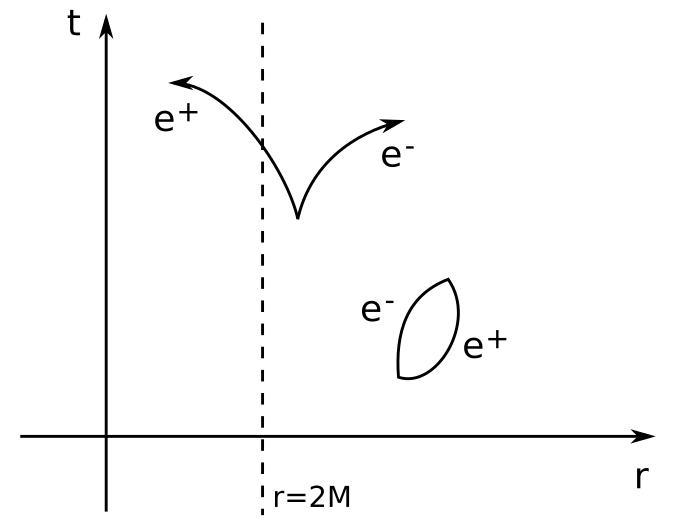}
	\caption{\textit{Particle/antiparticle pair creation near the event horizon of a Schwarzschild black hole.}}
	\label{Figure 3.1}
\end{figure}

The discovery of this thermal spectrum becomes essential for establishing the relation between black hole mechanics and thermodynamics. In particular, for the entropy we have found:
\begin{equation}\label{3.1}
S=\frac{A}{4}=4\pi M^2.
\end{equation}

The last equality holds for a Schwarzschild black hole. We see that this quantity is very large for supermassive black holes, whose mass is some million times the one of our Sun. From a statistical point of view, the entropy is a measure of the number of possible microstates of the black hole. Since a classical black hole is described only in terms of three parameters (its mass, angular momentum and charge) it is difficult to relate them with the microstates. But this problem could be avoided, if we assume that the information we need in order to know the state is hidden behind the event horizon.

\

When we introduce the semi-classical approach of QFT in curved spacetime the problem grows bigger. Now, by means of the Hawking radiation, the mass of the black hole shrinks and eventually all its mass evaporates. At this point we can not appeal the event horizon of the black hole to hide the information, because there is no black hole at all. We are left with a collection of thermal Hawking particles from which we cannot extract the information needed. Any state that collapses into a black hole of certain mass, angular momentum and charge will give us the same Hawking spectrum. This constitutes the basis of the \textbf{information loss paradox}.

\section{Niceness conditions for local evolution}

In our everyday physics we do not worry about the effects of quantum gravity. That is to say, we implicitly assume that there is a limit where those are negligible and we can use approximate expressions for time evolution of the systems of interest. A certain set of conditions are assumed in order to make this possible \footnote{We will mainly follow the discussion in \cite{m.info}}:

\begin{enumerate}
	\item The quantum state is defined on a 3-dimensional spacelike slice, intrinsic curvature which is much smaller than the Planck scale.
	\item The spacelike slice is embedded in a 4-dimensional spacetime in a way that its extrinsic curvature is smaller than the Planck scale.
	\item The curvature of the spacetime near the slice is small compared to the Planck scale.
	\item The matter on the slice has a good behaviour. Its wavelength is much bigger than the Planck length and its energy and momentum density are small compared to Planck scale.
	\item Considering evolution, all slices have the good behaviour described above. That requires the lapse and shift vectors to change smoothly along the slices. 
\end{enumerate}

Here, the quantum process we are concerned about is the Hawking effect, as seen in the previous chapter. This process implies the creation of pairs of particles near the horizon of the black hole. The particles created are entangled. If we want to take into account only the essence of this entanglement, we can assume that the state of the pair is a Bell state. Assuming locality and the existence of matter on the slice, but far away from where the pair is created, the state on the spacetime slide would be:
\begin{equation}\label{3.2}
\ket{\Psi}\approx\ket{\psi}_M\otimes\frac{1}{\sqrt{2}}\left(\ket{0}_r\ket{0}_l+\ket{1}_r\ket{1}_l\right).
\end{equation}

The presence of matter will affect the pair, even if it is located far away from where these are created. Consequently, we have written the $\approx$ sign in this last equation. If we construct the respective \textbf{density matrix of particle} $r$, $\rho=\sum_i p_i \ket{\psi_i}_r\bra{\psi_i}_r$, a straightforward computation gives us the entropy of entanglement:
\begin{equation}\label{3.3}
\begin{split}
S_{ent}&=-\tr(\rho \log \rho)\\
&=-\sum_i \rho_i \log \rho_i\\
&=-2 \ \frac{1}{2} \log\left(\frac{1}{2}\right)\\
&=\log 2.
\end{split}
\end{equation}

Locality allows small deviations from this state where the matter has no effect on the pair. Small deviations must preserve the entanglement between the particles. Because of that, we can express the locality condition in terms of this entanglement entropy as:
\begin{equation}\label{3.4}
\frac{S_{ent}}{\log 2}-1\ll 1.
\end{equation}

\section{Slicing the black hole geometry}

A black hole is called a traditional black hole if it has an information-free horizon. A point on the horizon is information-free if in its vicinity the evolution of the quantum fields with wavelength $\lambda\gg l_P$ is given by the semiclassical evolution of quantum field theory in curved spacetime.

\

The traditional Schwarzschild black hole geometry has a singularity at $r=0$. Since we want our niceness conditions to hold, the spacelike surfaces where our quantum state is defined can not intersect the singularity. Then, we have to construct them in a very specific way to make our next arguments consistent.

\

Consider a spacetime of the type we have worked with in Section \textbf{2.3}. We have flat spacetime with a spherical shell of mass $M$ collapsing towards the origin. At a perceptible distance far from the event horizon, like $r>4M$, we take the slice to be $t_1=constant$. Since space and time directions interchange roles inside the event horizon, there the slices are $r_1=constant$. We fix this for the interval $M/2<r<3M/2$ in order not to be near the horizon nor the singularity. Those parts are connected by a smooth segment that satisfies the niceness conditions. We finish our spacelike slice by extending it to early times when the black hole has not formed yet and smoothly taking it to $r=0$.

\begin{figure}[t]
	\centering
	\includegraphics[width=0.7\linewidth]{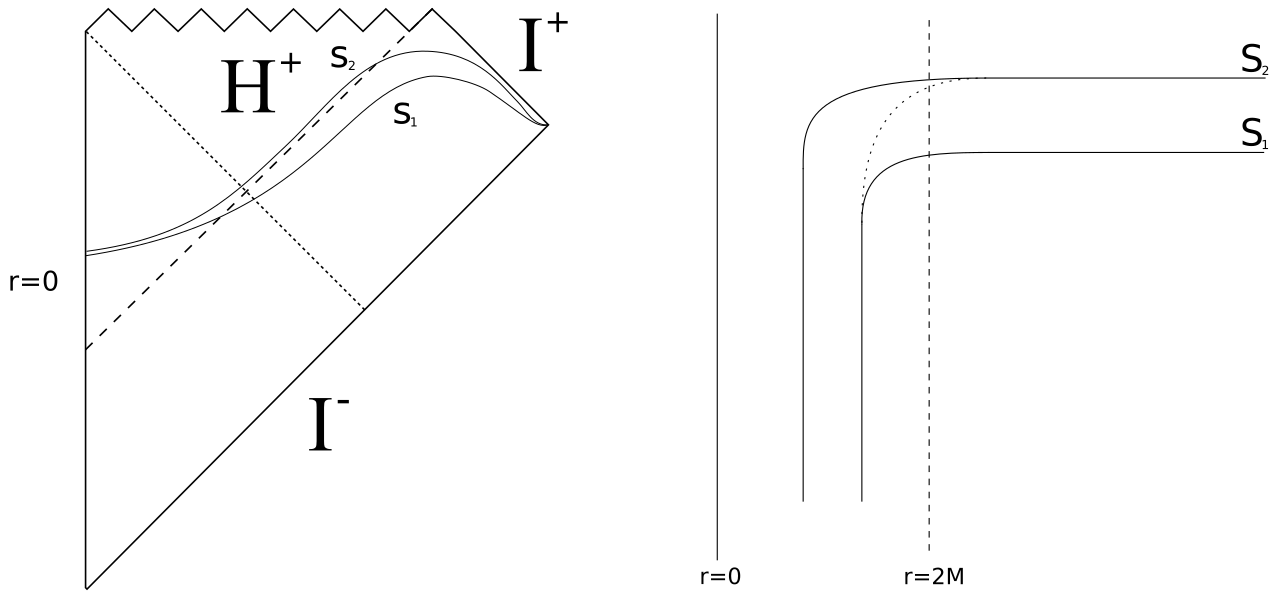}
	\caption{\textit{Conformal diagram of a black hole formed by gravitational collapse (left) and schematic coordinates of the Schwarzschild black hole geometry (right). In both cases we represent two spacelike slices constructed as we discussed.}}
	\label{Figure 3.2}
\end{figure}

Later-time slices are constructed using the same rules: taking $t_2=t_1+\Delta$ and at $r>4M$ and $r_2=r_1+\delta$ at the $r=constant$ part. If we take the limit $\delta\rightarrow 0$ then the connector is the same for all slices, but the $r=constant$ part is longer as the time runs. In the evolution, the first connector must stretch to cover the connector of the second slice and its extra $r=constant$ part. For a larger succession of spacelike slices the  evolution follows exactly the same steps.

\

The foliation constructed satisfies the niceness conditions of the previous section. The region of the spacelike hypersurfaces which covers the horizon keeps stretching with time and therefore the 
modes of the fields stretch its wavelengths, which leads to particle creation.

\section{The radiation process to leading order}

Let us consider an initial spacelike slice, where the massive shell that collapses to form the black hole is represented by the state $\ket{\psi}_M$. As we evolve to the next slice the middle region stretches and therefore a pair of particles is created. The matter state stays the same, because that part of the spacelike slice remains unchanged. The total quantum state is (\ref{3.2}):
\begin{equation}\label{3.5}
\ket{\Psi}\approx\ket{\psi}_M\otimes\frac{1}{\sqrt{2}}\left(\ket{0}_{r_1}\ket{0}_{l_1}+\ket{1}_{r_1}\ket{1}_{l_1}\right).
\end{equation}

We have calculated before the entanglement entropy of this state, which is given in (\ref{3.3}). If we evolve again to the next slide, the stretching moves away the previously created pair and creates a new one. Neglecting interaction between the different pairs, the state would be:
\begin{equation}\label{3.6}
\ket{\Psi}\approx\ket{\psi}_M\otimes\frac{1}{\sqrt{2}}\left(\ket{0}_{r_1}\ket{0}_{l_1}+\ket{1}_{r_1}\ket{1}_{l_1}\right)\otimes\frac{1}{\sqrt{2}}\left(\ket{0}_{r_2}\ket{0}_{l_2}+\ket{1}_{r_2}\ket{1}_{l_2}\right).
\end{equation}

The entropy of entanglement of the right particles with the lefts and the in-falling mass results:
\begin{equation}\label{3.7}
S_{ent}=2\log 2.
\end{equation}

We can repeat this step $N$ times, to obtain a state:
\begin{equation}\label{3.8}
\begin{split}
\ket{\Psi}\approx\ket{\psi}_M&\otimes\frac{1}{\sqrt{2}}\left(\ket{0}_{r_1}\ket{0}_{l_1}+\ket{1}_{r_1}\ket{1}_{l_1}\right)\\
&\otimes\frac{1}{\sqrt{2}}\left(\ket{0}_{r_2}\ket{0}_{l_2}+\ket{1}_{r_2}\ket{1}_{l_2}\right)\\
&...\\
&\otimes\frac{1}{\sqrt{2}}\left(\ket{0}_{r_N}\ket{0}_{l_N}+\ket{1}_{r_N}\ket{1}_{l_N}\right).\\
\end{split}
\end{equation}

The entanglement of $\left\{r_i\right\}$ with $\left\{l_i\right\}$ and $M$ is trivial to calculate:
\begin{equation}\label{3.9}
S_{ent}=N\log 2.
\end{equation}

As a consequence of the emission of quanta, the mass of the black hole decreases. We will stop evolving the spacelike slices when the mass reaches the Planck scale, because at that point they will not satisfy the niceness conditions.

\begin{figure}[t]
	\centering
	\includegraphics[width=0.4\linewidth]{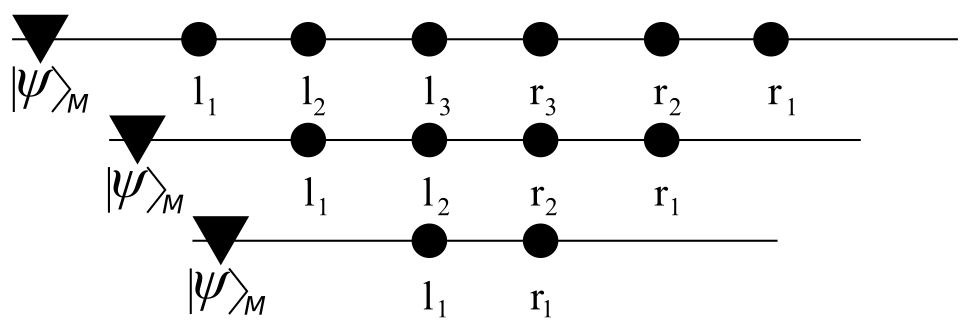}
	\caption{\textit{Creation of Hawking pairs. When new pairs are formed, the previous one move away as well as the mass state $\ket{\psi}_M$.}}
	\label{Figure 3.3}
\end{figure}

Once reached this point, our discussion forces us to choose between two distinct possibilities:

\begin{itemize}
	\item The black hole has disappeared completely. The emitted particles have an entanglement entropy $S_{ent}\neq0$, but there is nothing to be entangled with. Thus its state can not described by a single ket vector, instead it is described in terms of its associate density matrix. This is what is called a mixed state.
	\item We are left with an object with a Planck scale mass, which is bounded both in mass and length, and which has an arbitrarily high entanglement with the quanta escaped. In this case we say that our theory contains remnants.
\end{itemize}

The first possibility implies a non-unitary evolution: a pure state evolves to a mixed state, which is in contradiction with quantum mechanics. The latter leads to an unexpected phenomenon: it is not a violation of the laws of quantum mechanics, since a system with bounds on energy and length is expected to have a finite number of states. Using the previous argumentation we can not choose between these possibilities.

\section{Corrections of the leading order}

We can think that if we add small corrections to the state of our quantum system we may invalidate the previous argument and therefore bypass the information problem. In this section, our purpose is (i) to prove that small corrections are irrelevant to the paradox and (ii), when a correction that could avoid the paradox is considered, it violates the niceness conditions.

\

Consider that at certain time $t_n$ we have the state $\ket{\Psi_{M,l},\psi_r(t_n)}$. The first term, $\Psi_{M,l}$, denotes the matter that forms the hole and the Hawking particles that have fallen into it, and the second one, $\psi_r(t_n)$, refers to the ones that have escaped to infinity. There is entanglement between both parts. If we choose a orthonormal basis for both $(M,l)$ and $r$, which is denoted as $\left\{\psi_m,\chi_n\right\}$ we can write the state as:
\begin{equation}\label{3.10}
\ket{\Psi_{M,l},\psi_r(t_n)}=\sum_{m,n} C_{mn} \psi_m \chi_n.
\end{equation}

Then, it is possible to perform a series of unitary transformations on the basis vectors to express the state in terms of a diagonal matrix:
\begin{equation}\label{3.11}
\ket{\Psi_{M,l},\psi_r(t_n)}=\sum_{i} C_{i} \psi_i \chi_i.
\end{equation}

The density matrix describing the right particles is just the modulus of these coefficients, and therefore the entanglement results:
\begin{equation}\label{3.12}
S_{ent}=-\sum_{i} \left|C_i\right|^2 \log \left|C_i\right|^2.
\end{equation}

Now, we assume that when time passes the state of the pair created in the region of the spacelike slice that stretches is spanned by the set of vectors:
\begin{equation}\label{3.13}
\begin{split}
&S^{(1)}=\frac{1}{\sqrt{2}} \left(\ket{0}_{l_{n+1}}\ket{0}_{r_{n+1}}+\ket{1}_{l_{n+1}}\ket{1}_{r_{n+1}}\right),\\
&S^{(2)}=\frac{1}{\sqrt{2}} \left(\ket{0}_{l_{n+1}}\ket{0}_{r_{n+1}}-\ket{1}_{l_{n+1}}\ket{1}_{r_{n+1}}\right).\\
\end{split}
\end{equation}

Invoking locality, we neglect interaction between the recently created pair and the quanta that left the vicinity of the black hole. Thus, a general evolution to $t_{n+1}$ which obeys these conditions must have the form:
\begin{equation}\label{3.14}
\begin{split}
&\chi_i\rightarrow\chi_i,\\
&\psi_i\rightarrow\psi_i^{(1)} S^{(1)}+\psi_i^{(2)} S^{(2)}.\\
\end{split}
\end{equation}

Since we want the evolution to be unitary and the vectors $S^{(i)}$ are properly normalized, we have:
\begin{equation}\label{3.15}
\left\|\psi_i^{(1)}\right\|^2+\left\|\psi_i^{(2)}\right\|^2=1.
\end{equation}

The evolved state is therefore given by:
\begin{equation}\label{3.16}
\begin{split}
\ket{\Psi_{M,l},\psi_r(t_{n+1})}&=\sum_{i} C_{i} \left[\psi_i^{(1)} S^{(1)}+\psi_i^{(2)} S^{(2)}\right] \chi_i\\
&=S^{(1)} \sum_{i} C_{i} \psi_i^{(1)} \chi_i + S^{(2)} \sum_{i} C_{i} \psi_i^{(2)} \chi_i\\
&=S^{(1)} \Lambda^{(1)} + S^{(2)} \Lambda^{(2)}.\\
\end{split}
\end{equation}

Again, normalization of the state implies the following relation:
\begin{equation}\label{3.17}
\left\|\Lambda^{(1)}\right\|^2+\left\|\Lambda^{(2)}\right\|^2=1.
\end{equation}

In the leading order evolution we have $\psi_i^{(1)}=\psi_i$ and $\psi_i^{(2)}=0$ and as a consequence $\left\|\Lambda^{(1)}\right\|^2=1$ and $\left\|\Lambda^{(2)}\right\|^2=0$. These results help us to define what we call a small correction: we are going to consider corrections to be small if:
\begin{equation}\label{3.18}
\left\|\Lambda^{(2)}\right\|<\epsilon,
\end{equation}
with $\epsilon\ll1$.

\section{Entropy bounds}

Our next step is to compute the entanglement entropy between the right quanta outside the hole, including the ones that have been created before the last time evolution and the ones that have just emerged, and the particles that are inside. In order to do so, we divide our system in three subsystems: the radiation particles $\left\{r_i\right\}$ emitted until time $t_n$, the matter in the black hole $\left(M,\left\{l_i\right\}\right)$ at that same moment, and the recently created pair $\left(r_{n+1},l_{n+1}\right)$.

\

We first compute the entanglement entropy of the new pair with the rest of the system. Since every $2\times2$ matrix may be expressed in the basis $\left(I,\left\{\sigma_i\right\}\right)$, where $\left\{\sigma_i\right\}$ are the three Pauli matrices, the associated density matrix admits an expansion of the form:
\begin{equation}\label{3.19}
\rho_{\left(r_{n+1},l_{n+1}\right)}=\frac{1}{2} \ I + \left|\vec{\alpha}\right| \sigma_3.
\end{equation}

The density matrix is given by:
\begin{equation}\label{3.20}
\rho_{\left(r_{n+1},l_{n+1}\right)}=\left(\begin{array}{cc}\braket{\Lambda^{(1)}}{\Lambda^{(1)}}&\braket{\Lambda^{(1)}}{\Lambda^{(2)}}\\\braket{\Lambda^{(2)}}{\Lambda^{(1)}}&\braket{\Lambda^{(2)}}{\Lambda^{(2)}}\end{array}\right).
\end{equation}

A straightforward diagonalization exercise gives us the value of the $\left|\vec{\alpha}\right|$ factor:
\begin{equation}\label{3.21}
\left|\vec{\alpha}\right|=\frac{\sqrt{1-4\left(\left\|\Lambda^{(1)}\right\|\left\|\Lambda^{(2)}\right\|-\braket{\Lambda^{(1)}}{\Lambda^{(2)}}^2\right)}}{2}.
\end{equation}

We can compute the entropy of entanglement for a matrix of the type (\ref{3.20}):
\begin{equation}\label{3.22}
\begin{split}
S_{ent}&=-\tr\left(\rho \log\rho\right)=-\tr \left(\begin{array}{cc}\left(\frac{1}{2}+\left|\vec{\alpha}\right|\right)\log\left(\frac{1}{2}+\left|\vec{\alpha}\right|\right)&0\\0&\left(\frac{1}{2}-\left|\vec{\alpha}\right|\right)\log\left(\frac{1}{2}-\left|\vec{\alpha}\right|\right)\end{array}\right)\\
&=-\frac{1+2\left|\vec{\alpha}\right|}{2} \ \log\left(\frac{1+2\left|\vec{\alpha}\right|}{2}\right)-\frac{1-2\left|\vec{\alpha}\right|}{2} \ \log\left(\frac{1-2\left|\vec{\alpha}\right|}{2}\right)\\
&=\log 2 -\frac{1}{2}\left(1+2\left|\vec{\alpha}\right|\right) \ \log\left(1+2\left|\vec{\alpha}\right|\right)-\frac{1}{2}\left(1-2\left|\vec{\alpha}\right|\right) \ \log\left(1-2\left|\vec{\alpha}\right|\right).\\
\end{split}
\end{equation}

Using (\ref{3.18}) we have $\left\|\Lambda^{(2)}\right\|=\epsilon_1<\epsilon$, and using the \textit{Cauchy-Bunyakovsky inequality} we can obtain:
\begin{equation}\label{3.23}
\braket{\Lambda^{(1)}}{\Lambda^{(2)}}=\epsilon_2\leq\left\|\Lambda^{(1)}\right\|\left\|\Lambda^{(2)}\right\|=\left\|\Lambda^{(1)}\right\| \epsilon_1<\epsilon.
\end{equation}

Now, we can approximate equation (\ref{3.21}) to obtain a clearer expression for the entropy:
\begin{equation}\label{3.24}
\left|\vec{\alpha}\right|\simeq\frac{\sqrt{1-4\left(\epsilon_1^2-\epsilon_2^2\right)}}{2}=\frac{1-2\left(\epsilon_1^2-\epsilon_2^2\right)}{2} + \mathcal{O}(\epsilon^3).
\end{equation}

Thus, the entropy results:
\begin{equation}\label{3.25}
\begin{split}
S_{ent}\left(r_{n+1},l_{n+1}\right)&=\log 2 -\left(1-\left(\epsilon_1^2-\epsilon_2^2\right)\right)\log\left[2-2\left(\epsilon_1^2-\epsilon_2^2\right)\right]\\
&\hphantom{=}-\left(\epsilon_1^2-\epsilon_2^2\right)\log\left[2\left(\epsilon_1^2-\epsilon_2^2\right)\right]\\
&=-\left(1-\left(\epsilon_1^2-\epsilon_2^2\right)\right)\log\left[1-\left(\epsilon_1^2-\epsilon_2^2\right)\right]-\left(\epsilon_1^2-\epsilon_2^2\right)\log\left[\left(\epsilon_1^2-\epsilon_2^2\right)\right]\\
&=\left(\epsilon_1^2-\epsilon_2^2\right) \ \log\left(\frac{\mathrm{e}}{\epsilon_1^2-\epsilon_2^2}\right)+\mathcal{O}(\epsilon^3)<\epsilon.
\end{split}
\end{equation}

The entanglement of the right particles is given by (\ref{3.12}), we denote it as $S_{ent}(\left\{r_i\right\})=S_0$. Taking this along our previous result, the entropy of entanglement between those particles and the new pair with the rest of the systems satisfy, as a consequence of subadditivity property of the entropy;
\begin{equation}\label{3.26}
\begin{split}
S_{ent}\left(\left\{r_i\right\},r_{n+1},l_{n+1}\right)&\geq\left|S_{ent}(\left\{r_i\right\})-S_{ent}(r_{n+1},l_{n+1})\right|\\
&\geq S_0-\epsilon.\\
\end{split}
\end{equation}

Finally, we are going to calculate the entanglement entropy of the in-falling newly created particle with the rest of the system. To do so, we first re-write the state (\ref{3.16}) as:
\begin{equation}\label{3.27}
\begin{split}
\ket{\Psi_{M,l},\psi_r(t_{n+1})}&=S^{(1)} \Lambda^{(1)}+S^{(2)} \Lambda^{(2)}\\
&=\frac{1}{\sqrt{2}}\left[\ket{0}_{l_{n+1}}\ket{0}_{r_{n+1}}\left(\Lambda^{(1)}+\Lambda^{(2)}\right)\right.\\
&\hphantom{=}\left.+\ket{1}_{l_{n+1}}\ket{1}_{r_{n+1}}\left(\Lambda^{(1)}-\Lambda^{(2)}\right)\right].\\
\end{split}
\end{equation}

The density matrix associated with the $l_{n+1}$ particle is given by:
\begin{equation}\label{3.28}
\rho_{l_{n+1}}=\frac{1}{2}\left(\begin{array}{cc}\braket{\Lambda^{(1)}+\Lambda^{(2)}}{\Lambda^{(1)}+\Lambda^{(2)}}&0\\0&\braket{\Lambda^{(1)}-\Lambda^{(2)}}{\Lambda^{(1)}-\Lambda^{(2)}}\end{array}\right).
\end{equation}

As a consequence of the triangular inequality we find out:
\begin{equation}\label{3.29}
\begin{split}
\left\|\Lambda^{(1)}\pm\Lambda^{(2)}\right\|^2&\leq\left(\left\|\Lambda^{(1)}\right\|\pm\left\|\Lambda^{(1)}\right\|\right)^2\\
&=\left\|\Lambda^{(1)}\right\|^2+\left\|\Lambda^{(2)}\right\|^2\pm 2\left\|\Lambda^{(1)}\right\|\left\|\Lambda^{(2)}\right\|\\
&=1\pm 2\left\|\Lambda^{(1)}\right\|\left\|\Lambda^{(2)}\right\|.\\
\end{split}
\end{equation}

But, according to equation (\ref{3.23}), we can approximate the last expression as:
\begin{equation}\label{3.30}
\left\|\Lambda^{(1)}\pm\Lambda^{(2)}\right\|^2=1\pm 2\braket{\Lambda^{(1)}}{\Lambda^{(2)}}+\mathcal{O}(\epsilon^2)
\end{equation}

Thus the entanglement entropy may be calculated using (\ref{6.22}) since now our density matrix ha the suitable form. Performing a Taylor series and keeping terms to second order in $\epsilon_2$ gives:
\begin{equation}\label{3.31}
\begin{split}
S_{ent}(l_{n+1})&=\log 2 - \frac{1}{2} \left[\left(1+2\epsilon_2\right) \log \left(1+2\epsilon_2\right) + \left(1-2\epsilon_2\right) \log \left(1-2\epsilon_2\right)\right]\\
&=\log 2 - 2 \epsilon_2^2 + \mathcal{O}(\epsilon^3)\geq\log 2 -2\epsilon^2>\log 2 - \epsilon.\\
\end{split}
\end{equation}

Now, using the results (\ref{3.25}), (\ref{3.26}) and (\ref{3.31}), we are able to prove the following theorem:

\

\begin{theorem}[\textbf{Stability of entanglement}]
	At time $t_{n+1}$, taking the entropy of entanglement of the emitted quanta until the previous time step $t_n$ to be $S_0$, if the recently created pair is in a state which differs from the leading order by small corrections satisfying (\ref{3.18}) then the entropy of the emitted quanta to that time satisfy:
	\begin{equation}\label{3.32}
	S_{ent}\left(\left\{r_i\right\},r_{n+1}\right)>S_0+\log 2-2\epsilon.
	\end{equation}
\end{theorem}

\begin{proof}
	By virtue of the subadditivity theorem, the entropies of entanglement of our three subsystems obey:
	\begin{equation}\label{3.33}
	S_{ent}\left(\left\{r_i\right\},r_{n+1}\right)+S_{ent}\left(r_{n+1},l_{n+1}\right)\geq S_{ent}\left(\left\{r_i\right\}\right)+S_{ent}\left(l_{n+1}\right).
	\end{equation}
	
	Using the fact that $S_{ent}\left(\left\{r_i\right\}\right)=S_0$ and equations (\ref{3.25}) and (\ref{3.31}) we straightforwardly obtain:
	\begin{equation}\label{3.34}
	S_{ent}\left(\left\{r_i\right\},r_{n+1}\right)>S_0+\log 2-2\epsilon.
	\end{equation}
\end{proof}

This result establish that the entanglement entropy always increases, if we assume small deviations from the leading order state (\ref{3.5}). Thus the conclusion of Section \textbf{3.4} does not change: the evaporation process leads to mixed states/remnants since we have restricted to small corrections.

\section{Hawking's theorem}

At this stage of the proceedings, we can establish the following theorem:

\begin{theorem}[\textbf{Hawking's theorem}]
	If we assume:	
	\begin{enumerate}
		\item The niceness conditions of Section \textbf{3.2}, which ensures approximate expressions for time evolution.
		\item The existence of the traditional black hole.
	\end{enumerate}
	
	Then, the process of formation and evaporation of such a black hole leads to mixed states/remnants.
\end{theorem}

\begin{proof}
	Using the first assumption of the theorem, in the region near the horizon, in which by construction the niceness conditions hold, we can follow the evolution of an outgoing mode. We can expand this mode in the Fock basis:
	\begin{equation}\label{3.35}
	\ket{\psi}=\sum_n \alpha_n \ket{n}.
	\end{equation}
	
	If our mode has $\lambda>M$ then we should have:
	\begin{equation}\label{3.36}
	\sum_n \left|\alpha_n\right|^2 < \gamma; \hspace{5mm} \gamma\ll 1.
	\end{equation}
	
	That quantity can not be order unity, otherwise we would have particles at the horizon and the second assumption of the theorem would be violated. So, we are going to have vacuum modes which evolve in agreement to the leading order to some accuracy. As a consequence we will have some $\epsilon$ which verifies (\ref{3.18}).
	
	\
	
	Thus, the pairs of particles that will be produced must be in a state near $S^{(1)}$. As we have demonstrated in the previous section, the entropy of entanglement increases as time passes. If before the black hole becomes Planck-sized as a consequence of evaporation the mechanism has produced $N$ pairs, the entanglement would verify:
	\begin{equation}\label{3.37}
	S_{ent}>\frac{N}{2} \ \log 2.
	\end{equation}
	
	As discussed in Section \textbf{3.4}, this result will inevitably lead us to mixed states/remnants.
\end{proof}

We have stated this conclusion as a theorem. To avoid it we must violate one of its assumptions. We are left with two possibilities: either the niceness conditions are not sufficient and we have to add something else, or the traditional black hole can not arise in our theory.

\

Despite being unpopular, another possibility is the acceptance of the conclusions of the theorem and expect new unknown physics in the process of black hole formation and evaporation.

\section{The information problem}

In our previous discussion we have not even mentioned the term ``information''. We have just worried about the nature of the Hawking radiation. We can pay attention, for example, to the simple leading order state (\ref{3.8}). The radiation quanta is entangled with the matter inside the black hole, by virtue of (\ref{3.9}), and also carries no information about the initial matter state $\ket{\psi}_M$. Even if we consider corrections to that leading order, we end up with quanta entangled with the matter in the hole and the only amount of information which carries about $\ket{\psi}_M$ will be an infinitesimal part of the total, arising form the corrections of order $\epsilon$.

\

Thus, if the black hole at last disappears we are left with a mixed state and no information about the matter which had created the black hole. Otherwise, our theory would contain remnants. Let us consider two explicit examples to illustrate the problem.

\

\begin{example}
	Consider the simple initial matter state:
	\begin{equation}\label{3.38}
	\ket{\Psi}_0=\ket{\psi}_M=\alpha \ket{0}_M+ \beta \ket{1}_M,
	\end{equation}
	which evolves in a first time step as:
	\begin{equation}\label{3.39}
	\begin{split}
	\ket{\Psi}_1&=\left(\frac{1}{\sqrt{2}}\ket{1}_M\ket{0}_{l_1}+\frac{1}{\sqrt{2}}\ket{0}_M\ket{1}_{l_1}\right)\otimes\left(\alpha \ket{0}_{r_1} + \beta \ket{1}_{r_1}\right).\\
	\end{split}
	\end{equation}
	
	To this step, the black hole have emitted one radiation quantum, which in this case carries the information about the initial matter state. The system evolves one more time before the black hole completely disappears, giving the state:
	\begin{equation}\label{3.40}
	\begin{split}
	\ket{\Psi}_2&=\left(\frac{1}{\sqrt{2}}\ket{1}_M\ket{0}_{l_1}+\frac{1}{\sqrt{2}}\ket{0}_M\ket{1}_{l_1}\right)\otimes\left(\alpha \ket{0}_{r_1} + \beta \ket{1}_{r_1}\right)\\
	&\phantom{=}\otimes\left(\frac{1}{\sqrt{2}}\ket{1}_{r_2}\ket{0}_{l_2}+\frac{1}{\sqrt{2}}\ket{0}_{r_2}\ket{1}_{l_2}\right).\\
	\end{split}
	\end{equation}
	
	The second outgoing quantum is entangled with the matter in the hole. This leads to a loss of unitarity, however no information have been lost in the process.
	
\end{example}

\begin{example}
	Consider the same initial matter state $\ket{\Psi}_0$ in (\ref{3.38}), as in the previous example. In this case, it evolves as:
	\begin{equation}\label{3.41}
	\ket{\Psi}_1=\left(\alpha\ket{1}_M\ket{0}_{l}+\beta\ket{0}_M\ket{1}_{}\right)\otimes\left(\frac{1}{\sqrt{2}} \ket{0}_{r} + \frac{1}{\sqrt{2}} \ket{1}_{r}\right).
	\end{equation}
	
	Then, the outgoing quantum is not entangled with the rest. However, it does not carry any information about the initial matter state.
\end{example}

Our state exhibits these two problems simultaneously. It is expected that a solution to the paradox will solve both of them. Nevertheless, we have to keep in mind that mixed states and information loss are two different terms, which, in general, involve distinct physical situations.

\

In this chapter, we have proved that the conclusion of mixed states/remnants is unavoidable when the evaporation process of a black hole is studied. Moreover, the information loss which implies the process has been explicitly described. The next logical step in our discussion will be the study of a possible solution to the problem, the one that provides the BMS symmetry group.
\chapter{The BMS symmetry group}\label{Chapter 5}

When Bondi, van der Burg, Metzner and Sachs (BMS) first studied the asymptotic symmetries of asymptotically flat spacetimes \cite{bbm}, \cite{sachs} their expectations were to reproduce the symmetries of flat spacetimes, i.e. the Poincaré group (the group of all Lorentz transformations together with space-time translations). The surprise was that they found an infinite-dimensional group which contains as a subgroup the finite-dimensional Poincaré group. This result has an astounding conclusion: General Relativity does not reduce to Special Relativity in the case of weak fields at long distances.

\

The purpose of the first part of the present chapter is to give a brief review of asymptotic symmetries, paying special attention to references \cite{h.bms}, \cite{stro.lec} and \cite{zheng.bms}, as obtaining the fundamental results of the subject is a very complicated issue completely off the limits of this dissertation. We will review the structure of asymptotically flat spacetimes, and how supertransformations arise.

\

The second half is devoted to the study of the relation between these symmetries and the information loss problem.

\section{Asymptotically flat spacetimes}

If we ask for the symmetries that leave Minkowski spacetime invariant forms the so-called Poincaré group. This set of ten isometries arise when we look for solutions to the Killing equation:
\begin{equation}\label{5.1}
\pounds_{\xi} \eta_{\mu\nu}=\nabla_\mu\xi_\nu+\nabla_\nu\xi_\mu=0.
\end{equation}

The set of transformations generated by the solutions to this equation is made up of three boosts, three rotations and four translations. It is possible to extend the Poincaré group by adding conformal transformations, which are the ones that preserve the metric up to a conformal factor:
\begin{equation}\label{5.2}
\pounds_{\xi}\eta_{\mu\nu}=\Omega^2 \eta_{\mu\nu}.
\end{equation}

But, as we said earlier, the symmetry group in the case of curved spacetimes (even if they are asymptotically flat) is larger than this one. Let us look for symmetries which leave unchanged the boundary conditions of asymptotic flatness.

\begin{figure}[t]
	\centering
	\includegraphics[width=0.6\linewidth]{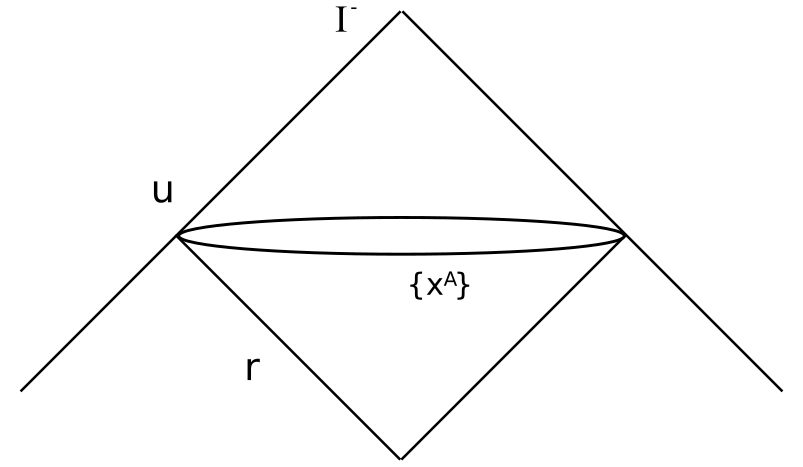}
	\caption{\textit{Minkowski spacetime in retarded Bondi coordinates $(u,r,\{x^A\})$.}}
	\label{Figure 5.1}
\end{figure}

We can use the retarded Bondi coordinates $(u,r,\{x^A\})$ to rewrite the usual Minkowski metric as:
\begin{equation}\label{5.3}
\diff s^2=-\diff u^2-2 \diff r \diff u + r^2 \gamma_{AB} \diff x^A \diff x^B.
\end{equation}

We see that this metric falls into the so-called Bondi gauge:
\begin{equation}\label{5.4}
\begin{split}
g_{rr}=g_{rA}&=0,\\
\partial_r\left(\frac{\det g_{AB}}{r^2}\right)&=0.\\
\end{split}
\end{equation}

We are interested in studying metrics which are asymptotically equal to flat spacetime. If we want to ensure the compliance with the asymptotically flat boundary conditions when we perform a transformation, the resulting components of the Riemann tensor must vanish sufficiently fast when taking the limit $r\rightarrow\infty$. This means that for every component of the curvature tensor, a fall-off rate must be established:
\begin{equation}\label{5.5}
R_{\mu\nu\rho\sigma}\left(g_{\alpha\beta}+\delta g_{\alpha\beta}\right) \xrightarrow{r\rightarrow\infty} 0 + \mathcal{O}\left(\frac{1}{r^{\kappa}}\right).
\end{equation}

As a consequence, the variations of the metric will obey some fall-off conditions too. The following set of fall-off rates was proposed by Sachs \footnote{This is a sufficient but non-necessary condition, weak enough to allow the existence of gravitational waves but strong enough to get rid of unphysical solutions. See page 64 of \cite{stro.lec}.}:
\begin{equation}\label{5.6}
\begin{split}
\delta g_{uA}&\sim\mathcal{O}\left(r^0\right),\\
\delta g_{ur}&\sim\mathcal{O}\left(r^{-2}\right),\\
\delta g_{uu}&\sim\mathcal{O}\left(r^{-1}\right),\\
\delta g_{AB}&\sim\mathcal{O}\left(r\right).\\
\end{split}
\end{equation}

The variations also have to satisfy the Bondi gauge:
\begin{equation}\label{5.7}
\begin{split}
&\delta g_{rr}=\delta g_{rA}=0,\\
&\partial_r\left(\frac{\det \left(g_{AB}+\delta g_{AB}\right)}{r^2}\right)=0.\\
\end{split}
\end{equation}

An asymptotically flat metric in Bondi coordinates can be expanded at null infinity as (choosing $x^A=(z,\bar{z})$, $z=\cot(\theta/2)\mathrm{e}^{\imath \varphi}$):
\begin{equation}\label{5.8}
\begin{split}
\diff s^2&=-\diff u^2-2 \diff r \diff u + r^2 \gamma_{z\bar{z}} \diff z \diff \bar{z}\\
&\phantom{=}+2 \ \frac{m_B}{r} \ \diff u^2+ r \ C_{zz} \ \diff z^2  +r \ C_{\bar{z}\bar{z}} \ \diff \bar{z}^2\\
&\phantom{=}+ \mathrm{D}^z C_{zz} \ \diff u \diff z+\mathrm{D}^{\bar{z}} C_{\bar{z}\bar{z}} \ \diff u \diff \bar{z}\\
&\phantom{=}+\frac{1}{r}\left[\frac{4}{3}\left(N_z+u\partial_z m_B\right)-\frac{1}{4}\partial_z\left(C_{zz}C^{zz}\right)\right] \diff u \diff z + c.c. + ...,\\
\end{split}
\end{equation}
where $\mathrm{D}_z$ is the covariant derivative with respect to the metric on the 2-sphere $\gamma_{z\bar{z}}$. The quantity $m_B$ is the Bondi mass aspect, whose integral over the sphere is the total Bondi mass, and $N_z$ is the angular momentum aspect, as its integration gives the total angular momentum. $C_{zz}$ describes the propagation of gravitational waves, and so does the tensor $\partial_u C_{zz}=N_{zz}$ which we call the Bondi news. This is the analogue to the electromagnetic field strength.

\section{Supertranslations}

In order to obtain the generators of the transformations which preserve that asymptotic structure, we need to find asymptotic solutions to Killing's equation. We have made, for the sake of simplicity, the following assumption: we restrict ourselves to diffeomorphisms with the large fall-off:
\begin{equation}\label{5.9}
\xi^{u},\xi^r \sim \mathcal{O}(1); \ \xi^{z},\xi^{\bar{z}} \sim \mathcal{O}\left(\frac{1}{r}\right).
\end{equation}

This assumption eliminates rotations and boosts, which grows with $r$ at infinity, as we can see from the generators of the Lorentz group (\ref{F.15}). The resulting equations are \footnote{See equation (5.2.2) of \cite{stro.lec}.}:
\begin{equation}\label{5.10}
\begin{split}
\pounds_\xi g_{ur}&=-\partial_u \xi^u+\mathcal{O}\left(\frac{1}{r}\right),\\
\pounds_\xi g_{zr}&=r^2 \gamma_{z\bar{z}} \ \partial_r \xi^{\bar{z}} -\partial_z\xi^u+\mathcal{O}\left(\frac{1}{r}\right),\\
\pounds_\xi g_{z\bar{z}}&=r^2 \gamma_{z\bar{z}}\left(2\xi^r+r\mathrm{D}_z\xi^z+r\mathrm{D}_{\bar{z}}\xi^{\bar{z}}\right)+\mathcal{O}\left(1\right),\\
\pounds_\xi g_{uu}&=-2\partial_u \xi^u-2\partial_u \xi^r+\mathcal{O}\left(\frac{1}{r}\right).\\
\end{split}
\end{equation}

The solutions to these equations are the following vector fields:
\begin{equation}\label{5.11}
\xi\equiv f\partial_u+ \mathrm{D}^z \mathrm{D}_zf \partial_r-\frac{1}{r} \left(\mathrm{D}^{z} f \partial_z+\mathrm{D}^{\bar{z}} f \partial_{\bar{z}}\right),
\end{equation}
where $f(z,\bar{z})$ is an arbitrary function. Therefore we have an infinite family of transformations. This vector field is the infinitesimal generator of supertranslations.

\begin{figure}[t]
	\centering
	\includegraphics[width=0.5\linewidth]{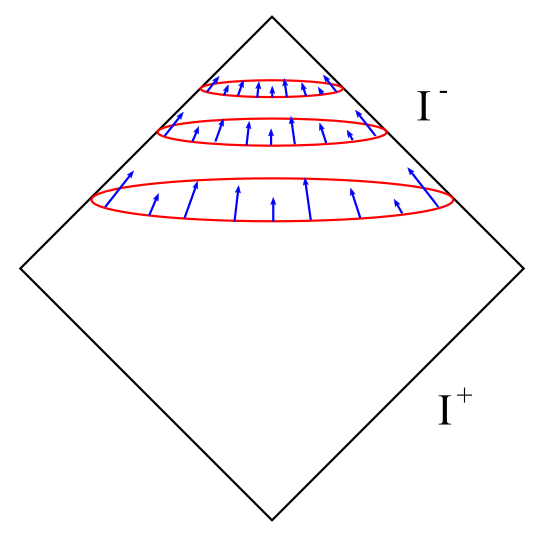}
	\caption{\textit{Supertranslations shift the retarded time $u$ in a different way at every angle.}}
	\label{Figure 5.2}
\end{figure}

The freedom to choose of $f$ allows to generate different translations along the null generators of $\mathrm{I}^+$. Imagine that we emit two light signals at the same time $u$ but from different angles of $\mathrm{I}^+$. If our supertranslations act in different manners depending on the angles, the data generated will be altered in a measurable way. The effects of supertranslations are discerned at a classical level. We can proceed in the same way at $I^-$, using the metric in the advanced Bondi coordinates.

\

Putting together the Lie derivatives (\ref{5.10}) with the solution, we can extract the action of the supertranslations on the data:
\begin{equation}\label{5.12}
\begin{split}
\pounds_f N_{zz}&=f\partial_u N_{zz},\\
\pounds_f m_B&=f\partial_u m_b+\frac{1}{4}\left(N^{zz}\mathrm{D}^2_z f+2\mathrm{D}_zN^{zz}\mathrm{D}_zf+c.c.\right),\\
\pounds_f C_{zz}&=f\partial_u C_{zz}-2\mathrm{D}^2_zf.\\
\end{split}
\end{equation}

We can now use the equations of motion in the form of the Einstein field equations:
\begin{equation}\label{5.13}
R_{\mu\nu}-\frac{1}{2}g_{\mu\nu}R=8\pi T_{\mu\nu}.
\end{equation}

Plugging in the previous expression the metric (\ref{5.8}) and expanding to large $r$, we get:
\begin{equation}\label{5.14}
\partial_u m_B=\frac{1}{4}\left[D_z^2N^{zz}+D_{\bar{z}}^2N^{\bar{z}\bar{z}}\right]-T_{uu}.
\end{equation}

This equation constrains the leading data at $\mathrm{I}^+$. For the null infinity $\mathrm{I}^-$ we can perform an analogue derivation. As the leading data at both $\mathrm{I}^+$ and $\mathrm{I}^-$ should obey some matching conditions (an infinite number of them as we have infinite choices of $f$). This implies the existence of conserved charges, given by:
\begin{equation}\label{5.15}
Q_f^+=\frac{1}{4\pi} \int_{\mathrm{I}^+_-} d^2 z \gamma_{z\bar{z}} f m_B,
\end{equation}
\begin{equation}\label{5.16}
Q_f^-=\frac{1}{4\pi} \int_{\mathrm{I}^-_+} d^2 z \gamma_{z\bar{z}} f m_B,
\end{equation}
where $\mathrm{I}^+_-$ is the region of $\mathrm{I}^+$ near $\mathrm{I}^-$ and $\mathrm{I}^-_+$ its analogous. The matching condition implies the conservation law $Q_f^+=Q_f^-$.

\section{Superrotations}

The fall-offs that we have assumed previously were highly restrictive. We have gotten rid of rotations and boost in the discussion of supertranslations. A generalized treatment of the previous leads to the appearance of superrotations. These are generated by an arbitrary vector field $Y^z$. An analogue matching condition for the angular momentum aspect $N_z$ could be obtain, which implies another infinite set of conserved charges \footnote{We have taken this result directly because our purpose is to study its consequences, rather than derive it. For a detailed treatment of the matter, check \textbf{Section 5.3} of \cite{stro.lec}.}:
\begin{equation}\label{5.17}
Q_Y^+=\frac{1}{8\pi} \int_{\mathrm{I}^+_-} d^2 z \left(Y_{\bar{z}}N_z+Y_zN_{\bar{z}}\right),
\end{equation}
\begin{equation}\label{5.18}
Q_Y^+=\frac{1}{8\pi} \int_{\mathrm{I}^-_+} d^2 z \left(Y_{\bar{z}}N_z+Y_zN_{\bar{z}}\right).
\end{equation}

The conservation of this superrotation charge is expressed as $Q_Y^+=Q_Y^-$.

\section{Gravitational memory effects}

Consider a region near null infinity $\mathrm{I}^+$ in which at some early and late times, $u_i$ and $u_f$ respectively, we have no Bondi news $N_{zz}$. In the in between, we may have gravitational waves passing through. From our assumption, we have for the retarded time $u_i$:
\begin{equation}\label{5.19}
M_B(u_i)=\mathrm{const.} \ C_{zz}(u_i)=0; \ N_{zz}(u_i)=0.
\end{equation}

At the late time $u_f$ the data preserve the form:
\begin{equation}\label{5.20}
M_B(u_f)=\mathrm{const.} \ C_{zz}(u_f)=0; \ N_{zz}(u_f)=0.
\end{equation}

We then notice that there is no retarded time nor energy flux dependence for the asymptotic data. So, the spacetime is the same up to a supertranslation, the radiation pulses passing through $\mathrm{I}^+$ changes the spacetime into another one, which is related to the former by a BMS transformation. Using (\ref{5.12}), for both times we will have:
\begin{equation}\label{5.21}
C_{zz}=-2D^2_z C.
\end{equation}

We can define the change in $C_{zz}$ and in the Bondi mass from initial to final retarded time as:
\begin{equation}\label{5.22}
\begin{split}
\Delta C_{zz}&=C_{zz}(u_f)-C_{zz}(u_i),\\
\Delta M_B&=M_B(u_f)-M_B(u_i).\\
\end{split}
\end{equation}

Now we can integrate (\ref{5.14}) with respect to $u$ to obtain:
\begin{equation}\label{5.23}
D_z^2 \Delta C^{zz}=2\Delta M_B + 2 \int_{u_i}^{u_f} \diff u T_{uu}.
\end{equation}

Using this last expression together with (\ref{5.17}), we can solve the differential equation upon finding an appropriate Green's function:
\begin{equation}\label{5.24}
\Delta C=\int \diff^2 z' \gamma_{z',\bar{z}'} \ G\left(z,\bar{z},z',\bar{z}'\right) \left(\Delta M_B + \int_{u_i}^{u_f} \diff u T_{uu}(z',\bar{z}')\right),
\end{equation}
with the following expression of the Green function:
\begin{equation}\label{5.25}
 G\left(z,\bar{z},z',\bar{z}'\right)=\frac{1}{\pi} \sin^2 \left(\frac{\Delta \theta}{2}\right) \log \sin^2 \left(\frac{\Delta \theta}{2}\right)
\end{equation}

This angle $\Delta \theta$ is the angle between the $(z,\bar{z})$ and the $(z',\bar{z}')$ points. Knowing this, we can deduce some important facts about how supertranlations generated by gravitational waves act (Figure \ref{Figure 5.3}). For instance, note that if the wave passes through the north pole ($\Delta \theta=0$), the effects of supertranslation vanish there and in the south pol e($\Delta \theta=\pi$), as long as they are larger at the equator ($\Delta \theta=\pi/2$). If we put a set of detectors near the region of study, the passage of gravitational radiation would displace them (depending on their position in the 2-sphere). This measurable phenomenon is known as the gravitational memory effect.

\begin{figure}[t]
	\centering
	\includegraphics[width=0.7\linewidth]{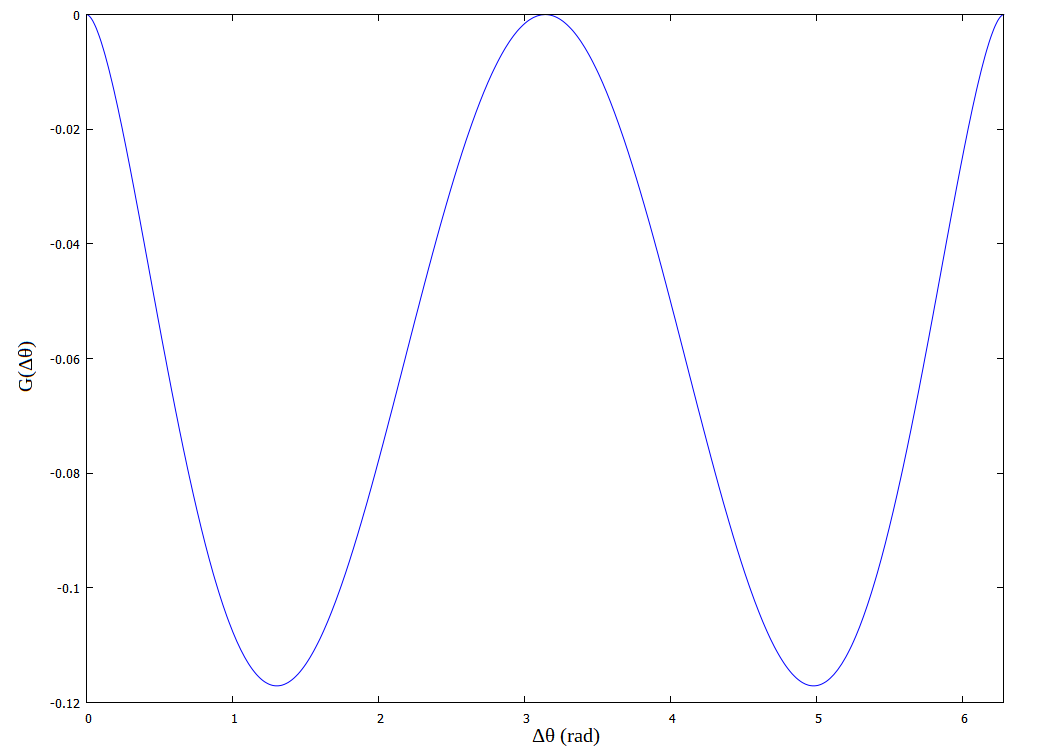}
	\caption{\textit{Green function as a function of the angle $\Delta \theta$ between the points $(z,\bar{z})$ and $(z',\bar{z}')$.}}
	\label{Figure 5.3}
\end{figure}

\begin{figure}[t]
	\centering
	\includegraphics[width=0.5\linewidth]{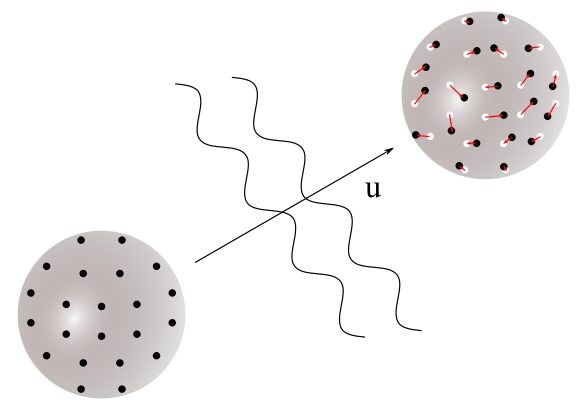}
	\caption{\textit{Schematic representation of the gravitational memory effect. After gravitational radiation passes through, the detectors (black dots) moves depending on its angular position.}}
	\label{Figure 5.4}
\end{figure}

\section{Soft hair}

Those who argument that information is lost in the formation/evaporation process of a black hole assume that such a static black hole is basically ``bald'', i.e., it is completely characterized by three parameters: its mass $M$, charge $Q$ and angular momentum $J$. This is known as the \textit{no-hair theorem}, and states that the static black hole is fully determined by these three quantities, up to diffeomorphisms. But as we have seen, the BMS group transformations, such as supertranslations, change our spacetime to a physically inequivalent one.

\

In the Hawking process, as a consequence of the conservation of charges, the sum of supertranslation charge of both black hole and Hawking quanta is fixed in the whole process. This forces the black hole to carry some \textbf{soft hair} which arises from supertrasnslations. Other symmetries, as superrotations, will lead to other kinds of hair. Moreover, as we have infinite families of supertransformations, and no favorite ones among them, we will have an infinite number of soft hairs.

\

As the Hawking radiation will carry supertransformation charges across null infinity $\mathrm{I}^+$, exact conservation of the charges requires that the black hole decreases its charges in the same amount. This enforces infinite correlations between the state of the outgoing Hawking radiation and the state of the black hole.

\

As a first approach to the problem, we can study the supertranslation hair on a classical Schwarzschild black hole. Let us consider the Schwarzschild metric in the advanced Bondi coordinates:
\begin{equation}\label{5.26}
\diff s^2=-\left(1-\frac{2m_B}{r}\right)\diff v^2+2\diff v\diff r+ r^2 \gamma_{z\bar{z}} \diff z \diff \bar{z}.
\end{equation}

We may apply the supertranslation generator (\ref{7.11}) in the form:
\begin{equation}\label{5.27}
\zeta= f\partial_u+ \frac{1}{r} \mathrm{D}^{z} f \partial_z -\frac{1}{2} \mathrm{D}^2f \partial_r.
\end{equation}

Now, choosing at linear level $\delta_\zeta g_{\mu\nu}=\pounds_\zeta g_{\mu\nu}$, we have to calculate the non-zero components of the Lie derivative of the metric along $\zeta$:
\begin{equation}\label{5.28}
\begin{split}
\pounds_\zeta g_{\mu\nu}&=\zeta^\beta_{\ ;\mu} g_{\beta\nu} + \zeta^\beta_{\ ;\nu} g_{\mu\beta}\\
&=\left(\zeta^\beta_{\ ,\mu}+\Gamma^{\beta}_{\ \gamma\mu} \zeta^{\gamma}\right) g_{\beta\nu} + \left(\zeta^\beta_{\ ,\nu}+\Gamma^{\beta}_{\ \gamma\nu} \zeta^{\gamma}\right) g_{\mu\beta}.\\
\end{split}
\end{equation}

In order to do so it is convenient to consult a catalogue of spacetimes as \cite{catalogue}, to directly use the expressions of the Christoffel symbols. This way one gets:
\begin{equation}\label{5.29}
\begin{split}
\pounds_{\zeta}g_{vv}&=2\left(\zeta^\beta_{\ ,v} +\Gamma^{\beta}_{\ \gamma v} \zeta^\gamma\right)g_{\beta v}\\
&=2\left[\left(\zeta^v_{\ ,v} +\Gamma^{v}_{\ \gamma v} \zeta^\gamma\right)g_{vv}+\left(\zeta^r_{\ ,v} +\Gamma^{r}_{\ \gamma v} \zeta^\gamma\right)g_{rv}\right]\\
&=2\left[\Gamma^{v}_{\ v v} \zeta^v g_{vv}+\left(\Gamma^{r}_{\ v v} \zeta^v+\Gamma^{r}_{\ r v} \zeta^r\right)g_{rv}\right]\\
&=\frac{2m_B}{r^2} \ f g_{vv}+2\left(\frac{m_B(r-2m_B)}{r^3} \ f+\frac{m_B}{2r^2} \ D^2 f\right) g_{rv}\\
&=\frac{m_B}{r^2} \ D^2 f,
\end{split}
\end{equation}
\begin{equation}\label{5.30}
\begin{split}
\pounds_{\zeta}g_{vz}&=\left(\zeta^\beta_{\ ,v} +\Gamma^{\beta}_{\ \gamma v} \zeta^\gamma\right)g_{\beta z}+\left(\zeta^\beta_{\ ,z} +\Gamma^{\beta}_{\ \gamma z} \zeta^\gamma\right)g_{v \beta}\\
&=\left(\zeta^{\bar{z}}_{\ ,v} +\Gamma^{\bar{z}}_{\ \gamma v} \zeta^\gamma\right)g_{\bar{z} z}+\left(\zeta^v_{\ ,z} +\Gamma^{v}_{\ \gamma z} \zeta^\gamma\right)g_{v v}+\left(\zeta^r_{\ ,z} +\Gamma^{r}_{\ \gamma z} \zeta^\gamma\right)g_{v r}\\
&=\left(\zeta^v_{\ ,z} +\Gamma^{v}_{\ z z} \zeta^z\right)g_{v v}+\left(\zeta^r_{\ ,z} +\Gamma^{r}_{\ z z} \zeta^z\right)g_{v r}\\
&=\left(D_z f - r\gamma_{zz} \zeta^z\right)g_{vv}+ \left(-\frac{1}{2} \ D_z D^2 f -r \gamma_{zz} \left(1-\frac{2m_B}{r}\right)\zeta^z\right)g_{vr}\\
&=-D_z f \left(1-\frac{2m_B}{r}\right) - \frac{1}{2} \ D_z D^2 f,\\
\end{split}
\end{equation}
\begin{equation}\label{5.31}
\begin{split}
\pounds_{\zeta}g_{z\bar{z}}&=\left(\zeta^\beta_{\ ,z} +\Gamma^{\beta}_{\ \gamma z} \zeta^\gamma\right)g_{\beta \bar{z}}+\left(\zeta^\beta_{\ ,\bar{z}} +\Gamma^{\beta}_{\ \gamma \bar{z}} \zeta^\gamma\right)g_{z\beta}\\
&=2\left(\zeta^z_{\ ,z} +\Gamma^{z}_{\ \gamma z} \zeta^\gamma\right)g_{z \bar{z}}+2\left(\zeta^{\bar{z}}_{\ ,z} +\Gamma^{\bar{z}}_{\ \gamma z} \zeta^\gamma\right)g_{\bar{z} \bar{z}}\\
&=2\left(\zeta^z_{\ ,z} +\Gamma^{z}_{\ r z} \zeta^r+\Gamma^{z}_{\ \bar{z} z} \zeta^{\bar{z}}\right)g_{z \bar{z}}+2\left(\zeta^{\bar{z}}_{\ , z}+\Gamma^{\bar{z}}_{\ zz}\zeta^z\right)g_{\bar{z}\bar{z}}\\
&=2\left(\frac{1}{r} \ D^2 f -\frac{1}{2r} \ D^2 f + \frac{1}{r} \gamma^{-1}_{z\bar{z}} D^z f\right)g_{z\bar{z}}\\
&=r D^2 f \gamma_{z\bar{z}} + 2r D_z D_{\bar{z}} f.\\
\end{split}
\end{equation}

The resulting geometry describing a black hole with linearized supertranslation hair takes the form:
\begin{equation}\label{5.32}
\begin{split}
\diff s^2 + \pounds_{\zeta} \diff s^2 &= -\left(1-\frac{2m_B}{r}-\frac{m_B}{r^2} \ D^2 f\right) \diff v^2+2\diff v \diff r\\
&\phantom{=}- D_z \left(2f \left(1-\frac{2m_B}{r}\right) + D^2 f\right) \diff v \diff z\\
&\phantom{=}+\left(r^2\gamma_{z\bar{z}}+2rD_zD_{\bar{z}} f+ r\gamma_{z\bar{z}} D^2 f\right) \diff z \diff \bar{z}.\\
\end{split}
\end{equation}

The event horizon is located at $r=2m_B+\frac{1}{2} \ D^2 f$ \footnote{In \cite{stro.lec}, places the event horizon at that location, but I don't see clearly why. That coordinate results from supertranslate the original event horizon applying $\zeta^r$, but does not make the new $g_{vv}$ vanish. The one which actually does it is $r=m_B\pm\sqrt{m_B^2+m_B D^2 f}$.}. This supertransformation does not add supertranslation charge to the black hole, as a common translation does not add linear momentum to a solution. But, as supertranslations do not commute with superrotations, the supertranslated black hole carries superrotation charge. Under a supertranslation, by comparing (\ref{5.32}) with (\ref{5.8}), the angular momentum aspect varies as \footnote{In \cite{stro.lec}, this result is presented by making no more considerations, but it seems a bit \textit{ad hoc}.}:
\begin{equation}\label{5.33}
\begin{split}
\delta_\zeta N_z=-3m_B \partial_z f.
\end{split}
\end{equation}

Thus, the superrotations charge which carries the supertranslated black hole will be, for a generic vector field $Y_z$:
\begin{equation}\label{5.34}
Q^{-}_Y(g,\delta_\zeta g)=-\frac{3}{8\pi} \int_{I^-_+} \diff^2 z \sqrt{\gamma} \ Y^{z} m_B \partial_z f.
\end{equation}

By imposing different choices of $f$ we can add an infinite number of charges to the black hole. This way, we can see that a classical black hole carries an infinite set of supertranslation hair.

\

In this chapter we have briefly introduced the BMS symmetry group, and studied its generators and associated charges. We have discussed the gravitational memory effect, and how it is related to the supertransformations. Finally, we have computed a supertranslation on a Schwarzschild back hole and noted that supertranslated black holes carry superrotation charge. Black holes carry an infinite number of charges that are conserved.
\chapter{Conclusion and Outlook}

In this work we have studied the information loss paradox in detail. To do so, we have learned QFT in curved spacetimes, amongst other things. Finally, we have considered the BMS symmetry group, as it is one of the most novel proposals that aim to solve this conjecture.

\

In chapter \ref{Chapter 1} we have contextualized motivated the topics that have been studied in the thesis. We have introduced them in a logical way, starting from the basics of GR and QM and finishing with the information loss paradox and BMS symmetries as a plausible solution.

\

In chapter \ref{Chapter 2} we have made an analysis of quantum field theory in curved spacetime. We studied the properties of the Klein-Gordon equation and its solutions in a curved background. We have expanded the solutions in orthonormal modes, introducing the creation and annihilation operators. This has led us to introduce the Bogoliubov transformation. This is a key concept, as from it we can predict the creation of particles in non-stationary spacetimes. Its application to accelerated motion in Minkowski space reveals the existence of the Unruh effect.

\

In chapter \ref{Chapter 3} we have introduced concepts as surface gravity or redshift factor. We have seen that the Unruh effect is closely related to the Hawking effect of black holes. A straightforward computation, involving the redshift factor, relates both effects. Returning to the idea of particle creation in non-stationary spacetimes, we model the black hole formation with the simplest Vaidya spacetime. It consists of the collapse of a single shock wave, located at some time. Taking the Klein-Gordon equation into spherical coordinates, for both Minkowski and Schwarzschild regions, and applying the formalism of the Bogoliubov transformation we re-obtain the result of Hawking radiation. Finally, we performed an estimation of the time that it takes the evaporation of a black hole.

\

In chapter \ref{Chapter 4} we turned ourselves to consider the quantum-mechanical nature of the states of the radiation quanta. We conclude that, in order to preserve the niceness conditions which ensure a low-energy limit where we can neglect the effects of quantum gravity, the state of the Hawking quanta plus the matter in the hole admit only small deviations from the leading order state. We proved that for this type of state, the entropy of entanglement always increases. This allow us to prove as theorem which states that the process of formation and evaporation of black holes leads unavoidably to mixed states/remnants. We have seen that this conclusion leads also to an information loss, with respect to the matter state that originates the black hole. This is the so-called information loss paradox.

\

In chapter \ref{Chapter 5} we have explored a possible solution to this problem. The BMS symmetry group consists of transformations which leave unchanged the asymptotically-flat behaviour of the metric. We followed the main steps in order to obtain the generators of the supertranslations. We also presented the charges associated with supertranslations and superrotations. We have studied the known as gravitational memory effect, associated with the passage of gravitational radiation. The invariance under supertransformations forces black holes to carry some sort of soft hair associated which these. As a simple example, we compute the effect of a supertranslation on the Schwarzschild metric, and noted that the supertranslated black hole carries some superrotation charge. The infinite correlations arising between the radiation quanta and the black hole are one of the most popular solutions to the information problem.

\

The relation between the gravitational memory effect, Weinberg's soft graviton theorem and asymptotic symmetries is possibly the fundamental piece to solve the paradox. A possible solution can be found in correlations between the final vacuum state of the evaporation process and/or outgoing soft particles and the Hawking radiation quanta. Thus, a good understanding of this soft modes (which will be produced in every scattering process) could lead to solve the information paradox, as they can carry that information \cite{stro.lec}, \cite{lust.bms}.

\

Anyway, a proper solution to the paradox must fulfil two conditions: (i) it should reproduce the Bekenstein entropy (\ref{3.1}) microscopically and (ii) it must enable us to perform explicit computations involving the information of black holes. Despite this first attempt to combine gravity with Quantum Mechanics has led us to a problem which endangers the very base of our theories, we can consider it a triumph. The efforts that the scientific community has made to solve the paradox have given multitude of new theories and possibilities. It is a huge problem, and thus a huge motivation to keep unravelling the ultimate understanding of Nature.

\

Some aspects of the thesis are not fully developed, specially the calculations involving the BMS symmetry group. This is because these are part of the hottest topics of current research in theoretical physics. This dissertation may be seen as a first step into a research career, that we intend to carry out in the next years.

\begin{appendices}
\chapter{Classical Field Theory}\label{Appendix A}

\section{Lagrangian formulation of field theory}

Consider a $n$ dimensional spacetime, $R$ a bounded region enclosed by the hypersurface $B=\partial R$. We want to study the dynamics of certain fields $\psi_{A}(x)$, $A=1,...,f$, knowing its values in the boundary.

\

Functions $\psi_{A}(x)$, $A=1,...,f$ satisfy Hamilton's Principle. The action functional defined as:
\begin{equation}\label{A.1}
S=\int_{R}\mathcal{L}\left(x,\psi(x),\frac{\partial \psi(x)}{\partial x}\right) \ \diff^{n}x,
\end{equation}
is stationary and the linera function $\mathcal{L}\left(x,\psi(x),\frac{\partial \psi(x)}{\partial x}\right)$ is the Lagrangian density.

\

We can take $\delta \psi_{A}(x)=0 \ \forall \ x\in B$. This way, variation in $S$ results:
\begin{equation}\label{A.2}
\begin{split}
\delta S&=\int_{R} \left(\frac{\partial \lag}{\partial \psi_{A}} \ \delta \psi_{A} + \frac{\partial \lag}{\partial \psi_{A,\alpha}} \ \delta \psi_{A,\alpha}\right) \ \diff^{n}x\\
&=\int_{R} \left(\frac{\partial \lag}{\partial \psi_{A}} \ \delta \psi_{A} + \frac{\partial}{\partial x^{\alpha}} \left[\frac{\partial \lag}{\partial \psi_{A,\alpha}} \ \delta \psi_{A}\right] - \frac{\partial}{\partial x^{\alpha}} \left(\frac{\partial \lag}{\partial \psi_{A,\alpha}}\right) \delta \psi_{A}\right) \ \diff^{n}x\\
&=\int_{R} \left[\frac{\partial \lag}{\partial \psi_{A}} - \frac{\partial}{\partial x^{\alpha}} \left(\frac{\partial \lag}{\partial \psi_{A,\alpha}}\right) \right]\ \delta \psi_{A} \ \diff^{n}x + \int_{R} \frac{\partial}{\partial x^{\alpha}} \left[\frac{\partial \lag}{\partial \psi_{A,\alpha}} \ \delta \psi_{A}\right] \ \diff^{n}x\\
&=\int_{R} \left[\frac{\partial \lag}{\partial \psi_{A}} - \frac{\partial}{\partial x^{\alpha}} \left(\frac{\partial \lag}{\partial \psi_{A,\alpha}}\right) \right]\ \delta \psi_{A} \ \diff^{n}x + \int_{B} \frac{\partial \lag}{\partial \psi_{A,\alpha}} \ \delta \psi_{A} \ \diff B.\\
\end{split}
\end{equation}

The last integral vanishes  because the boundary conditions imposed. Hamilton's Principle implies that this variation must be zero. Because the variations are completely arbitrary, it has to be:
\begin{equation}\label{A.3}
\frac{\partial \lag}{\partial \psi_{A}} - \frac{\partial}{\partial x^{\alpha}} \left(\frac{\partial \lag}{\partial \psi_{A,\alpha}}\right) = 0.
\end{equation}

These relations are the well-known Euler-Lagrange equations.

\

If the Lagrangian density depends on derivatives of the fields of superior order, the dynamical equations would be of order $>2$ and it'd be necessary to know more than just $\psi_{A}(x)$ at $B$ to determine the fields in $R$.

\

Lagrangian density is undetermined in the addition of the divergence of an arbitrary function $\lambda(x,\psi(x))$. Taking
\begin{equation}\label{A.4}
\lag'=\lag+\frac{\partial \lambda^\alpha(x,\psi(x))}{\partial x^\alpha},
\end{equation}
we can define:
\begin{equation}
S'=S+\int_{R}\frac{\partial \lambda^\alpha(x,\psi(x))}{\partial x^\alpha} \ \diff^n x=S+\int_{B}\lambda^\alpha(x,\psi(x)) \ \diff B_\alpha.
\end{equation}

This way $\delta S'=\delta S=0$ for the real fields and because of that this new Lagrangian density will lead us to the same dynamical equations.

\

Two sets of fields $\psi_A$ and $\chi_B$, with Lagrangian densities$\lag_1$ and $\lag_2$ generate the new Lagrangian density $\lag_0=\lag_1(\psi)+\lag_2(\chi)$ whose Euler-Lagrange equations are just the sum of the ones of $\lag_1$ and $\lag_2$.

\

Interaction between fields may be generated by a coupling term $\lag_I(\psi,\chi)$, resulting a Lagrangian density $\lag=\lag_1(\psi)+\lag_2(\chi)+\lag_I(\psi,\chi)$. We had assumed that Lagrangian densities doesn't depend on position (isolated system) and because of that it remains invariant under spacetime translation.

\section{Functional derivative}

A functional is said to be a function which takes another function as its argument. In variational calculus, the integrand to be minimized is a functional of certain unknown function which must satisfy some boundary conditions. So, functional derivative is a generalization of usual derivative, in this case the functional is differentiate about a function.

\

Action is a functional over $\psi_A(x)$:
\begin{equation}\label{A.5}
\begin{split}
S: &V \rightarrow \mathbb{R}\\
&\psi_{A}\mapsto S[\psi_{A}],
\end{split}
\end{equation}
and it's defined by:
\begin{equation}\label{A.6}
S[\psi_{A}]=\int_{R}\mathcal{L}\left(x,\psi(x),\frac{\partial \psi(x)}{\partial x},\frac{\partial^2 \psi(x)}{\partial x^2},...\right) \ \diff^{n}x.
\end{equation}

Under infinitesimal variations $\delta \psi_{A}(x)$, action changes like (assuming negligible variation in coordinates):
\begin{equation}\label{A.7}
\begin{split}
\delta S&=S[\psi_A+\delta\psi_A]-S[\psi_A]=\int_{R}\delta\mathcal{L} \ \diff^{n}x\\
&=\int_{R} \left[\frac{\partial \lag}{\partial \psi_{A}} \ \delta\psi_{A}+\frac{\partial \lag}{\partial \psi_{A,\alpha}} \ \delta\psi_{A,\alpha}+\frac{\partial \lag}{\partial \psi_{A,\alpha\beta}} \ \delta\psi_{A,\alpha\beta}+...\right] \ \diff^n x.
\end{split}
\end{equation}

The second term of the last integral reduce to:
\begin{equation}\label{A.8}
\frac{\partial \lag}{\partial \psi_{A,\alpha}} \ \delta\psi_{A,\alpha}=\frac{\partial}{\partial x^\alpha} \left(\frac{\partial \lag}{\partial \psi_{A,\alpha}} \ \delta\psi_{A}\right) - \frac{\partial}{\partial x^\alpha} \left(\frac{\partial \lag}{\partial \psi_{A,\alpha}}\right) \ \delta\psi_{A},
\end{equation}
and the third:
\begin{equation}\label{A.9}
\begin{split}
\frac{\partial \lag}{\partial \psi_{A,\alpha\beta}} \ \delta\psi_{A,\alpha\beta}&=\frac{\partial}{\partial x^\beta} \left(\frac{\partial \lag}{\partial \psi_{A,\alpha\beta}} \ \delta\psi_{A,\alpha}\right)-\frac{\partial}{\partial x^\beta} \left(\frac{\partial \lag}{\partial \psi_{A,\alpha\beta}}\right) \ \delta\psi_{A,\alpha}\\
&=\frac{\partial}{\partial x^\beta} \left(\frac{\partial \lag}{\partial \psi_{A,\alpha\beta}} \ \delta\psi_{A,\alpha}\right)-\frac{\partial}{\partial x^\alpha}\left[\frac{\partial}{\partial x^\beta} \left(\frac{\partial \lag}{\partial \psi_{A,\alpha\beta}} \ \delta\psi_{A}\right)\right]\\
&+\frac{\partial}{\partial x^\alpha}\left[\frac{\partial}{\partial x^\beta} \left(\frac{\partial \lag}{\partial \psi_{A,\alpha\beta}}\right)\right] \ \delta\psi_{A}.\\
\end{split}
\end{equation}

We can rewrite the variation of the action as:
\begin{equation}\label{A.10}
\delta S = \int_{R} \frac{\delta S}{\delta \psi_{A}} \ \delta \psi_A + \frac{\partial}{\partial x^\alpha} \left[\left(\frac{\partial \lag}{\partial \psi_{A,\alpha}} - \frac{\partial}{\partial x^\beta} \left(\frac{\partial \lag}{\partial \psi_{A,\alpha\beta}}\right)\right) \ \delta \psi_A + \frac{\partial \lag}{\partial \psi_{A,\alpha\beta}} \ \delta\psi_{A,\beta}+...\right] \ \diff^n x.
\end{equation}

We've defined the functional derivative of the action with respect to the field as:
\begin{equation}\label{A.11}
\frac{\delta S}{\delta \psi_{A}} = \frac{\partial \lag}{\partial \psi_{A}}-\frac{\partial}{\partial x^\alpha} \left(\frac{\partial \lag}{\partial \psi_{A,\alpha}}\right)+\frac{\partial}{\partial x^\alpha}\left[\frac{\partial}{\partial x^\beta} \left(\frac{\partial \lag}{\partial \psi_{A,\alpha\beta}}\right)\right]+...
\end{equation}

Using Gauss's Theorem, the integral of the total derivative becomes a surface integral:
\begin{equation}\label{A.12}
\delta S = \int_{R} \frac{\delta S}{\delta \psi_{A}} \ \delta \psi_A + \int_{B} \left[\left(\frac{\partial \lag}{\partial \psi_{A,\alpha}} - \frac{\partial}{\partial x^\beta} \left(\frac{\partial \lag}{\partial \psi_{A,\alpha\beta}}\right)\right) \ \delta \psi_A + \frac{\partial \lag}{\partial \psi_{A,\alpha\beta}} \ \delta\psi_{A,\beta}+...\right] \ \diff B_\alpha.
\end{equation}

If Lagrangian density does not depend on derivatives of superior order, imposing $\delta \psi_A(x)=0 \ \forall x\in B$ the surface term vanishes. Applying Hamilton's Principle, we obtain the usual Euler-Lagrange equations:
\begin{equation}\label{A.13}
\frac{\delta S}{\delta \psi_{A}} = \frac{\partial \lag}{\partial \psi_{A}}-\frac{\partial}{\partial x^\alpha} \left(\frac{\partial \lag}{\partial \psi_{A,\alpha}}\right)=0.
\end{equation}

If the Lagrangian density depends on derivatives of the fields of superior order, we must impose boundary conditions  to the derivatives of the variations of the field, or introduce boundary terms in the action that cancel the ones appearing in the variation of $S$. In this case, the equations of motion are:
\begin{equation}\label{A.14}
\frac{\delta S}{\delta \psi_{A}} = \frac{\partial \lag}{\partial \psi_{A}}-\frac{\partial}{\partial x^\alpha} \left(\frac{\partial \lag}{\partial \psi_{A,\alpha}}\right)+\frac{\partial}{\partial x^\alpha}\left[\frac{\partial}{\partial x^\beta} \left(\frac{\partial \lag}{\partial \psi_{A,\alpha\beta}}\right)\right]+...=0.
\end{equation}

This differential equations are of order $>2$, and for solving them we need to introduce initial values for the fields and its derivatives.

\

The Lagrangian density is undetermined by the addition of the divergence of an arbitrary function of the fields $\lambda(\psi_A)$, with the restriction $\lambda(\psi_A)=0$ over the boundary.

\section{Noether's theorem}

Consider a continuous transformation of the coordinates and the fields:
\begin{equation}\label{A.15}
x \rightarrow x' ; \ \psi_{A}(x) \rightarrow \psi'_{A}(x'), \ A=1,...,f.
\end{equation}

A transformation like this is called a symmetry if:
\begin{equation}\label{A.16}
\begin{split}
\Delta S= S'-S&=\int_{R}\lag\left(x',\psi'_{A}(x'),\frac{\partial \psi'_{A}(x')}{\partial x'}\right) \ \diff^n x'-\int_{R}\lag\left(x,\psi_{A}(x),\frac{\partial \psi_{A}(x)}{\partial x}\right) \ \diff^n x\\
&=\int_{R}\frac{\partial \lambda^{\alpha}(x,\psi_{A}(x))}{\partial x^\alpha} \ \diff^n x,
\end{split}
\end{equation}
where $\lambda(x,\psi_{A}(x))$ is an arbitrary function of position and the fields and $\lag\left(x',\psi'_{A}(x'),\frac{\partial \psi'_{A}(x')}{\partial x'}\right)$ has the same functional dependence as the original Lagrangian density but expressed in the new variables.

\

Hamilton's Principle gives us $\delta S=\delta S'=0$. If $\psi_A(x)$ are solutions to the equations of motion, $\psi'_A(x')$ would be too but transformed. The change of the variables and the local change of the fields may be write as:
\begin{equation}\label{A.17}
\begin{split}
&x'^\alpha=x^\alpha+\delta x^\alpha,\\
&\psi'_A(x)=\psi_A(x)+\bar{\delta}\psi_A(x).\\
\end{split}
\end{equation}

To first order, the change of the fields is:
\begin{equation}\label{A.18}
\begin{split}
\delta \psi_A(x)&=\psi'_A(x')-\psi_A(x)=\psi'_A(x')-\psi_A'(x)+\psi'_A(x)-\psi_A(x)\\
&=\psi'_A(x')-\psi_A'(x)+\bar{\delta}\psi_A(x)\simeq\psi'_A(x)+\psi'_{A,\alpha}(x) \ \delta x^\alpha-\psi_A'(x)+\bar{\delta}\psi_A(x)\\
&=\psi_{A,\alpha}(x) \ \delta x^\alpha+\cancelto{0}{\bar{\delta}\psi_{A,\alpha}(x) \ \delta x^\alpha}+\bar{\delta}\psi_A(x)\\
&=\psi_{A,\alpha}(x) \ \delta x^\alpha+\bar{\delta}\psi_A(x),
\end{split}
\end{equation}
where we have employed the Taylor's series expansion of $\psi'_A(x')=\psi'_A(x+\delta x)$ around $0$, keeping only terms to first order of $\delta x^\alpha$. It's directional derivative is calculated as:
\begin{equation}\label{A.19}
\begin{split}
\delta  \psi_{A,\alpha}(x)&=\frac{\partial \psi'_{A}(x')}{\partial x'^\alpha}-\psi_{A,\alpha}(x)=\frac{\partial}{\partial x'^\alpha} \left[\psi_A(x)+\bar{\delta}\psi_{A}(x)+\psi_{A,\gamma}(x) \ \delta x^{\gamma}\right]-\psi_{A,\alpha}(x)\\
&=\frac{\partial x^\beta}{\partial x'^\alpha}\left[\psi_{A,\beta}(x)+\bar{\delta}\psi_{A,\beta}(x)+\psi_{A,\gamma \beta}(x) \ \delta x^{\gamma}+\psi_{A,\gamma}(x) \ \delta x^{\gamma}_{,\beta}\right]-\psi_{A,\alpha}(x).\\
\end{split}
\end{equation}

We calculate separately the partial derivative of the coordinates. First we compute:
\begin{equation}\label{A.20}
\frac{\partial x'^\alpha}{\partial x^\beta}=\frac{\partial}{\partial x^\beta}\left(x^\alpha+\delta x^\alpha\right)=\delta^{\alpha}_{\beta}+\delta x^{\alpha}_{,\beta},
\end{equation}
and now take the inverse:
\begin{equation}\label{A.21}
\frac{\partial x^\beta}{\partial x'^\alpha}=\left(\frac{\partial x'^\alpha}{\partial x^\beta}\right)^{-1}=\left(\delta^{\alpha}_{\beta}+\delta x^{\alpha}_{,\beta}\right)^{-1}\simeq \delta^{\beta}_{\alpha}-\delta x^{\beta}_{,\alpha}.
\end{equation}

Taking all together, neglecting terms of superior order in $\delta$:
\begin{equation}\label{A.22}
\begin{split}
\delta  \psi_{A,\alpha}(x)&\simeq\left(\delta^{\beta}_{\alpha}-\delta x^{\beta}_{,\alpha}\right)\left[\psi_{A,\beta}(x)+\bar{\delta}\psi_{A,\beta}(x)+\psi_{A,\gamma \beta}(x) \ \delta x^{\gamma}+\psi_{A,\gamma}(x) \ \delta x^{\gamma}_{,\beta}\right]-\psi_{A,\alpha}(x)\\
&\simeq \psi_{A,\alpha}(x)-\psi_{A,\beta}(x) \ \delta x^{\beta}_{,\alpha}+\bar{\delta}\psi_{A,\alpha}(x)+\psi_{A,\gamma \alpha}(x) \ \delta x^{\gamma}+\psi_{A,\gamma}(x) \ \delta x^{\gamma}_{,\alpha}-\psi_{A,\alpha}(x)\\
&=\bar{\delta}\psi_{A,\alpha}(x)+\psi_{A,\alpha \beta}(x) \ \delta x^{\beta}.
\end{split}
\end{equation}

Now, the Jacobian of the change of variable is:
\begin{equation}\label{A.23}
\mathrm{det} \left(\frac{\partial x'^\alpha}{\partial x^\beta}\right)=\mathrm{det} \left(\delta^{\alpha}_{\beta}+\delta x^{\alpha}_{,\beta}\right)=1+\delta x^{\alpha}_{,\alpha}+\mathcal{O}\left(\left(\delta x^{\alpha}\right)^2\right).
\end{equation}

For a better perform of the calculation, we evaluate first the difference of the Lagrangian densities:
\begin{equation}\label{A.24}
\begin{split}
&\lag\left(x',\psi'_{A}(x'),\frac{\partial \psi'_{A}(x')}{\partial x'}\right)-\lag\left(x,\psi_{A}(x),\frac{\partial \psi_{A}(x)}{\partial x}\right)=\frac{\partial \lag}{\partial \psi_A} \ \left(\psi_{A,\alpha} \ \delta x^\alpha + \bar{\delta}\psi_A\right)\\
&+\frac{\partial \lag}{\partial \psi_{A,\alpha}} \ \left(\bar{\delta}\psi_{A,\alpha}(x)+\psi_{A,\alpha \beta}(x) \ \delta x^{\beta}\right)+\frac{\partial \lag}{\partial x^\alpha} \ \delta x^\alpha=\frac{\partial \lag}{\partial \psi_A} \ \bar{\delta}\psi_A+\frac{\partial \lag}{\partial \psi_{A,\alpha}} \ \bar{\delta}\psi_{A,\alpha}(x)\\
&+\frac{\partial \lag}{\partial x^\alpha} \ \delta x^\alpha.
\end{split}
\end{equation}

Then, the difference of the action is:
\begin{equation}\label{A.25}
\begin{split}
\Delta S&=\int_{R} \ \diff^nx \left(\frac{\delta S}{\delta \psi_A} \ \bar{\delta}\psi_A + \frac{\diff}{\diff x^\alpha}\left[\frac{\partial \lag}{\partial \psi_{A,\alpha}} \ \bar{\delta}\psi_A\right]+\frac{\partial \lag}{\partial x^\alpha} \ \delta x^\alpha + \frac{\partial \delta x^{\alpha}}{\partial x^\alpha} \ \lag\right)\\
&=\int_{R} \ \diff^nx \left(\frac{\delta S}{\delta \psi_A} \ \bar{\delta}\psi_A + \frac{\diff}{\diff x^\alpha}\left[\frac{\partial \lag}{\partial \psi_{A,\alpha}} \ \bar{\delta}\psi_A+\lag \delta x^\alpha\right]\right)=\int_{R} \ \diff^nx \ \frac{\diff \lambda^\alpha(x,\psi_A)}{\diff x^\alpha}.
\end{split}
\end{equation}

The last equality only holds when the transformation is a symmetry. If the fields $\psi_A$ are the true ones, then the equations of Euler-Lagrange are satisfied and we can write:
\begin{equation}\label{A.26}
\frac{\diff}{\diff x^\alpha}\left[\frac{\partial \lag}{\partial \psi_{A,\alpha}} \ \bar{\delta}\psi_A+\lag \delta x^\alpha\right]=\frac{\diff \lambda^\alpha(x,\psi_A)}{\diff x^\alpha}.
\end{equation}

\chapter{Lagrangian formulation of General Relativity}\label{Appendix B}

The action functional of General Relativity contains a contribution from the gravitational field and another contribution from matter fields. We can denote them as:
\begin{equation}\label{B.1}
S_{GR} \left[g,\phi\right]=S_{G}\left[g\right] + S_{M} \left[\phi;g\right].
\end{equation}

In turn, the gravitational action contains three terms. The Einstein-Hilbert term, a boundary term and a nondynamical term (affecting numerical value of the action but not the equations of motion we are interested). Explicitly, we write:
\begin{equation}\label{B.2}
S_{G}\left[g\right] = S_{EH} \left[g\right] + S_{B} \left[g\right] - S_0.
\end{equation}

The Einstein field equations:
\begin{equation}\label{B.3}
G_{\mu \nu} = 8 \pi T_{\mu \nu},
\end{equation}
are recovered varying the action with respect to the metric $g_{\mu \nu}$, restricting the variations to the condition:
\begin{equation}\label{B.4}
\left.\delta g_{\mu \nu}\right|_{\partial \mathcal{V}}=0,
\end{equation}
where $\mathcal{V}$ is the region of the spacetime manifold we are integrating over, bounded by the closed hypersurface $\partial \mathcal{V}$.

\

We are going to study the terms of the action separately.

\section{The Einstein-Hilbert action}

Using natural units where $G=c=1$ (if we don't say anything this is going to be assumed), the Einstein-Hilbert action in $n$ dimensions is:
\begin{equation}\label{B.5}
S_{EH} \left[g\right]=\frac{1}{16 \pi} \ \int_{\mathcal{V}} \diff^n x \sqrt{-g} R.
\end{equation}

There, $R$ is the Ricci scalar of the metric $g_{\mu \nu}$ in the region $\mathcal{V}$ we are integrating over. We are going to vary this action with respect $g_{\mu \nu}$. Since we have a product of two quantities that depends on the metric, we must recall the Leibniz rule of differentiation.

\

The Ricci scalar is just the contraction of the Ricci tensor:
\begin{equation}\label{B.6}
R=g^{\mu \nu}R_{\mu \nu}.
\end{equation}

Ricci tensor only depends on the metric through the Levi-Cività (affine) connection. In terms of the Riemann curvature tensor, the connection is written as:
\begin{equation}\label{B.7}
R^{\mu}_{\ \nu \rho \sigma}=\partial_{\rho} \Gamma^{\mu}_{\ \nu \sigma} - \partial_{\sigma} \Gamma^{\mu}_{\ \nu \rho} + \Gamma^{\mu}_{\ \kappa \rho} \Gamma^{\kappa}_{\ \nu \sigma} - \Gamma^{\mu}_{\ \kappa \sigma} \Gamma^{\kappa}_{\ \nu \rho}.
\end{equation}

Now, it's known that the symbols $\Gamma^{\alpha}_{\ \beta \gamma}$ are not tensors. This is easy to see in its transformation formula:
\begin{equation}\label{B.8}
\bar{\Gamma}^{\alpha}_{\ \beta \gamma}=\frac{\partial \bar{x}^{\alpha}}{\partial x^{\mu}} \ \left(\Gamma^{\mu}_{\ \nu \sigma} \ \frac{\partial x^{\nu}}{\partial \bar{x}^{\beta}} \frac{\partial x^{\sigma}}{\partial \bar{x}^{\gamma}} + \frac{\partial^2 x^{\mu}}{\partial \bar{x}^{\beta} \partial \bar{x}^{\gamma}}\right).
\end{equation}

The last term destroys the tensor character of $\Gamma^{\mu}_{\ \nu \sigma}$. But when we take the difference between two sets of $\Gamma$'s, $\delta \Gamma^{\mu}_{\ \nu \sigma}$, this term vanishes and it turns out to be tensorial.

\

The variation of the Riemann, $\delta R^{\mu}_{\ \nu \rho \sigma}$, consist in two terms of derivatives of $\delta \Gamma^{\mu}_{\ \nu \sigma}$ and four terms of the form $\Gamma^{\mu}_{\ \kappa \rho}\delta \Gamma^{\kappa}_{\ \nu \sigma}$:
\begin{equation}\label{B.9}
\delta R^{\mu}_{\ \nu \rho \sigma}=\partial_{\rho}\delta \Gamma^{\mu}_{\ \nu \sigma} - \partial_{\sigma} \delta \Gamma^{\mu}_{\ \nu \rho}+\delta\Gamma^{\mu}_{\ \kappa \rho} \Gamma^{\kappa}_{\ \nu \sigma}+\Gamma^{\mu}_{\ \kappa \rho} \delta \Gamma^{\kappa}_{\ \nu \sigma}-\delta\Gamma^{\mu}_{\ \kappa \sigma} \Gamma^{\kappa}_{\ \nu \rho}-\Gamma^{\mu}_{\ \kappa \sigma} \delta \Gamma^{\kappa}_{\ \nu \rho}.
\end{equation}

Since $\delta \Gamma^{\mu}_{\ \nu \sigma}$ is a tensor, we can take its covariant derivative:
\begin{equation}\label{B.10}
\nabla_{\rho}\delta \Gamma^{\mu}_{\ \nu \sigma}=\partial_{\rho}\delta \Gamma^{\mu}_{\ \nu \sigma}+ \Gamma^{\mu}_{\ \lambda \rho} \delta \Gamma^{\lambda}_{\ \nu \sigma} - \Gamma^{\lambda}_{\ \nu \rho} \delta \Gamma^{\mu}_{\ \lambda \sigma} - \Gamma^{\lambda}_{\ \rho \sigma} \delta \Gamma^{\mu}_{\ \nu \lambda}.
\end{equation}

Performing the same calculus for $\delta \Gamma^{\mu}_{\ \nu \rho}$ and solving for the normal derivative, one gets:
\begin{equation}\label{B.11}
\delta R^{\mu}_{\ \nu \rho \sigma}=\nabla_{\rho}\delta \Gamma^{\mu}_{\ \nu \sigma} - \nabla_{\sigma} \delta \Gamma^{\mu}_{\ \nu \rho}.
\end{equation}

The variation of the Ricci tensor is obtained by simply contracting two indices of the variation of the Riemann. The resulting expression is the \textit{Palatini identity}:
\begin{equation}\label{B.12}
\delta R_{\nu \sigma}=\delta R^{\mu}_{\ \nu \mu \sigma}=\nabla_{\mu}\delta \Gamma^{\mu}_{\ \nu \sigma} - \nabla_{\sigma} \delta \Gamma^{\mu}_{\ \nu \mu}.
\end{equation}

Recalling equation (\ref{B.6}), the variation of the Ricci scalar (with respect to the inverse metric) is:
\begin{equation}\label{B.13}
\begin{split}
\delta R&=\delta g^{\mu \nu} R_{\mu \nu} + g^{\mu \nu} \delta R_{\mu \nu}\\
&=\delta g^{\mu \nu} R_{\mu \nu} + g^{\mu \nu} \left(\nabla_{\lambda}\delta \Gamma^{\lambda}_{\ \mu \nu} - \nabla_{\nu} \delta \Gamma^{\lambda}_{\ \mu \lambda}\right)\\
&=\delta g^{\mu \nu} R_{\mu \nu} + \nabla_{\lambda}\left(g^{\mu \nu}\delta \Gamma^{\lambda}_{\ \mu \nu} - g^{\mu \lambda}\delta \Gamma^{\nu}_{\ \mu \nu}\right).
\end{split}
\end{equation}

In this last step we have used $\nabla_{\rho} g^{\mu \nu}=0$. Taking a closer look to the second term (denote the parenthesis as $v^{\lambda}$), we can introduce it in the integral and apply Stoke's theorem:
\begin{equation}\label{B.14}
\int_{\mathcal{V}} \diff^n x \sqrt{-g} \ \nabla_{\lambda}v^{\lambda}=(-1)^{n}\oint_{\partial\mathcal{V}} \diff^{n-1}y \sqrt{\left|h\right|} n^2 n_{\lambda} v^{\lambda},
\end{equation}
where $n_{\mu}$ is the unit normal vector to $\partial \mathcal{V}$ and $h$ the determinant of the induced metric on $\partial \mathcal{V}$. At the boundary, we must impose the condition:
\begin{equation}\label{B.15}
\delta g^{\mu \nu}=\delta g_{\mu \nu}=0
\end{equation}

Under this assumption, we can evaluate the term $n_{\lambda} v^{\lambda}$ at $\partial \mathcal{V}$. But first, it's necessary to calculate the variation of the connexion parameters there:
\begin{equation}\label{B.16}
\left.\delta \Gamma^{\mu}_{\nu \sigma}\right|_{\partial \mathcal{V}}=\frac{1}{2} g^{\mu \alpha} \left[\partial_{\nu}\delta g_{\alpha \sigma}+\partial_{\sigma} \delta g_{\nu \alpha}-\partial_{\alpha} \delta g_{\sigma \nu}\right].
\end{equation}

Substituting that in $v^{\lambda}$ leads us to:
\begin{equation}\label{B.17}
\begin{split}
\left.v^{\lambda}\right|_{\partial \mathcal{V}}&=g^{\mu \nu} \frac{1}{2} g^{\lambda \sigma} \left[\partial_{\mu}\delta g_{\sigma \nu}+\partial_{\nu} \delta g_{\mu \sigma}-\partial_{\sigma} \delta g_{\nu \mu}\right]-g^{\mu \lambda} \frac{1}{2} g^{\nu \sigma} \left[\partial_{\mu}\delta g_{\sigma \nu}+\partial_{\nu} \delta g_{\mu \sigma}-\partial_{\sigma} \delta g_{\nu \mu}\right]\\
&=\frac{1}{2} g^{\mu \nu} g^{\lambda \sigma} \left[\partial_{\mu}\delta g_{\sigma \nu}+\partial_{\nu} \delta g_{\mu \sigma}-\partial_{\sigma} \delta g_{\nu \mu}\right] - \frac{1}{2} g^{\sigma \lambda} g^{\nu \mu} \left[\partial_{\sigma}\delta g_{\mu \nu}+\partial_{\nu} \delta g_{\sigma \mu}-\partial_{\mu} \delta g_{\nu \sigma}\right]\\
&=g^{\mu \nu} g^{\lambda \sigma} \left(\partial_{\mu}\delta g_{\sigma \nu} - \partial_{\sigma} \delta g_{\nu \mu}\right).
\end{split}
\end{equation}

So now we're able to write:
\begin{equation}\label{B.18}
\begin{split}
\left.n_{\lambda} v^{\lambda}\right|_{\partial \mathcal{V}}&=n^{\lambda} g^{\mu \nu} \left(\partial_{\mu}\delta g_{\lambda \nu} - \partial_{\lambda} \delta g_{\nu \mu}\right)\\
&=n^{\lambda} \left(n^2 n^\mu n^\nu + h^{\mu \nu}\right) \left(\partial_{\mu}\delta g_{\lambda \nu} - \partial_{\lambda} \delta g_{\nu \mu}\right)\\
&=n^{\lambda} h^{\mu \nu} \left(\partial_{\mu}\delta g_{\lambda \nu} - \partial_{\lambda} \delta g_{\nu \mu}\right),
\end{split}
\end{equation}
where we have used the definition of induced metric, $h_{\mu \nu}=g_{\mu \nu}-n^2 n_\mu n_\nu$, on the hypersurface $\partial \mathcal{V}$. As $\delta g_{\mu \nu}$ vanishes on $\partial \mathcal{V}$, its tangential derivative must also vanish, $\delta g_{\mu \nu, \sigma} e^{\sigma}_{\ c}=0$. It follows that:
\begin{equation}\label{B.19}
h^{\mu \nu} \delta g_{\lambda \nu, \mu}=h^{ab} e^{\mu}_{\ a} e^{\nu}_{\ b} \delta g_{\lambda \nu, \mu}=0.
\end{equation}

Finally, we obtain:
\begin{equation}\label{B.20}
\left.n_{\lambda} v^{\lambda}\right|_{\partial \mathcal{V}}=- h^{\mu \nu} n^{\lambda} \delta g_{\nu \mu, \lambda}.
\end{equation}

Until now, we have calculated the variation of $R$. Finally we focus our attention on the variation of the determinant of the metric.

\

In matrix algebra, \textit{Jacobi's Formula} express the derivative of the determinant of a matrix in terms of its adjugate and its derivative. It may be written as:
\begin{equation}\label{B.21}
\frac{\diff}{\diff t} \det{A (t)}=\trace \left( \mathrm{adj}(A(t)) \frac{\diff A(t)}{\diff t} \right).
\end{equation}

This may be applied to calculate the variation of the determinant of the metric as follows:
\begin{equation}\label{B.22}
\delta g=g g^{\mu \nu} \delta g_{\mu \nu}.
\end{equation}

Using Leibniz's rule for product differentiation one can easily get:
\begin{equation}\label{B.23}
\delta \sqrt{-g}=\frac{-1}{2 \sqrt{-g}} \ \delta g=\frac{-1}{2 \sqrt{-g}} \ g g^{\mu \nu} \delta g_{\mu \nu}=\frac{1}{2} \ \sqrt{-g} \ g^{\mu \nu} \delta g_{\mu \nu}.
\end{equation}

Similarly, differentiation of the inverse of the variation of the metric gives:
\begin{equation}\label{B.24}
\delta g^{\mu \nu}=-g^{\mu \alpha} \delta g_{\alpha \beta} g^{\beta \nu}.
\end{equation}

The use of this last expression, together with (\ref{B.23}) leads to:
\begin{equation}\label{B.25}
\delta \sqrt{-g}=-\frac{1}{2} \ \sqrt{-g} \ g_{\mu \nu} \delta g^{\mu \nu}.
\end{equation}

Finally, we have calculated the variation of the Einstein-Hilbert action:
\begin{equation}\label{B.26}
\begin{split}
16 \pi \ \delta S_{EH}&=\int_{\mathcal{V}} \diff^n x \left(R_{\mu \nu} - \frac{1}{2} \ g_{\mu \nu}\right) \sqrt{-g} \ \delta g^{\mu \nu} - (-1)^{n}\oint_{\partial\mathcal{V}} \diff^{n-1}y \sqrt{\left|h\right|} n^2 h^{\mu \nu} n^{\lambda} \delta g_{\nu \mu, \lambda}\\
&=\int_{\mathcal{V}} \diff^n x \ G_{\mu \nu} \sqrt{-g} \ \delta g^{\mu \nu} - (-1)^{n}\oint_{\partial\mathcal{V}} \diff^{n-1}y \sqrt{\left|h\right|} n^2 h^{\mu \nu} n^{\lambda} \delta g_{\nu \mu, \lambda}.
\end{split}
\end{equation}

As we can see, except for the surface integral, this variation leads us to the Einstein field equations. This last term will be cancelled by the variation of the boundary term added to the total action. The need for that term comes from the dependency of the Ricci scalar (the Lagrangian density of General Relativity) on second derivatives os the metric.

\section{Variation of the boundary term}

The boundary term is the integral over the boundary of the trace of the extrinsic curvature of the boundary, $K$:
\begin{equation}\label{B.27}
S_{B}\left[g\right]=\frac{(-1)^n}{8 \pi} \ \oint_{\partial \mathcal{V}} \diff^{n-1}y\  n^2 K \sqrt{\left|h\right|}.
\end{equation}

Since the induced metric is fixed on $\partial \mathcal{V}$, the variation only affects to $K$. The trace of the extrinsic curvature may be written as:
\begin{equation}\label{B.28}
\begin{split}
K&=n^{\mu}_{\ ;\mu}=g^{\mu \nu} n_{\mu;\nu}=(n^2 n^{\mu} n^{\nu} + h^{\mu \nu})n_{\mu;\nu}\\
&=h^{\mu \nu} n_{\mu;\nu}=h^{\mu \nu}\left(n_{\mu,\nu} - \Gamma^{\lambda}_{\ \mu \nu} n_{\lambda}\right).
\end{split}
\end{equation}

The variation of this quantity is found:
\begin{equation}\label{B.29}
\begin{split}
\delta K&=-h^{\mu \nu} \delta \Gamma^{\lambda}_{\ \mu \nu} n_{\lambda}\\
&=-h^{\mu \nu} \left(\frac{1}{2} g^{\lambda \alpha} \left[\delta g_{\alpha \nu,\mu}+\delta g_{\mu \alpha,\nu}-\delta g_{\nu \mu,\alpha}\right]\right) n_{\lambda}\\
&=-\frac{1}{2} h^{\mu \nu} \left[\delta g_{\alpha \nu,\mu}+\delta g_{\mu \alpha,\nu}-\delta g_{\nu \mu,\alpha}\right] n^{\alpha}\\
&=\frac{1}{2} h^{\mu \nu} \delta g_{\mu \nu,\alpha} n^{\alpha},\\
\end{split}
\end{equation}
where we used (\ref{B.16}) and (\ref{B.19}). This way we obtain that variation of the boundary term is:
\begin{equation}\label{B.30}
16 \pi \ \delta S_{B} = (-1)^n \ \oint_{\partial \mathcal{V}} \diff^{n-1}y \ \sqrt{\left|h\right|} n^2 h^{\mu \nu} n^{\alpha} \delta g_{\mu \nu,\alpha}.
\end{equation}

We see that indeed it cancels out the surface integral in the variation of the Einstein-Hilbert action. Because, as we said at the beginning, the nondynamical term only affects to the numerical values of the gravitational action but not to the variation, the variation of the total gravitational action is:
\begin{equation}\label{B.31}
\delta S_{G}=\frac{1}{16 \pi} \ \int_{\mathcal{V}} \diff^n x \ G_{\mu \nu} \sqrt{-g} \ \delta g^{\mu \nu},
\end{equation}
and evidently, as the variations of the metric are arbitrary and Hamilton's Principle is applied for the real metric, we obtain the vacuum Einstein field equations:
\begin{equation}\label{B.32}
G_{\mu \nu}=0.
\end{equation}

\section{Variations of the matter action}

The matter action is taken to be:
\begin{equation}\label{B.33}
S_{M} \left[\phi;g\right]=\int_{\mathcal{V}} \diff^n x \ \sqrt{-g} \ \lag(\phi,\phi_{,\mu};g_{\mu\nu}),
\end{equation}
for some Lagrangian density $\lag$ which depends only on the matter fields and its first derivatives and the metric (and none of its derivatives). The variation of this action is given by (recalling equation (\ref{B.25})):
\begin{equation}\label{B.34}
\begin{split}
\delta S_M&=\int_{\mathcal{V}} \diff^n x \ \delta\left(\sqrt{-g} \ \lag(\phi,\phi_{,\mu};g_{\mu\nu})\right)\\
&=\int_{\mathcal{V}} \diff^n x \ \left(\sqrt{-g} \ \frac{\partial \lag}{\partial g^{\mu \nu}} \ \delta g^{\mu \nu} + \lag \ \delta\sqrt{-g} \right)\\
&=\int_{\mathcal{V}} \diff^n x \left(\frac{\partial \lag}{\partial g^{\mu \nu}} - \frac{1}{2} \ \lag \ g_{\mu \nu}\right) \ \sqrt{-g} \ \delta g^{\mu \nu}.
\end{split}
\end{equation}

We define the stress-energy tensor as:
\begin{equation}\label{B.35}
T_{\mu \nu}= \lag \ g_{\mu \nu} - 2 \ \frac{\partial \lag}{\partial g^{\mu \nu}}.
\end{equation}

This redefinition yields:
\begin{equation}\label{B.36}
\delta S_M=- \frac{1}{2} \ \int_{\mathcal{V}} \diff^n x \ T_{\mu \nu} \ \delta g^{\mu \nu} \ \sqrt{-g}.
\end{equation}

Since the variation of the nondynamical term is zero, as we claimed when we introduced it, putting together all the results the variation of the whole variation of the action functional of General Relativity may be written as:
\begin{equation}\label{B.37}
\delta S_{GR}=\delta \left(S_{G}+S_{M}\right)=\frac{1}{2} \ \int_{\mathcal{V}} \diff^n x \ \left(\frac{1}{8 \pi} \ G_{\mu \nu}-T_{\mu \nu}\right) \sqrt{-g} \ \delta g^{\mu \nu}=0.
\end{equation}

Because the variation $\delta g^{\mu \nu}$ is arbitrary within the region $\mathcal{V}$, the Einstein field equations follows from this variational principle:
\begin{equation}\label{B.38}
G_{\mu \nu}=8 \pi T_{\mu \nu}.
\end{equation}

\section{The Einstein-Maxwell action}

Adding the gravitational action (Einstein-Hilbert+Boundary+Nondynamical) to the Maxwell action we end up with the so-called Einstein-Maxwell action:
\begin{equation}\label{B.39}
S_{EM}\left[g_{\mu \nu},A_{\mu}\right]=S_{G}\left[g\right] - \frac{1}{4} \ \int_{\mathcal{V}} \diff^4 x \ F^2 \ \sqrt{-g}.
\end{equation}

The electromagnetic tensor $F_{\mu \nu}$ is a combination of the electric and magnetic fields into a covariant antisymmetric tensor. In terms of the electromagnetic potential $A^{\mu}$ it is written as:
\begin{equation}\label{B.40}
F_{\mu \nu}=\nabla_{\mu} A_{\nu} - \nabla_{\nu} A_{\mu}.
\end{equation}

First of all, we calculate the variation of that action under an arbitrary variation of the metric. The first term produce the Einstein tensor (equation (\ref{B.31})). We have to take into account the variation of the Maxwell term:
\begin{equation}\label{B.41}
\begin{split}
\delta S_{\mathrm{M}}&=-\frac{1}{4} \ \int_{\mathcal{V}} \diff^4 x \ \left(\delta F^2 \sqrt{-g} + F^2 \delta \sqrt{-g}\right)\\
&=-\frac{1}{4} \ \int_{\mathcal{V}} \diff^4 x \ \left(\delta(F_{\mu \nu} F_{\sigma \rho} g^{\mu \sigma} g^{\nu \rho}) \sqrt{-g} - \frac{1}{2} \ F^2  \sqrt{-g} \ g_{\mu \nu} \delta g^{\mu \nu}\right)\\
&=-\frac{1}{4} \ \int_{\mathcal{V}} \diff^4 x \ \left(2 F_{\mu \sigma} F^{\ \sigma}_{\nu}   - \frac{1}{2} \ F^2  \ g_{\mu \nu} \right) \sqrt{-g} \ \delta g^{\mu \nu}\\
&=-\frac{1}{2} \ \int_{\mathcal{V}} \diff^4 x \ T_{\mu \nu} \sqrt{-g} \ \delta g^{\mu \nu}.\\
\end{split}
\end{equation}

We have defined the energy-momentum tensor of the electromagnetic vector field $T_{\mu \nu}$ as
\begin{equation}\label{B.42}
T_{\mu \nu}=F_{\mu \sigma} F^{\ \sigma}_{\nu} - \frac{1}{4} \ F^2  \ g_{\mu \nu}.
\end{equation}

The resulting field equations would be:
\begin{equation}\label{B.43}
G_{\mu \nu} - 8 \pi \left(F_{\mu \sigma} F^{\ \sigma}_{\nu} - \frac{1}{4} \ F^2  \ g_{\mu \nu}\right)=0.
\end{equation}

By taking the covariant derivative of this last equation and using the \textit{contracted Bianchi identity}, $\nabla_{\mu} G^{\mu \nu}=0$, we find out:
\begin{equation}\label{B.44}
F_{\mu\sigma} \nabla_\mu F^{\mu\sigma} - \frac{3}{2} \ F^{\mu\sigma} \nabla_{\left[\mu\right.}F_{\left.\sigma\nu\right]}=0.
\end{equation}

As the affine connection is symmetric, we will have:
\begin{equation}\label{B.45}
\nabla_\mu F^{\mu\nu}=0.
\end{equation}

We can see that Einstein field equations imply generically Maxwell equations.

\section{The nondynamical term}

In order to obtain the field equations we have made no mention to this term, because it's a constant and its variation with respect $g_{\mu \nu}$ gives zero. Its role is to regularize the numerical value of the gravitational action.

\

Consider a solution of the vacuum field equations. In this case $R=0$ and the action (omitting the present term) turns to be:
\begin{equation}\label{B.46}
S_{G}\left[g\right]=\frac{1}{8 \pi} \ \oint_{\partial \mathcal{V}} \diff^{3}y\  n^2 K \sqrt{\left|h\right|}.
\end{equation}

We are going to evaluate the integral in the case of flat spacetime. Consider the integration region bounded by a three cylinder at $r=R$ and two $t=const.$ hypersurfaces. In the hypersurfaces of constant time $K=0$, and in the cylinder we have a induced metric:
\begin{equation}\label{B.47}
\diff s^2=- \diff t^2 + R^2 \diff \Omega^2.
\end{equation}

A straightforward calculation give us $\left|h\right|^{1/2}=R^2 \sin{\theta}$. The unit normal vector is $n_\alpha=\partial_\alpha r$ and so $K=n^{\alpha}_{\ ;\alpha}=2/R$. Then, we can evaluate the integral:
\begin{equation}\label{B.48}
\oint_{\partial \mathcal{V}} \diff^{3}y\  n^2 K \ \sqrt{\left|h\right|}=8\pi R \left(t_2-t_1\right).
\end{equation}

If we take the limit $R\rightarrow\infty$, i.e. when the boundary is pushed to infinity, this integral diverges. For asymptotically-flat spacetimes this problem does not go away, and therefore the gravitational action is not well-defined.

\

The term $S_0$ is introduced to be equal to the gravitational action of flat spacetime. This way the difference $S_B-S_0$ converges in the limit and therefore is well-defined for asymptotically-flat spacetimes. We choose then $K_0$ to be the extrinsic curvature of the boundary embedded in flat spacetime:
\begin{equation}\label{B.49}
S_{0}=\frac{1}{8 \pi} \ \oint_{\partial \mathcal{V}} \diff^{3}y\  n^2 K_0 \sqrt{\left|h\right|}.
\end{equation}
\chapter{Foliations, masses and momentum}\label{Appendix C}

This chapter is devoted to the more subtle approach to General Relativity provided by the Hamiltonian formulation of Classical Mechanics.

\

This treatment involves a decomposition of spacetime into space and time separately. In the first part we study this 3+1 decompositon, the foliation of our spacetime region $\mathcal{V}$ by spacelike hypersurfaces. Later we pay special attention to the foliation of the boundary, in order to write the action functional in terms of this decomposition. Then we construct the gravitational Hamiltonian, which inherits the boundary terms of the action.

\

Finally we study the connection between this Hamiltonian and the mass and angular momentum of asymptotically-flat spacetime.

\section{The 3+1 decomposition}

The Hamiltonian $H\left[p,q\right]$ is a functional of $q$,the field configuration, and $p$, its canonical conjugate momentum, on a spacelike hypersurface $\Sigma$. If we want to express our gravitational action in terms of this Hamiltonian we need to \textit{foliate} our integration region $\mathcal{V}$ with a family of such spacelike hypersurfaces, corresponding one for each instant of time. This is what is known as the 3+1 decomposition.

\

In order to perform this task, we introduce a scalar single-valued field $t(x^\alpha)$ such that the relation $t=const.$ describes a hypersurface $\Sigma_t$, with normal vector $n_{\alpha}\propto\partial_\alpha t$. We can introduce also a congruence of geodesics $\gamma$ intersecting the hypersurfaces, with $t$ as a parameter and the vector $t^\alpha$ tangent to it. On each $\Sigma_t$ we prepare coordinates $y^{a}$ which are constants on a member of the congruence, $\gamma_p$, in the way $y^{a}(P)=y^{a}(P')=y^{a}(P'')$ where the $'$ denotes that we are talking about a point in a different hypersurface $\Sigma_{t'}$ or $\Sigma_{t''}$.

\

Thus, this construction defines a coordinate system in $\mathcal{V}$, $(t,y^{a})$. This system is related to the original $x^\alpha$ by:
\begin{equation}\label{C.1}
t^\alpha=\left(\frac{\partial x^\alpha}{\partial t}\right)_{y^a}.
\end{equation}

We now define the tangent vectors and the unit normal to the hypersurfaces $\Sigma_t$ as:
\begin{equation}\label{C.2}
e^{\alpha}_a=\left(\frac{\partial x^\alpha}{\partial y^{a}}\right)_t,
\end{equation}
\begin{equation}\label{C.3}
n_\alpha=-N\partial_\alpha t,
\end{equation}
where the function $N$, the lapse, normalize properly the vector field. Since we did not construct the congruence of geodesics to be orthogonal to $\Sigma_t$ the vector $t^{\alpha}$ may be decomposed as:
\begin{equation}\label{C.4}
t^\alpha=N n^\alpha+ N^a e^{\alpha}_a,
\end{equation}
where the vector $N^{a}$ is what we call shift. Now it is possible to express the line in terms of this new coordinates. First, we have:
\begin{equation}\label{C.5}
\begin{split}
\diff x^\alpha&=t^\alpha \diff t + e^{\alpha}_a \diff y^a\\
&=N \diff t \ n^\alpha+  \left(N^a \diff t +  \diff y^a\right) e^{\alpha}_a.\\
\end{split}
\end{equation}

Now, using the definition of the induced metric on the hypersurface:
\begin{equation}\label{C.6}
h_{ab}=g_{\alpha \beta} \ e^{\alpha}_a e^{\beta}_b,
\end{equation}
we end up with:
\begin{equation}\label{C.7}
\begin{split}
\diff s^2&=g_{\alpha \beta} \ \diff x^\alpha \diff x^\beta\\
&=-N^2 \diff t^2 + h_{ab} \left(N^a \diff t +  \diff y^a\right) \left(N^b \diff t +  \diff y^b\right).
\end{split}
\end{equation}

We can also calculate the determinant of the metric. We have that $g^{tt}=g/h$. Then, we find out:
\begin{equation}\label{C.8}
g^{tt}=g^{\alpha \beta} \ t_{, \alpha} t_{,\beta}=N^{-2} \ g^{\alpha \beta} \ n_\alpha n_\beta=-N^{-2}.
\end{equation}

We finally obtain the desired relation:
\begin{equation}\label{C.9}
\sqrt{-g}=N\sqrt{h}.
\end{equation}

All these results will be very useful in later sections, when we foliate our region of integration to express the gravitational action in terms of the Hamiltonian.

\section{Foliation of the boundary}

We consider our region of spacetime $\mathcal{V}$ to be foliated by spacelike hypersurfaces $\Sigma_t$ such as the introduced in the previous section. Those 3-surfaces are bounded by a set of 2-surfaces $S_t$. So, our region $\mathcal{V}$ is itself bounded by two spacelike hypersurfaces $\Sigma_{t_1}$ and $\Sigma_{t_2}$, and a timelike hypersurface $\mathcal{B}$, the union of all the $S_t$.

\

Before trying to obtain a expression for the gravitational Hamiltonian we must consider some properties of the foliation of the timelike boundary $\mathcal{B}$. The closed 2-surfaces $S_t$ are the boundaries of the spacelike hypersurfaces $\Sigma_t$ and are given by a set of equations of the form $\Phi(y^a)=0$ or by parametric relations $y^a(\theta^A)$ where $\theta^A$ are coordinates on $S_t$.

\

If we denote the unit normal to $S_t$ as $r_a$, and define the tangent vectors as:
\begin{equation}\label{C.10}
e^{a}_A=\frac{\partial y^a}{\partial \theta^A}.
\end{equation}

We can introduce also associated 4-vectors to these quantities:
\begin{equation}\label{C.11}
r^\alpha=r^a e^\alpha_a,
\end{equation}
\begin{equation}\label{C.12}
e^{\alpha}_A=e^{\alpha}_a e^{a}_A=\left(\frac{\partial x^\alpha}{\partial \theta^A}\right)_t.
\end{equation}

The induced metric on $S_t$ would be given by:
\begin{equation}\label{C.13}
\begin{split}
\sigma_{AB}&=h_{ab} \ e^a_A e^b_B\\
&=\left(g_{\alpha \beta} \ e^\alpha_a e^\beta_b\right) \ e^a_A e^b_B\\
&=g_{\alpha \beta} \ e^\alpha_A e^\beta_B.
\end{split}
\end{equation}

From this we can obtain easily the three-dimensional completeness relation and then express the analogue for four dimensions in terms of the inverse induced metric $\sigma^{AB}$.

\

The extrinsic curvature of the 2-surfaces $S_t$ can be computed as follows:
\begin{equation}\label{C.14}
\begin{split}
k_{AB}&=r_{a|b} \ e^a_A e^b_B\\
&=r_{a;b} \ e^\alpha_A e^\beta_B.\\
\end{split}
\end{equation}

We now want to relate the coordinates $\theta^A$ on one surface $S_t$ to others in a different surface $S_{t'}$. Consider a congruence of geodesics $\beta$ on $\mathcal{B}$ intersecting $S_t$ orthogonally. Under this assumption, $n^\alpha$ would be their tangent vectors. If the coordinate $\theta^A$ does not vary along a curve of the congruence, i.e. if a geodesic $\beta_P$ intersects one surface at the point $P$ and another surface at $P'$ then the same coordinate would label that points, then $t$ is a suitable parameter on the curves. For all these we can write:
\begin{equation}\label{C.15}
n^\alpha=N^{-1} \ \left(\frac{\partial x^\alpha}{\partial t}\right)_{\theta^A}.
\end{equation}

The factor $N^{-1}$ comes from equation (\ref{2.3}), to yield the normalization condition $n_\alpha n^\alpha=-1$. This way we ensure that $n^{\alpha}$ and $e^{\alpha}_A$ are everywhere orthonormal.

\

Now we place coordinates $z^i$ on $\mathcal{B}$. Since it is foliated by the 2-surfaces $S_t$ the unit normal $r^{\alpha}$ is the same. We can introduce the tangent vectors to this timelike hypersurface:
\begin{equation}\label{C.16}
e^{\alpha}_{i}=\frac{\partial x^\alpha}{\partial z^i}.
\end{equation}

Then, the induced metric on $\mathcal{B}$ is given by:
\begin{equation}\label{C.17}
\gamma_{ij}=g_{\alpha \beta} \ e^{\alpha}_i e^{\beta}_j.
\end{equation}

We choose our coordinates to be $z^i=(t,\theta^A)$. Infinitesimal displacement in $\mathcal{B}$ is then written as:
\begin{equation}\label{C.18}
\begin{split}
\diff x^{\alpha}&=\left(\frac{\partial x^{i}}{\partial t}\right)_{\theta^A} \ \diff t + \left(\frac{\partial x^{i}}{\partial \theta^A}\right)_{t} \ \diff \theta^A\\
&=N n^\alpha \ \diff t + e^\alpha_A \ \diff \theta^A.
\end{split}
\end{equation}

The line element is now easily computed, taking into account the orthogonality relation $n_\alpha e^\alpha_A=0$:
\begin{equation}\label{C.19}
\begin{split}
\diff s^2_{\mathcal{B}}&=g_{\alpha \beta} \ \diff x^\alpha \diff x^\beta\\
&=\left(g_{\alpha \beta} \ n^\alpha n^\beta\right) N^2 \ \diff t^2 + \left(g_{\alpha \beta} \ e^\alpha_A e^\beta_B\right) \ \diff \theta^A \diff \theta^B\\
&=-N^2 \ \diff t^2 + \sigma_{AB} \ \diff \theta^A \diff \theta^B=\gamma_{ij} \ \diff z^i \diff z^j.\\
\end{split}
\end{equation}

The same way we obtained equation (\ref{2.9}), we can now establish the relation:
\begin{equation}\label{C.20}
\sqrt{-\gamma}=N\sqrt{\sigma}.
\end{equation}

For the last, we point out the form that the extrinsic curvature of $\mathcal{B}$ embedded in the four dimensional spacetime would take:
\begin{equation}\label{C.21}
\mathcal{H}_{ij}=r_{\alpha;\beta} \ e^{\alpha}_i e^\beta_j,
\end{equation}
because, as we said previously, the unit normal to $S_t$ is also normal to $\mathcal{B}$.

\section{Decomposition of the action}

The gravitational action introduced in the previous chapters (equations (\ref{1.5}), (\ref{1.27}) and ()) may be re-expressed in terms of the 3+1 decomposition. We are going to start from:
\begin{equation}\label{C.22}
\left(16 \pi\right) S_{G}=\int_{\mathcal{V}} \diff^4 x \ \sqrt{-g} \ R + 2 \oint_{\partial \mathcal{V}} \diff^3 y \ n^2 K \sqrt{h}.
\end{equation}

Here $\mathcal{V}$ is our four-dimensional integration region, $\partial \mathcal{V}$ its boundary, $y^\alpha$ coordinates on that boundary, $h_{ab}$ the induced metric, $K_{ab}$ the extrinsic curvature and $n^\alpha$ the unit (outward) normal to $\partial \mathcal{V}$. He have omitted the nondynamical term $S_0$, but we will introduce it at the end of the calculations.

\

First of all, the boundary $\partial \mathcal{V}$ is the union of three different hypersurfaces, two spacelike $\Sigma_{t_1}$, $\Sigma_{t_2}$ and one timelike $\mathcal{B}$. As the unit normal vector to the spacelike hypersurfaces are future directed we must add an additional minus sign to the terms involving $\Sigma_{t_1}$ (because $t_1<t_2$).

\

From now and all over the chapter the quantities $n^\alpha$, $y^\alpha$, $h_{ab}$ and $K_{ab}$ are referred to the hypersurfaces $\Sigma_t$. Therefore, the gravitational action takes the form:
\begin{equation}\label{C.23}
\begin{split}
\left(16 \pi\right) S_{G}=&\int_{\mathcal{V}} \diff^4 x \ \sqrt{-g} \ R + 2 \int_{\Sigma_{t_1}} \diff^3 y \ K \sqrt{h}\\
&-2 \int_{\Sigma_{t_2}} \diff^3 y \ K \sqrt{h}+2 \int_{\mathcal{B}} \diff^3 z \ \mathcal{H} \sqrt{- \gamma}.\\
\end{split}
\end{equation}

The region $\mathcal{V}$ is foliated by spacelike hypersurfaces $\Sigma_t$. On those the Ricci is given by \cite{p.math}:
\begin{equation}\label{C.24}
R=\vphantom{a}^3R+K^{ab}K_{ab}-K^2-2\left(n^\alpha_{;\beta}n^\beta-n^\alpha n^\beta_{;\beta}\right)_{;\alpha}.
\end{equation}

There $\vphantom{a}^3R$ is the Ricci scalar given by the induced metric. Using (\ref{2.9}) and the Stoke's theorem the volume integral reduces to:
\begin{equation}\label{C.25}
\begin{split}
\int_{\mathcal{V}} \diff^4 x \ \sqrt{-g} \ R=&\int_{t_1}^{t_2} \diff t \int_{\Sigma_t} \diff^3 y \ \left(\vphantom{a}^3R+K^{ab}K_{ab}-K^2\right) N \sqrt{h}\\
&-2\oint_{\partial \mathcal{V}} \diff \Sigma_\alpha \ \left(n^\alpha_{;\beta}n^\beta-n^\alpha n^\beta_{;\beta}\right).
\end{split}
\end{equation}

Now we can split the last surface integral in three contributions. In the case of $\Sigma_{t_1}$ we have:
\begin{equation}\label{C.26}
\begin{split}
-2\int_{\Sigma_{t_1}} \diff \Sigma_\alpha \ \left(n^\alpha_{;\beta}n^\beta-n^\alpha n^\beta_{;\beta}\right)&=2\int_{\Sigma_{t_1}} \diff^3 y \ \left(n^\alpha_{;\beta}n^\beta-n^\alpha n^\beta_{;\beta}\right) n^\alpha \ \sqrt{h}\\
&=-2\int_{\Sigma_{t_1}} \diff^3 y \ n^\beta_{;\beta} \ \sqrt{h}=-2\int_{\Sigma_{t_1}} \diff^3 y \ K\ \sqrt{h}.
\end{split}
\end{equation}

In a similar fashion, the corresponding integral over $\Sigma_{t_2}$ is:
\begin{equation}\label{C.27}
\begin{split}
-2\int_{\Sigma_{t_2}} \diff \Sigma_\alpha \ \left(n^\alpha_{;\beta}n^\beta-n^\alpha n^\beta_{;\beta}\right)&=-2\int_{\Sigma_{t_2}} \diff^3 y \ \left(n^\alpha_{;\beta}n^\beta-n^\alpha n^\beta_{;\beta}\right) n^\alpha \ \sqrt{h}\\
&=2\int_{\Sigma_{t_2}} \diff^3 y \ n^\beta_{;\beta} \ \sqrt{h}=2\int_{\Sigma_{t_2}} \diff^3 y \ K\ \sqrt{h}.
\end{split}
\end{equation}

Introducing this into equation (\ref{2.23}) cancels out the surface integrals over 
$\Sigma_{t_1}$ and $\Sigma_{t_2}$. Using that $n^\alpha r_\alpha=0$ and integrating by parts, the remaining contribution of $\mathcal{B}$ results:
\begin{equation}\label{C.28}
\begin{split}
-2\int_{\mathcal{B}} \diff \Sigma_\alpha \ \left(n^\alpha_{;\beta}n^\beta-n^\alpha n^\beta_{;\beta}\right)&=-2\int_{\mathcal{B}} \diff^3 z \ \left(n^\alpha_{;\beta}n^\beta-n^\alpha n^\beta_{;\beta}\right) r^\alpha \ \sqrt{-\gamma}\\
&=-2\int_{\mathcal{B}} \diff^3 z \ n^\alpha_{;\beta}n^\beta r_\alpha \ \sqrt{-\gamma}=2\int_{\mathcal{B}} \diff^3 z \ n^\alpha n^\beta r_{\alpha;\beta} \ \sqrt{-\gamma}.
\end{split}
\end{equation}

Putting all together the action is now:
\begin{equation}\label{C.29}
\begin{split}
\left(16 \pi\right) S_{G}=&\int_{t_1}^{t_2} \diff t \int_{\Sigma_t} \diff^3 y \ \left(\vphantom{a}^3R+K^{ab}K_{ab}-K^2\right) N \sqrt{h}\\
&+2\int_{\mathcal{B}} \diff^3 z \ \left( \mathcal{H} + n^\alpha n^\beta r_{\alpha;\beta}\right) \ \sqrt{-\gamma}.
\end{split}
\end{equation}

Until now we only have used the foliation of the region $\mathcal{B}$. At this point we are going to apply the results of the previous section and use the fact that $\mathcal{B}$ is also foliated by the set of closed 2-surfaces $S_t$. First, we can find a more suitable form for $\mathcal{H}$ in terms of the metric and the unit normal:
\begin{equation}\label{C.30}
\begin{split}
\mathcal{H}&=\gamma^{ij} \mathcal{H}_{ij}\\
&=\gamma^{ij} \ r_{\alpha;\beta} \ e^{\alpha}_i e^\beta_j\\
&=r_{\alpha;\beta} \left(g^{\alpha \beta} - r^\alpha r^\beta\right).\\
\end{split}
\end{equation}

To obtain this we have used (\ref{C.21}) and the completeness relation deduced from (\ref{C.17}). Taking this into account, the second integrand becomes:
\begin{equation}\label{C.31}
\begin{split}
\mathcal{H} + n^\alpha n^\beta r_{\alpha;\beta}&=r_{\alpha;\beta} \left(g^{\alpha \beta} - r^\alpha r^\beta + n^\alpha n^\beta\right)\\
&=r_{\alpha;\beta} \ \sigma^{AB} \ e^\alpha_A e^\beta_B\\
&=\sigma^{AB} \ k_{AB}=k.
\end{split}
\end{equation}

To obtain this the inverse relation of (\ref{C.13}) was used. Finally, we obtain a expression for the gravitational action completely expressed in terms of a foliation. Re-inserting the corresponding term to $S_0$, we end up with:
\begin{equation}\label{C.32}
\begin{split}
\left(16 \pi\right) S_{G}=&\int_{t_1}^{t_2} \diff t \left[ \int_{\Sigma_t} \diff^3 y \ \left(\vphantom{a}^3R+K^{ab}K_{ab}-K^2\right) N \sqrt{h} \right.\\
&\left.+ 2\oint_{S_t} \diff^2 \theta \ \left(k-k_0\right) \ N \sqrt{\sigma}\right].
\end{split}
\end{equation}

There, $k_0$ is the extrinsic curvature of $S_t$ when embedded in flat spacetime. This term prevents the integral to diverge in the limit $S_t\rightarrow \infty$, in a similar way to what we discussed in the previous chapter. Similarly, if we include a matter action term, it must be also subjected to the 3+1 decomposition. This step would be straightforward, and we are going to omit it.

\section{The gravitational Hamiltonian}

In the case of general relativity the Hamiltonian is a functional of the induced metric on the hypersurfaces $\Sigma_t$ which foliates the region $\mathcal{V}$ and its conjugate momentum. In order to calculate this momentum, we must first compute what we call the time derivative of the induced metric, defined as:
\begin{equation}\label{C.33}
\dot{h}_{ab}\equiv  \ \pounds_t h_{ab},
\end{equation}
where the timelike vector $t^\alpha$ is given by equation (\ref{2.4}). Recalling the definition of induced metric and the general relation $\pounds_t e^\alpha_a=0$ we obtain:
\begin{equation}\label{C.34}
\begin{split}
\dot{h}_{ab}&=\pounds_t \ h_{ab}\\
&=\pounds_t \left(g_{\alpha \beta} \ e^\alpha_a e^\beta_b\right)\\
&=\pounds_t \ g_{\alpha \beta} \ e^\alpha_a e^\beta_b\\
&=\left(\cancelto{0}{g_{\alpha \beta; \mu} t^\mu} + t_{\alpha;\beta} + t_{\beta;\alpha}\right) \ e^\alpha_a e^\beta_b\\
&=\left[\left(N n_\alpha+ N_\alpha\right)_{;\beta} + \left(N n_\beta+ N_\beta\right)_{;\alpha}\right] \ e^\alpha_a e^\beta_b.
\end{split}
\end{equation}

Expanding the expression and using the definitions of extrinsic curvature and intrinsic differentiation \cite{p.math} we obtain:
\begin{equation}\label{C.35}
\dot{h}_{ab}=2NK_{ab} + N_{a|b}+ N_{b|a}.
\end{equation}

Now we can write the extrinsic curvature in terms of the time derivative of the induced metric and derivatives of the shift. The action depends on $\dot{h}_{ab}$ through the extrinsic curvature, and does not depend on time derivatives of lapse and shift, because those only specify the foliation of $\mathcal{V}$.

\

The conjugate momentum is defined as:
\begin{equation}\label{C.36}
p^{ab}\equiv\frac{\partial}{\partial \dot{h}_{ab}} \ \left(\sqrt{-g} \ \lag_{G}\right)=\frac{\partial K_{mn}}{\partial \dot{h}_{ab}} \ \frac{\partial}{\partial K_{mn}} \ \left(\sqrt{-g} \ \lag_{G}\right).
\end{equation}

The quantity $\lag_G$ is the volume part of the gravitational Lagrangian. We may write it as:
-5
\begin{equation}\label{C.37}
\left(16 \pi\right) \sqrt{-g} \ \lag_G = \left[\vphantom{a}^3 R + \left(h^{ac}h^{bd}-h^{ab}h^{cd}\right) K_{ab} K_{cd} \right] N \ \sqrt{h}.
\end{equation}

Evaluating the two partial derivatives give us an expression of the canonical momentum in terms of the extrinsic curvature:
\begin{equation}\label{C.38}
\left(16 \pi\right) p^{ab}=\sqrt{h} \ \left(K^{ab} - K h^{ab}\right).
\end{equation}

The volume part of the Hamiltonian is then given by:
\begin{equation}\label{C.39}
\begin{split}
\mathcal{H}&\equiv p^{ab} \dot{h}_{ab} - \sqrt{-g} \ \lag_G\\
&=\left(16 \pi\right)^{-1} \sqrt{h} \left[\left(K^{ab} - K h^{ab}\right)\left(2NK_{ab}+N_{a|b}+N_{b|a}\right)\right.\\
&\phantom{=}\left.-\left(\vphantom{a}^3R + K^{ab}K_{ab}-K^2\right)N\right]\\
&=\left(16 \pi\right)^{-1} \sqrt{h} \left[N \left( K^{ab}K_{ab}-K^2-\vphantom{a}^3R \right)\right.\\
&\phantom{=} \left.+ 2\left(K^{ab}-K h^{ab}\right) N_{a|b}\right].
\end{split}
\end{equation}

The complete Hamiltonian is given by integrating the volume part over the foliation and adding the boundary terms. Integrating by parts the term in $N_{a|b}$ we obtain:
\begin{equation}\label{C.40}
\begin{split}
\left(16\pi\right)H_{G}&=\int_{\sigma_t} \diff^3 y \ \left(16\pi\right)\mathcal{H}_G - 2 \oint_{S_t} \diff^2 \theta \ \left(k-k_0\right) N \ \sqrt{\sigma}\\
&=\int_{\sigma_t} \diff^3 y \ \left[N \left( K^{ab}K_{ab}-K^2-\vphantom{a}^3R \right)- 2\left(K^{ab}-K h^{ab}\right)_{|b} N_{a}\right] \ \sqrt{h}\\
&\phantom{=} - 2 \oint_{S_t} \diff^2 \theta \ \left[\left(k-k_0\right) N-\left(K^{ab}-K h^{ab}\right)N_{a} r_b\right] \ \sqrt{\sigma}.\\
\end{split}
\end{equation}

In the same way we obtained the field equations varying the action functional, if we re-write it in terms of the gravitational Hamiltonian and calculate its variation we obtain a set of canonical equations and two constraints \cite{p.math}. Those constitute the vacuum Einstein field equations in its Hamilton form.

\

If the fields satisfy this vacuum field equations then the gravitational Hamiltonian must be subject to the constraints:
\begin{equation}\label{C.41}
\begin{split}
&\vphantom{a}^3 R+K^2-K^{ab}K_{ab}=0,\\
&\left(K^{ab}-Kh^{ab}\right)_{|b}=0.\\
\end{split}
\end{equation}

Thus, just the boundary terms contributes the gravitational Hamiltonian:
\begin{equation}\label{C.42}
H_{G}^{sol}=- \frac{1}{8 \pi} \oint_{S_t} \diff^2 \theta \ \left[\left(k-k_0\right) N-\left(K^{ab}-K h^{ab}\right)N_{a} r_b\right] \ \sqrt{\sigma}.
\end{equation}

In the next section we will study the value of this solution for asymptotically-flat spactetimes.

\section{ADM mass and angular momentum}

The vacuum state Hamiltonian depends on the asymptotic behaviour of $\Sigma_t$. Choosing our spacetime to be asymptotically-flat, we demand that $\Sigma_t$ asymptotically coincide with a constant time surface in Minkowski spacetime. On this portion of $\Sigma_t$ the coordinates are related to the Minkowski coordinates $x^\alpha \rightarrow \left(\bar{t},\bar{x},\bar{y},\bar{z}\right)$. The vector $u^\alpha=\frac{\partial x^\alpha}{\partial \bar{t}}$ is orthogonal to the surfaces $\bar{t}=const.$ and therefore it must asymptotically coincide with $n^alpha$. This gives us an asymptotic relation for the flow vector:
\begin{equation}\label{C.43}
t^\alpha \rightarrow N \left(\frac{\partial x^\alpha}{\partial \bar{t}}\right)_{y^\alpha}+ N^a \left(\frac{\partial x^\alpha}{\partial y^\alpha}\right)_{\bar{t}}.
\end{equation}

We define the gravitational mass (of an asymptotically-flat spacetime) to be the limit of the solution-valued Hamiltonian when $S_t$ is a sphere at spatial infinity, with a choice of lapse and shift $N=1$ and $N^a=0$:
\begin{equation}\label{C.44}
M_{ADM}=- \frac{1}{8 \pi} \lim_{S_t \rightarrow\infty} \oint_{S_t} \diff^2 \theta \ \left(k-k_0\right)  \ \sqrt{\sigma}.
\end{equation}

The quantity defined by this equation is called \textit{ADM} mass of asymptotically-flat spacetime, after the work of Arnowitt, Deser and Misner. This particular choice of lapse and shift leaves a flow vector $t^\alpha \rightarrow \left(\frac{\partial x^\alpha}{\partial \bar{t}}\right)_{y^\alpha}$, which generates an asymptotic time traslation. This relation conects the total energy of the system and time traslations. In the same mood, it is possible to obtain a relation between energy and asymptotic rotations which leads us to a definition of angular momentum. In order to obtain the generator of rotations, the flow vector must be $t^\alpha \rightarrow \frac{\partial x^\alpha}{\partial \phi}$, where $\phi$ is the rotation angle define in the asymptotic Minkowski frame. This would correspond whit a choice of lapse and shift $N=0$ and $N^a=\frac{\partial y^a}{\partial \phi}$.

\

As we define the gravitational mass, the angular momentum of asymptotically flat spacetime is (minus) the limit of the solution-valued Hamiltonian when $S_t$ is a sphere at spatial infinity with the previous choice of lapse and shift:
\begin{equation}\label{C.45}
J= -\frac{1}{8 \pi} \lim_{S_t \rightarrow\infty} \oint_{S_t} \diff^2 \theta \ \left(K^{ab}-K h^{ab}\right) \ \frac{\partial y^a}{\partial \phi} \ r_b \ \sqrt{\sigma}.
\end{equation}

The extra minus sign is added to recover the right-hand rule for angular momentum.
\chapter{The Gibbons-Hawking-York term and AdS Black Holes}\label{Appendix D}

\section{Review of the GHY term}

The boundary term added to the gravitational action (\ref{1.27}) is known as the Gibbons–Hawking –York boundary term. The term is needed when the underlying spacetime manifold is bounded, so the Einstein-Hilbert variational principle is well defined. This is because the gravitational Lagrangian density contains second derivatives of the metric.

\

When this term is added, the boundary integral that appears when varying the Einstein-Hilbert action vanishes without imposing any restriction on the normal derivatives of the metric. If this is the only role we want this term plays, then we have freedom to adding an arbitrary function $F(g_{\mu \nu},n_\mu,h^{a b}\partial_b)$ because when we take the variation it will vanish under the assumption $\delta g^{\mu \nu}=0$ on $\partial \mathcal{V}$.

\

Before this term was proposed, Einstein constructed the Lagrangian using the object:
\begin{equation}\label{D.1}
H=g^{\mu \nu} \left[\Gamma^{\beta}_{\ \mu \alpha} \Gamma^{\alpha}_{\ \beta \nu} - \Gamma^{\alpha}_{\ \beta \alpha} \Gamma^{\beta}_{\ \mu \nu}\right],
\end{equation}

instead of $R$. With this quantity you get the so-called gamma-gamma Lagrangian of General Relativity, which is first order in the metric and so it is no need of fixing any derivatives of the metric at the boundaries. We can check that it only differs from $R$ by a total derivative. First, we must write $R$ explicitly, in terms of the connection:
\begin{equation}\label{D.2}
R=g^{\mu \nu} \left(\Gamma^{\alpha}_{\ \mu \nu,\alpha} - \Gamma^{\alpha}_{\ \mu \alpha,\nu} + \Gamma^{\alpha}_{\ \beta \alpha} \Gamma^{\beta}_{\ \mu \nu} - \Gamma^{\alpha}_{\ \beta \nu} \Gamma^{\beta}_{\ \mu \alpha}\right).
\end{equation}

Now, we compute the difference:
\begin{equation}\label{D.3}
\begin{split}
R-H&=g^{\mu\nu} \left(\Gamma^{\alpha}_{\ \mu \nu,\alpha} - \Gamma^{\alpha}_{\ \mu \alpha,\nu}\right)\\
&=g^{\mu\nu} \Gamma^{\alpha}_{\ \mu \nu,\alpha} - g^{\mu\nu} \Gamma^{\alpha}_{\ \mu \alpha,\nu}\\
&=g^{\mu\nu} \Gamma^{\alpha}_{\ \mu \nu,\alpha} - g^{\mu\alpha} \Gamma^{\nu}_{\ \mu \nu,\alpha}\\
&=\nabla_{\alpha}\left(g^{\mu\nu} \Gamma^{\alpha}_{\ \mu \nu} - g^{\mu\alpha} \Gamma^{\nu}_{\ \mu \nu}\right).\\
\end{split}
\end{equation}

Thus, this gamma-gamma action differs from the Einstein-Hilbert plus boundary (GHY) term by a surface integral:
\begin{equation}\label{D.4}
\begin{split}
\int_{\mathcal{V}} \diff^4 x \sqrt{-g} \ H&= \int_{\mathcal{V}} \diff^4 x \sqrt{-g} \ R - \int_{\mathcal{V}} \diff^4 x \sqrt{-g} \ \nabla_{\alpha}\left(g^{\mu\nu} \Gamma^{\alpha}_{\ \mu \nu} - g^{\mu\alpha} \Gamma^{\nu}_{\ \mu \nu}\right)\\
&= \int_{\mathcal{V}} \diff^4 x \sqrt{-g} \ R + \int_{\mathcal{V}} \diff^3 y \sqrt{h} \ n^2 n_{\alpha}\left(g^{\mu\nu} \Gamma^{\alpha}_{\ \mu \nu} - g^{\mu\alpha} \Gamma^{\nu}_{\ \mu \nu}\right).\\
\end{split}
\end{equation}

\chapter{Group theory}\label{Appendix E}

\setcounter{theorem}{0}

To a certain time-space transformation of a system we can associate a mathematical transformation such as:
\begin{equation}
\begin{split}
&\ket{\psi}\longrightarrow\mathcal{F}\left(\ket{\psi}\right)=\ket{\psi'}\\
&A\longrightarrow\mathcal{F}\left(A\right)=A'.\\
\end{split}
\end{equation}

If this transformation is a symmetry, we will have:

\begin{enumerate}
	\item The spectrum of observables remains invariant:
	\begin{equation}
	A\ket{\varphi_n}=a_n\ket{\varphi_n} \ \Rightarrow \ A'\ket{\varphi_n'}=a_n\ket{\varphi_n'}.
	\end{equation}
	
	\item We have equivalent probabilities:
	\begin{equation}
	\left|\braket{\varphi}{\psi}\right|=\left|\braket{\varphi'}{\psi'}\right|.
	\end{equation}
	\item The commutators are conserved:
	\begin{equation}
	\comm{A}{B}=\comm{A'}{B'}.
	\end{equation}
\end{enumerate}

\section{Wigner's theorem}

\begin{definition}[Antilinear operator]
	An operator $A$ is said antilinear if $\forall \ \ket{\varphi}, \ket{\psi} \in \mathcal{H}$ and $\alpha,\beta\in\mathbb{C}$ holds:
	\begin{equation}
	A\left(\alpha\ket{\varphi}+\beta\ket{\psi}\right)=\alpha^* A \ket{\varphi}+\beta^* A \ket{\psi}.
	\end{equation}
\end{definition}

\begin{definition}[Antiunitary operator]
	An operator $A$ is said antiunitary if it is antilineal and satisfy:
	\begin{equation}
	AA^{\dagger}=A^{\dagger}A=I.
	\end{equation}
\end{definition}

\begin{theorem}[Wigner's theorem]
	Have an observable $A$ with an orthonormal base of eigenvectors $\left\{\varphi_n\right\}$, and $A'=\mathcal{F}(A)$ with its own orthonormal base of eigenvectors $\left\{\varphi_n'\right\}$. Every transformation which satisfies:
	\begin{equation}
	\left|\braket{\psi}{\varphi}\right|=\left|\braket{\psi'}{\varphi'}\right|,
	\end{equation}
	may be represented with and unitary or antiunitary operator $U$:
	\begin{equation}
	\ket{\psi'}=U\ket{\psi}.
	\end{equation}
\end{theorem}

\begin{corollary}
	The transformed operator $A'$ may be expressed in terms of the same operator $U$:
	\begin{equation}
	A'=U^{-1}AU.
	\end{equation}
\end{corollary}

\section{Groups and transformations}

The families of transformations have group structure. We can distinguish:

\begin{itemize}
	\item Discrete groups. This transformations can be unitary or antiunitary.
	\item Continuous (Lie) groups. Must be unitary. They can be monoparametric or pluriparametric.
\end{itemize}

\begin{definition}[Group]
	Consider a set $\mathcal{G}\neq\varnothing$ and an operation $*$ in it. The pair $(\mathcal{G},*)$ is said to have group structure if:
	
	\begin{enumerate}
		\item $\forall A,B \in \mathcal{G}; \ A*B=C\in\mathcal{G}$.
		\item $\exists E\in\mathcal{G} /\forall A\in\mathcal{G}: \ A*E=E*A=A$.
		\item $\forall A\in\mathcal{G}, \ \exists A'\in\mathcal{G}/ A*A'=A'*A=E$.
		\item $\forall A,B,C\in\mathcal{G}: \ A*(B*C)=(A*B)*C$.
	\end{enumerate}
\end{definition}

If in addition we have that $\forall A,B\in\mathcal{G}: \ A*B=B*A$ then we say that the group is Abelian.

\begin{definition}[Lie group]
	A set with group structure with an infinite number of elements $\left\{G(\alpha)\right\}$ depending on one or more continuous parameters is called a Lie group.
\end{definition}

\begin{definition}[Lie algebra]
	A vector space $\mathrm{A}$ over a field $\mathbb{K}$ together with a bilinear map called the Lie bracket:
	\begin{equation}
	\begin{split}
	\left[\phantom{a},\phantom{a}\right]: \ &\mathrm{A}\times\mathrm{A}\longrightarrow\mathrm{A}\\
	&\left(x,y\right) \longmapsto\comm{x}{y}=z,
	\end{split}
	\end{equation}
	is called a Lie algebra if they both satisfy:
	
	\begin{enumerate}
		\item $\forall x,y,z\in\mathrm{A}$ and $\alpha,\beta\in\mathbb{K}: \ \comm{\alpha x + \beta y}{z}=\alpha\comm{x}{z}+\beta\comm{y}{z}$.
		\item $\forall x,y,z\in\mathrm{A}$ and $\alpha,\beta\in\mathbb{K}: \ \comm{x}{\alpha y + \beta z}=\alpha\comm{x}{y}+\beta\comm{x}{z}$.
		\item $\forall x,y,z\in\mathrm{A}: \ \comm{\comm{x}{y}}{z}+\comm{\comm{y}{z}}{x}+\comm{\comm{z}{x}}{y}=0$.
		\item $\forall x\in\mathrm{A}: \ \comm{x}{x}=0$.
	\end{enumerate}

\end{definition}

\subsection{One-parameter groups}
	
Consider a family of unitary operators $U(s)$ which depends on a single continuous parameter $s$. We can choose the parametrization:
\begin{equation}
\begin{split}
&U(s=0)=I,\\
&U(s_1+s_2)=U(s_1)U(s_2).\\
\end{split}
\end{equation}

\begin{theorem}[Gleason's theorem]
	Every continuous group is also differentiable.
\end{theorem}

\begin{corollary}
	An unitary representation of an element of a continuous group can be Taylor expanded:
	\begin{equation}
	U(s)=I+\left.\frac{\diff U(s)}{\diff s}\right|_{s=0} s + \mathcal{O}(s^2).
	\end{equation}
\end{corollary}

\begin{theorem}[Stone's theorem]
	The family of operators $U(s)$ is unambiguously determined by the generator $K$:
	\begin{equation}
	\left\{\begin{array}{c}\frac{\diff U(x)}{\diff x}=-\imath K U(x)\\U(0)=I\end{array}\right..
	\end{equation}
\end{theorem}

\begin{definition}[su(2) algebra]
	We call su(2) algebra to the Lie algebra generated by three operators $\left\{S_i\right\}$ which obey the commutation relations:
	\begin{equation}
	\comm{S_i}{S_j}=\imath \epsilon_{ijk} S_k.
	\end{equation}
\end{definition}

\begin{theorem}[Fundamental theorem of the su(2) algebra]
	Given a dimensionality $n=2q+1$ with $q=0,1/2,1,...$ exists one and only one set of matrices that obey the commutation relations of the su(2) algebra.
\end{theorem}

\section{Movement constants and degeneration}

\begin{definition}
	Consider a set of lineal and unitary transformations $G=\left\{T_\alpha\right\}_{\alpha\in A}$. The observable $Q$ is said to be invariant under $G$ if:
	\begin{equation}
	\comm{Q}{T_\alpha}=0; \ \forall T_\alpha\in G.
	\end{equation}
\end{definition}

If $Q=H$ then we have $\comm{H}{T_\alpha}=0$ and therefore the $T_\alpha$ are called movement constants.

\begin{theorem}
	If a operator $Q$ is invariant under a set of continuous one-parametric transformations $G=\left\{T(\alpha)\right\}$ then the operator also commutes with its generator $K$:
	\begin{equation}
	\comm{Q}{K}=0.
	\end{equation}
\end{theorem}

If $H$ is invariant under a certain transformation, then its spectrum is degenerated.

\begin{theorem}
	Consider a Hamiltonian $H$ and a transformation $T$ such as $\comm{H}{T}=0$. If $H\ket{\alpha}=E_{\alpha}\ket{\alpha}$, then:
	
	\begin{itemize}
		\item $\ket{\beta}=T\ket{\alpha}$ is an eigenvector of $H$.
		\item if $\ket{\alpha}$ and $\ket{\beta}$ are lineally independent, the eigenvalue $E_\alpha$ is degenerated.
	\end{itemize}
\end{theorem}
\chapter{The Lorentz Group}\label{Appendix F}

The postulates of Special Relativity tells that the velocity of light \textit{c} is the same in all inertial frames. If in one frame we have a light signal at space-time point $(t,x_i)$ and in another frame we found it at $(t',x'_i)$, the previous restriction implies:
\begin{equation}\label{F.1}
s^2\equiv c^2 t^2 - x_i x_i=c^2 t'^2 - x'_i x'_i.
\end{equation}

Choosing units such that $c=1$ and adopting the contracted notation $x^\mu$, $\mu=0,1,2,3$ with $x^0=t$ and $(x^1,x^2,x^3)=\vec{x}$, $s^2$ may be written as:
\begin{equation}\label{F.2}
s^2\equiv x^\mu x^\nu g_{\mu \nu}=x'^\mu x'^\nu g_{\mu \nu},
\end{equation}

where the metric $g_{\mu \nu}$ is zero except for $\mu=\nu$ when $g_{00}=-g_{11}=-g_{22}=-g_{33}=1$

\

Now we look for a linear transformation which preserves $s^2$:
\begin{equation}\label{F.3}
x'^{\mu}=\Lambda^{\mu}_{\nu} x^{\nu}.
\end{equation}

It must satisfy:
\begin{equation}\label{F.4}
x'^\mu x'^\nu g_{\mu \nu}=\Lambda^{\mu}_{\rho} x^{\rho} \Lambda^{\nu}_{\sigma} x^{\sigma} g_{\mu \nu}=x^{\rho} x^{\sigma} g_{\rho \sigma}.
\end{equation}

Because that must be hold by every $x^{\mu}$, we conclude:
\begin{equation}\label{F.5}
g_{\rho \sigma}=g_{\mu \nu} \Lambda^{\mu}_{\rho} \Lambda^{\nu}_{\sigma}.
\end{equation}

One may show that these transformations satisfy the group axioms.

\

If we use matrix notation, regarding $x^{\mu}$ as a column vector $x$, the metric as a squared matrix $g$ and the transformation as the matrix equivalent off the coefficients $L$, we can rewrite the previous relation as:
\begin{equation}\label{F.6}
g=L^tgL.
\end{equation}

We must point out two consequences of the above discussion. First, if we take the determinant in the last equation:
\begin{equation}\label{F.7}
\det{g}=\det{L^t} \det{g} \det{L},
\end{equation}

from which we deduce:
\begin{equation}\label{F.8}
\det{L}=\pm 1.
\end{equation}

Second, taking the $00$ entry in equation (5):
\begin{equation}\label{F.9}
1=g_{\mu \nu} \Lambda^{\mu}_{0} \Lambda^{\nu}_{0}=\left(\Lambda^{0}_{0}\right)^2-\left(\Lambda^{i}_{i}\right)^2,
\end{equation}
shows that:
\begin{equation}\label{F.10}
\left|\Lambda^{0}_{0}\right| \geq 1.
\end{equation}

Any Lorentz transformation can be decomposed as the product of rotations, boost transformations, time inversion and full inversion. Consider the infinitesimal Lorentz Transformation:
\begin{equation}\label{F.11}
\Lambda^{\mu}_{\nu}=\delta^{\mu}_{\nu} + \varepsilon^{\mu}_{\nu}.
\end{equation}

Evaluating (\ref{F.5}) to first order in $\epsilon$ gives:
\begin{equation}\label{F.12}
g_{\rho \sigma}=g_{\rho \sigma} + g_{\rho \nu} \varepsilon^{\nu}_{\sigma} + g_{\mu \sigma} \varepsilon^{\mu}_{\rho},
\end{equation}
which is equivalnet to:
\begin{equation}\label{F.13}
g_{\nu \rho} \varepsilon^{\rho}_{\mu} + g_{\mu \rho} \varepsilon^{\rho}_{\nu}=0.
\end{equation}

Since the metric can low indices, that equation becomes:
\begin{equation}\label{F.14}
\varepsilon_{\nu \mu} + \varepsilon_{\mu \nu}=0
\end{equation}

That is, $\varepsilon_{\mu \nu}$ is an antisymmetric tensor with only 6 independent entries. It could be advertised, since there's three boosts and three rotations, one for each space direction.

\

We can introduce the Hermitian generators:
\begin{equation}\label{F.15}
L_{\mu \nu}\equiv\imath \left(x_{\mu} \partial_{\nu}-x_{\nu} \partial_{\mu}\right).
\end{equation}

The $L_{\mu \nu}$'s satisfy the Lie algebra of $SO(3,1)$:
\begin{equation}\label{F.16}
\left[L_{\mu \nu},L_{\rho \sigma}\right] = \imath g_{\nu \rho} L_{\mu \sigma} - \imath g_{\mu \rho} L_{\nu \sigma} - \imath g_{\nu \sigma} L_{\mu \rho} + \imath g_{\mu \sigma} L_{\nu \rho}.
\end{equation}

The most general representation of the generators of $SO(3,1)$ obeying that commutation relation is given by:
\begin{equation}\label{F.17}
M_{\mu \nu}\equiv \imath \left(x_{\mu} \partial_{\nu}-x_{\nu} \partial_{\mu}\right) + S_{\mu \nu},
\end{equation}
where the Hermitian operator $S_{\mu \nu}$ commute with $L_{\mu \nu}$ and satisfy the same Lie algebra. These generators form an algebra among themselves:
\begin{equation}\label{F.18}
\left[M_{i j},M_{k l}\right] = -\imath \delta_{j k} M_{i l} + \imath \delta_{i k} M_{j l} + \imath \delta_{j l} M_{i k} - \imath \delta_{i l} M_{j k},
\end{equation}
which is that of the rotation group $SU(2)$. If we introduce the operators:
\begin{equation}\label{F.19}
J_i\equiv\frac{1}{2} \ \epsilon_{ijk} M_{jk},
\end{equation}
\begin{equation}\label{F.20}
K_i\equiv M_{0i},
\end{equation}
where $\epsilon_{ijk}$ is the Levi-Civita symbol, the next commutation relations follows:
\begin{equation}\label{F.21}
\left[J_i,J_j\right]=\imath \epsilon_{ijk} J_k,
\end{equation}
\begin{equation}\label{F.22}
\left[K_i,K_j\right]=-\imath \epsilon_{ijk} J_k,
\end{equation}
\begin{equation}\label{F.23}
\left[J_i,K_j\right]=\imath \epsilon_{ijk} K_k.
\end{equation}

These relations may be disentangled by introducing the linear combination:
\begin{equation}\label{F.24}
N_i\equiv \frac{1}{2} \ \left( J_i + \imath K_i\right).
\end{equation}

Although it's not Hermitian, it yield simple commutation relations:
\begin{equation}\label{F.25}
\left[N_i,N_j^\dagger\right]=0,
\end{equation}
\begin{equation}\label{F.26}
\left[N_i,N_j\right]=\imath \epsilon_{ijk} N_k,
\end{equation}
\begin{equation}\label{F.27}
\left[N_i^\dagger,N_j^\dagger\right]=\imath \epsilon_{ijk} N_k^\dagger.
\end{equation}
\end{appendices}

\end{document}